\documentclass{article}

\usepackage[psamsfonts]{amssymb}
\usepackage{amsmath}
\usepackage{verbatim}
\usepackage{amsthm}
\usepackage{hyperref}

\newtheorem{theorem}{Theorem}[section]
\newtheorem{proposition}[theorem]{Proposition}
\newtheorem{lemma}[theorem]{Lemma}
\newtheorem{corollary}[theorem]{Corollary}

\newtheorem{remark}[theorem]{Remark}
\numberwithin{equation}{section}

\def\func#1{\mathop{\rm {}#1}\nolimits}

\begin{document}

\title{ \textbf{The Poincar\'e-Cartan Form \\
in Superfield Theory} }
\author{ \textsc{Juan Monterde} \\
{\small Departament de Geometria i Topologia}\\
{\small Universitat de Val\`encia}\\
{\small Av.\ V. A. Estell\'es \indent 1, 46100-Burjassot, Spain}\\
{\small \emph{E-mail:\/}\texttt{juan.l.monterde@uv.es}} \medskip \\
\textsc{Jaime Mu\~{n}oz Masqu\'e}\\
{\small Instituto de F\'{\i}sica Aplicada, CSIC }\\
{\small C/ Serrano 144, 28006-Madrid, Spain }\\
{\small \emph{E-mail:\/}\texttt{jaime@iec.csic.es}} \medskip \\
\textsc{Jos\'e A. Vallejo}\\
{\small Departament de Matem\'atica Aplicada IV}\\
{\small Universitat Polit\'ecnica de Catalunya}\\
{\small Av.\ Canal Ol\'impic s/n, 08860-Castelldefels, Spain }\\
{\small \emph{E-email:\/}\texttt{jvallejo@ma4.upc.edu}}\\
{\small and}\\
{\small Facultad de Ciencias}\\
{\small Universidad Aut\'onoma de San Luis Potos\'i}\\
{\small Lateral Av.\ Salvador Nava s/n, 78220-SLP, M\'exico }\\
}
\date{}
\maketitle

\newpage

\tableofcontents

\newpage

\begin{abstract}
\noindent An intrinsic description of the Hamilton$-$Cartan
formalism for first-order Berezinian variational problems
determined by a submersion of supermanifolds is given.
This is achieved by studying the associated higher-order
graded variational problem through the Poincar\'e$-$Cartan
form. Noether theorem and examples from superfield theory
and supermechanics are also discussed.
\end{abstract}

\bigskip


\medskip


\medskip


\newpage

\section{Introduction}

In this paper, we generalize some of the results already presented in \cite%
{Mon-Mun 02, Mon-Val 03}, where supermechanics (that is, variational
problems defined for supercurves $\sigma \colon {\mathbb{R}}^{1|1}\to {%
\mathbb{R}}^{1|1} \times (M,\mathcal{A})$ with $(M,\mathcal{A})$ a
supermanifold and ${\mathbb{R}}^{1|1}$ the parameter superspace), is
considered from the viewpoint of Poin\-car\'e$-$Cartan theory. Now, we
intend to deal with superfield theory; that is, with first order variational
problems defined for superfields $\sigma \colon (M,\mathcal{A})\to (N,%
\mathcal{B})$ (here $(M,\mathcal{A}),(N,\mathcal{B})$ are supermanifolds).

The basic object in our study is the Poincar\'e$-$Cartan form, for which we
present an intrinsic construction in the context of Berezinian variational
problems (intrinsic up to a volume form on the base manifold, as we will
see).

Let us recall that there are two kind of integration theories defined on
supermanifolds: the one associated to the Berezin integral and the other
associated to what is called the graded integral. The first one is more
suitable to state physical problems in the supermanifold setting, but it
lacks from an associated theory of Berezinian superdifferential forms. So,
it is not possible to work directly with a Poincar\'{e}$-$Cartan form and to
develop a Hamilton$-$Cartan formalism from it.

The second theory of integration does not have a good physical
interpretation but, conversely, a consistent theory of differential forms is
available and therefore, it is possible to define a Poincar\'e$-$Cartan form
and to develop the corresponding Hamilton$-$Cartan formalism.

Accordingly to these two possibilities, variational problems can be stated
using either the Berezin integral or the graded integral; we call them
Berezinian or graded variational problems, respectively. However, there is a
deep connection between both problems. In brief, the relationship is based
on the fact that to each first$-$order Berezinian variational problem over a
graded submersion $p\colon (N,\mathcal{B})\rightarrow (M,\mathcal{A})$ we
can associate a graded variational problem of order $n+1$ over $p$, where $%
(m|n)$ is the dimension of $(M,\mathcal{A})$ (see Section \ref{comparison}
below); we refer the reader to Theorem \ref{comparisonth} for formal
definitions and statement of this result, known as the Comparison Theorem.

With the help of the Comparison Theorem the way to build a Poincar\'e$-$%
Car\-tan form and to develop a Hamilton$-$Cartan formalism for a \emph{first$%
-$order} Berezinian variational problem is clear: Firstly, we define the
graded Poincar\'{e}$-$Cartan form for the associated graded variational
problem, \emph{now of order} $n+1$, and secondly we translate, with the hint
offered by the Comparison Theorem, this form to an object which will play
the role of Berezinian Poincar\'{e}$-$Cartan form for the Berezinian
variational problem. From this object it is possible to obtain the Euler$-$%
Lagrange superequations and a Noether Theorem.

A question arises at this point. In the classical case it is well known that
a canonical Poincar\'e$-$Cartan form of higher order does \emph{not} exist.
Of course, objects which can be called higher-order Poincar\'e$-$Cartan
forms can be defined, but the problem is that they depend on some additional
parameters (such as a connection, see \cite{Fer-Fra 87, Gar-Mun 83}).
Nevertheless, here we give a \emph{canonical} formulation of the graded
Poincar\'e$-$Cartan form for higher-order graded variational problems; the
key to understand how this is achieved is to note that we deal with a
special subclass of these problems: those coming from first$-$order
Berezinian variational ones through the Comparison Theorem. Actually, our
purpose is to solve these first order Berezinian problems, so we could
consider this feature as a byproduct.

Another very important consequence of this formalism in the classical case,
is the existence of a Noether Theorem, which is a basic tool in the study of
the symmetries of a variational problem. We present here a generalization to
the graded setting.

In order to make the paper relatively self$-$contained, the first sections
contain a review of previous results on jet bundles and calculus of
variations on supermanifolds.

Finally, there are some worked out examples (the $(m|2)$ field theory) and
we analyze a particular case of interest in Physics (supermechanics) showing
the coincidence with the results obtained by other methods (\cite{Mon-Mun
02, Mon-Val 03}).

\section{Basics of supermanifold theory\label{previous}}

\subsection{General definitions}

For general references, we refer the reader to \cite{Var 04}, \cite[Chapters
2 and 3]{Del-Morg 99}, \cite{Kos 77}, \cite{Lei 80}, \cite{Bar-Bru-Her 91}, 
\cite{Man 88}, and \cite{Wess-Bag 92}. The basic idea underlying the
definition of a graded manifold is the substitution of the commutative sheaf
of algebras of differentiable functions on a smooth manifold by another
sheaf in which we can accommodate some objects with a $\mathbb{Z}_{2}-$%
grading (in what follows, all the gradings considered are assumed to be $%
\mathbb{Z}_{2}-$gradings, unless otherwise explicitly stated.)

A graded manifold (or a supermanifold) of dimension $(m|n)$ on a $C^{\infty}-
$mani\-fold $M$ of dimension $m$, is a sheaf $\mathcal{A}$ on $M$ of graded $%
\mathbb{R}-$commutative algebras---the structure sheaf---such that,

\begin{enumerate}
\item There exists an exact sequence of sheaves, 
\begin{equation}
0\to \mathcal{N}\rightarrow \mathcal{A}\overset{\sim }{\rightarrow }
C^{\infty }(M)\to 0,  \label{eq1_5}
\end{equation}
where $\mathcal{N}$ is the sheaf of nilpotents in $\mathcal{A}$ and $\sim $
is a surjective morphism of graded $\mathbb{R}-$commutative algebras.

\item \label{Cond2} $\mathcal{N}/\mathcal{N}^{2}$ is a locally free module
of rank $n$ over $C^{\infty }(M)=\mathcal{A}/\mathcal{N}$, and $\mathcal{A}$
is locally isomorphic, as a sheaf of graded $\mathbb{R}-$commutative
algebras, to the exterior bundle $\bigwedge_{C^{\infty }(M)}(\mathcal{N}/%
\mathcal{N}^{2})$.
\end{enumerate}

For any open subset $U\subset M$, from the exact sequence \eqref{eq1_5} we
obtain the exact sequence of graded algebras, 
\begin{equation*}
0\to \mathcal{N}(U) \to \mathcal{A}(U) \overset{\sim }{\rightarrow }
C^\infty (U)\to 0.
\end{equation*}
A section $f$ of $\mathcal{A}$ is called a graded function (or a
superfunction). The image of such a graded function $f\in \mathcal{A}(U)$ by
the structure morphism $\sim $ is denoted by $\tilde{f}$.

The fact that $\mathcal{A}$ is a sheaf of graded $\mathbb{R}-$commutative
algebras induces a grading on its sections, and we denote the degree of such
an $f$ by $\left\vert f\right\vert $.

From the very definition of a supermanifold the structure sheaf of $(M,%
\mathcal{A})$ is locally isomorphic to $\bigwedge_{C^{\infty }(M)}(\mathcal{N%
}/\mathcal{N}^{2})$. An important theorem (known as Batchelor Theorem \cite%
{Bat 79, Bat 80}, but also see \cite{Gaw 77}), guarantees that in the $%
C^{\infty }$ category this holds not only locally, but also globally,
although this is no longer true in the complex analytic category. Thus, for
any smooth supermanifold $(M,\mathcal{A})$ there exists a vector bundle $%
E\to M$ which is isomorphic to $\mathcal{N}/\mathcal{N}^{2}$ and such that $%
\mathcal{A}\cong \bigwedge_{C^{\infty }(M)}(E)$, but this isomorphism is not
canonical.

A splitting neighborhood of a supermanifold $(M,\mathcal{A})$ is an open
subset $U$ in $M$ such that the bundle $E=\mathcal{N}/\mathcal{N}^{2}$ is
trivial over $U$ and 
\begin{equation*}
\mathcal{A}|_{U}\cong \bigwedge\nolimits_{C^{\infty }(U)}(E|_{U}).
\end{equation*}
If $U$ is a splitting neighborhood, there exists a basis of sections for $%
E|_{U}$, denoted by $(x^{-1},\dotsc ,x^{-n})$, along with an isomorphism 
\begin{equation}
\mathcal{A}(U)\cong C^{\infty }(U)\otimes _{\mathbb{R}}\bigwedge E_{n},
\label{eq1_6}
\end{equation}
where $E_{n}$ denotes the vector $\mathbb{R}-$space generated by $%
(x^{-1},\dotsc ,x^{-n})$. Therefore, the natural projection $\mathcal{A}%
(U)\rightarrow C^{\infty }(U)$, $f\mapsto \tilde{f}$, admits a global
section of $\mathbb{R}-$algebras, $\sigma \colon C^{\infty
}(U)\hookrightarrow \mathcal{A}(U)$. If $U$ is a splitting neighborhood, a
family of superfunctions $(x^{i},x^{-j})$, $1\leq i\leq m$, $1\leq j\leq n$, 
$|x^{i}|=0$, $|x^{-j}|=1$, is called a graded coordinate system (or a
supercoordinate system) if,

\begin{enumerate}
\item $x^i=\sigma (\tilde{x}^i)$, $1\leq i\leq m$, where $(\tilde{x}%
^1,\dotsc,\tilde{x}^m)$ is an ordinary coordinate system on $U$,

\item $\{x^{-1},\dotsc,x^{-n}\}$ is a basis of sections of $E|_U$; i.e., $%
x^{-1},\dotsc,x^{-n} \in \bigwedge E_n$ and $\prod \limits _{j=1}^n
x^{-j}\neq 0$.
\end{enumerate}

A morphism of graded manifolds $\phi \colon(M,\mathcal{A}) \to (N,\mathcal{B}%
)$ is a pair of mappings $(\tilde{\phi },\phi^\ast )$ where $\tilde{\phi }%
\colon M\to N$ is a differentiable mapping of smooth manifolds and for every
open subset $U\subset N$, $\phi ^\ast \colon \mathcal{B}(U) \to (\tilde{\phi 
}_\ast \mathcal{A})(U) =\mathcal{A} (\tilde{\phi }^{-1}(U))$ is an even
morphism of graded algebras compatible with the restrictions, and all such
that the diagram 
\begin{equation*}
\begin{array}{ccc}
\mathcal{B}(U) & -\!\!\!\overset{\phi ^\ast }{\longrightarrow } & \mathcal{A}
(\tilde{\phi }^{-1}(U)) \\ 
\begin{array}{c}
|\vspace{-0.1cm} \\ 
\downarrow%
\end{array}
&  & 
\begin{array}{c}
|\vspace{-0.1cm} \\ 
\downarrow%
\end{array}
\\ 
C^\infty (U) & \underset{\tilde{\phi }^\ast }{-\!\!\!\longrightarrow } & 
C^\infty (\tilde{\phi }^{-1}(U))%
\end{array}%
\end{equation*}
commutes.

Throughout this paper, we assume that $M$ is connected and oriented by a
volume form $\eta$. We confine ourselves to consider coordinate systems
adapted to this volume form; i.e., 
\begin{equation*}
\eta =d\tilde{x}^1\wedge \cdots \wedge d\tilde{x}^m.
\end{equation*}
We refer all our constructions to this volume, but we simply call
``intrinsic constructions'' those results which are independent of $\eta $,
in order to avoid continuous mention to $\eta $. Note that, by Batchelor's
theorem (see \cite{Bat 79}), the natural projection $\mathcal{A}(M)\to
C^\infty (M)$ admits a global section $\sigma \colon C^\infty (M) \to 
\mathcal{A}(M)$. Thus, once a section $\sigma $ has been fixed, every
ordinary volume form $\eta $ on $M$ induces a graded volume $\eta ^G$ on $(M,%
\mathcal{A})$.

Let $\mathcal{F},\mathcal{G}$ be sheaves on a topological space $X$. For any
open subset $U\subset M$,$\mathop{\rm {}Hom}\nolimits (\mathcal{F}|_U,%
\mathcal{G}|_U)$ denotes the space of morphisms between the sheaves $%
\mathcal{F}|_U$ and $\mathcal{G}|_U$; this is an abelian group in a natural
way. The sheaf of homomorphisms is the sheaf $\mathop{\rm {}Hom}\nolimits (%
\mathcal{F},\mathcal{G})$ given by $\mathop{\rm {}Hom}\nolimits (\mathcal{F}%
, \mathcal{G})(U) =\mathop{\rm {}Hom}\nolimits (\mathcal{F}|_U, \mathcal{G}%
|_U)$ with the natural restriction morphisms.

The sheaf of left $\mathcal{A}-$modules of derivations of a graded manifold $%
(M,\mathcal{A})$ is the subsheaf of $\mathop{\rm {}End}\nolimits_{\mathbb{R}%
}(\mathcal{A})$ whose sections on an open subset $U\subseteq M$ are $\mathbb{%
R}-$linear graded derivations $D\colon \mathcal{A}|_{U}\rightarrow \mathcal{A%
}|_{U}$. This sheaf is denoted by $\mathop{\rm {}Der}\nolimits_{\mathbb{R}}(%
\mathcal{A})$ or simply $\mathop{\rm {}Der}\nolimits(\mathcal{A})$, and its
elements are called graded vector fields (or supervector fields) on the
graded manifold $(M,\mathcal{A})$. The notation $\mathcal{X}_{G}(M)$ is also
often used.

Let $U$ be a coordinate neighborhood for a graded manifold $(M,\mathcal{A})$
with graded coordinates $(x^{i},x^{-j})$, $1\leq i\leq m$, $1\leq j\leq n$.
There exist even derivations $\partial /\partial x^{1},\dotsc ,\partial
/\partial x^{m}$ and odd derivations $\partial /\partial
x^{-1},\dotsc,\partial /\partial x^{-m}$ of $\mathcal{A}(U)$ uniquely
characterized by the conditions 
\begin{equation*}
\frac{\partial x^{j}}{\partial x^{i}}=\delta _{i}^{j}, \quad \frac{\partial
x^{-j}}{\partial x^{i}}=0, \quad \frac{\partial x^{j}}{\partial x^{-i}} =0,
\quad \frac{\partial x^{-j}}{\partial x^{-i}} =\delta _{i}^{j}
\end{equation*}
(negative indices running from $-n$ to $-1$, positive ones from $1$ to $m$)
and such that every derivation $D\in \mathop{\rm {}Der}\nolimits\mathcal{A}%
(U)$ can be written as 
\begin{equation*}
D=\sum\limits_{i=1}^{m}D(x^{i})\frac{\partial }{\partial x^{i}}
+\sum\limits_{j=1}^{m}D(x^{-j})\frac{\partial }{\partial x^{-j}}.
\end{equation*}
In particular, $\mathop{\rm {}Der}\nolimits(\mathcal{A}(U))$ is a free right 
$\mathcal{A}(U)-$module with basis 
\begin{equation*}
\frac{\partial }{\partial x^{1}},\dotsc ,\frac{\partial }{\partial x^{m}}; 
\frac{\partial }{\partial x^{-1}},\dotsc ,\frac{\partial }{\partial x^{-m}}.
\end{equation*}
If $U\subseteq M$ is an open subset, the algebraic dual of the graded $%
\mathcal{A}-$module $\mathop{\rm {}Der}\nolimits(\mathcal{A}(U))$ is $(%
\mathop{\rm {}Der}\nolimits\mathcal{A}(U))^{\ast } =\mathop{\rm {}Hom}%
\nolimits_{\mathcal{A}}(\mathop{\rm {}Der}\nolimits (\mathcal{A}(U)),%
\mathcal{A}(U))$, which has itself a natural structure of graded $\mathcal{A}%
-$module and it defines a sheaf $U\mapsto (\mathop{\rm {}Der}\nolimits%
\mathcal{A}(U))^{\ast }$.

The sheaves of right $\mathcal{A}-$modules of graded differential forms on $%
(M,\mathcal{A})$ are the sheaves 
\begin{equation*}
\Omega _{G}^{p}(M) =\bigwedge^{p}(\mathop{\rm {}Der}\nolimits\mathcal{A}%
)^{\ast }.
\end{equation*}
We also set $\Omega _{G}(M) =\sum\limits_{p\in \mathbb{N}}\Omega _{G}^{p}(M)$%
, with $\Omega _{G}^{0}(M)=\mathcal{A}$.

The graded differential forms on $(M,\mathcal{A})$ are simply called graded
forms. The three usual operators: insertion of a graded vector field, graded
Lie derivative with respect to a graded vector field and the graded exterior
differential, are defined in a similar way to the classical case (e.g., see 
\cite{Kos 77}), and denoted by $\iota _X$, $\mathcal{L}_X^G$, and $d^{G}$,
respectively.

\subsection{Supervector bundles}

Let $\mathrm{GL}(V)$ be the general linear supergroup of a supervector space 
$V=V_0\oplus V_1$. We set $\mathrm{GL}(p|q) =\mathrm{GL}(\mathbb{R}^{p|q})$.
For the definition of the graded structure of $\mathrm{GL}(p|q)$ as a super
Lie group, we refer the reader to \cite[I, \S 3]{Bar-Bru-Her 91}, \cite[%
Chapter 2, \S 1]{Ber 87}, \cite[\S 1.5]{Boy-San 91}, \cite[\S 2.11]{Del-Morg
99}, \cite{Dubois 01}, \cite[Chapter 4, \S 10]{Man 88}, \cite[\S 2.14]{San
88}, \cite[\S 2]{San 88a}, \cite[\S 4.19]{Sch 84}, and \cite[\S 2.2.1]{Vor
91}.

Let $(M,\mathcal{A})$ be an $(m|n)$-dimensional supermanifold. As is well
known (e.g., see \cite[\S 3.2]{Del-Morg 99}, \cite[7.10]{Sch 84}), a
supervector bundle of rank $(p|q)$ over $(M,\mathcal{A})$ can be described
either (i) as a fibre bundle $V$ over $M$ with typical fibre $\mathbb{R}%
^{p|q}$ and structure group $\mathrm{GL}(p|q)$, or (ii) as a locally free
sheaf of $\mathcal{A}-$modules $\mathcal{V}$ of rank $(p|q)$. The
description in (ii) means that every point $x\in M$ admits an open
neigbourhood $U\subseteq M$ such that $\mathcal{V}|_{U}$ is isomorphic---as
a sheaf of $\mathcal{A}|_{U}-$modules---to $\mathcal{A}^{p|q}|_{U}=\mathcal{A%
}^{p}|_{U}\oplus \Pi \mathcal{A}^{q}|_{U}$ (direct sum of $p$ copies of $%
\mathcal{A}$ and $q$ copies of $\Pi \mathcal{A}$), where $\Pi $ denotes the
functor of change of parity; precisely, for every open subset $O\subseteq U$
we have $\mathcal{V}(O)\cong \mathcal{A}^{p}(O)\oplus \Pi \mathcal{A}^{q}(O)$%
.

More formally, we can state (see \cite[2.11 Theorem]{San 86}): There is a
one-to-one (functorial) correspondence between the set of isomorphism
classes of locally free sheaves of (left) graded $\mathcal{A}-$modules of
rank $(p|q)$ over $M$ and the set of isomorphisms classes of supervector
bundles of rank $(p|q)$ over the graded manifold $(M,\mathcal{A})$. Also see 
\cite[7.10. Theorem]{Sch 84} for a slightly different approach.

We remark that the tangent and cotangent `supervector bundles' introduced in 
\cite{Kos 77} are not supervector bundles in the previous sense, as they are
not locally trivial. Because of this, we prefer to work with the
supertangent bundle $\mathcal{ST}(M,\mathcal{A})$ of $(M,\mathcal{A})$
introduced by S\'{a}nchez-Valenzuela, which corresponds to the locally free
sheaf of $\mathcal{A}-$modules of derivations, $\mathop{\rm {}Der}\nolimits%
\mathcal{A}$. For our purposes, another important reason to do this, is that
the graded manifold of $1-$jets of graded curves from $\mathbb{R}^{1|1}$ to
a graded manifold $(M,\mathcal{A})$ is isomorphic to $\mathcal{ST}(M,%
\mathcal{A})$; i.e., $J_{G}^{1}(p)\simeq \mathcal{ST}(M,\mathcal{A})$, where 
$J_{G}^{1}(p)$ is the graded manifold of graded $1-$jets of sections of the
natural projection onto the first factor, $p\colon \mathbb{R}^{1|1}\times (M,%
\mathcal{A})\to \mathbb{R}^{1|1}$.

Let $\pi \colon (E,\mathcal{E})\rightarrow (M,\mathcal{A})$ be a supervector
bundle. For any $x\in M$, we denote by $\pi ^{-1}(x)$ the superfibre over $x$%
, i.e., the supermanifold whose underlying topological space is $\tilde{\pi}%
^{-1}(x)$ and whose structure sheaf is 
\begin{equation*}
\mathcal{A}_{x} =(\mathcal{E}\diagup \mathcal{K}_{x})|_{\tilde{\pi}^{-1}(x)},
\end{equation*}
where $\mathcal{K}_{x}$ is the subsheaf of $\mathcal{E}$ whose sections
vanish when restricted to $\tilde{\pi}^{-1}(x)$.

For any $x\in M$, $\pi ^{-1}(x)$ is isomorphic with the standard fibre of $%
\pi $.

A supervector bundle morphism from the vector bundle $\pi _{E}\colon (E,%
\mathcal{E})\rightarrow (M,\mathcal{A})$ to the vector bundle $\pi
_{F}\colon (F,\mathcal{F})\rightarrow (M,\mathcal{A})$ is a supermanifold
morphism 
\begin{equation*}
H\colon (E,\mathcal{E})\rightarrow (F,\mathcal{F})
\end{equation*}
such that $\pi _{F}\circ H=\pi _{E}$ the restriction of which to each
superfibre $\pi _{E}^{-1}(x)$ is superlinear. The following consequence can
be proved:

\begin{proposition}[Proposition 3.3 in \protect\cite{San 86}]
Let $(M,\mathcal{A})$ be a graded manifold, let $\mathcal{K},\mathcal{L}$ be
two locally free sheaves of graded $\mathcal{A}-$modules of ranks $(p|q)$
and $(r|s)$, respectively, and let $\pi _{E}\colon (E,\mathcal{E}%
)\rightarrow (M,\mathcal{A})$, $\pi _{F}\colon (F,\mathcal{F})\rightarrow (M,%
\mathcal{A})$ be the supervector bundles that $\mathcal{K}$ and $\mathcal{L}$
give rise to, respectively. Each morphism $\psi \colon \mathcal{K}%
\rightarrow \mathcal{L}$ of sheaves of graded $\mathcal{A}-$modules over $M$
defines a morphism 
\begin{equation*}
H_{\psi }\colon (E,\mathcal{E})\rightarrow (F,\mathcal{F})
\end{equation*}
such that $\pi _{F}\circ H_{\psi }=\pi _{E}$ and it restricts to a
superlinear morphism over each fibre.
\end{proposition}

Another construction which we will use is the pull-back (or inverse image)
of a supervector bundle along a graded submersion, which is a particular
case of the pull-back of modules over ringed spaces. For our purposes, it
suffices the following description.

Let $p\colon (N,\mathcal{B})\rightarrow (M,\mathcal{A})$ be a graded
submersion, and let $\mathcal{K}$ be a sheaf of graded $\mathcal{A}-$modules
over $M$ with projection $\pi $. The pull-back $p^{\ast }\mathcal{K}$ is the
sheaf of $p^{\ast }\mathcal{A}-$modules over $N$ where to each open $%
V\subset N$, it corresponds 
\begin{equation*}
p^{\ast }\mathcal{K}(V) =\{(k,y)\in \mathcal{K}(\tilde{p}(V))\times V:\pi
(k)=\tilde{p}(y)\}.
\end{equation*}
It is customary to write $p^{\ast }\mathcal{K=K\times }_{(M,\mathcal{A})}(N,%
\mathcal{B})$. Note that if we consider the supervector bundle on $(M,%
\mathcal{A})$ given by $\mathcal{K}$, then $p^{\ast }\mathcal{K}$ gives a
supervector bundle on $(N,\mathcal{B})$.

\section{Graded jet bundles\label{jets}}

\subsection{Notations and definitions \label{notdef}}

For the details of the construction of graded jet bundles associated to a
graded submersion $p\colon (N,\mathcal{B})\rightarrow (M,\mathcal{A})$, we
refer the reader to \cite{Her-Mun 82, Her-Mun 84a, Her-Mun 87, Mon 92a,
Mon-Mun 92a}. We also note that other approaches to superjet bundles of
interest in Physics are possible, see \cite{GMS 05}.

We denote by 
\begin{equation*}
p_{k}\colon \left( J_{G}^{k}(p),\mathcal{A}_{J_{G}^{k}(p)}\right) \to (M,%
\mathcal{A})
\end{equation*}
the graded $k-$jet bundle of local sections of $p$, with natural projections 
\begin{equation*}
p_{kl}\colon \left( J_{G}^{k}(p),\mathcal{A}_{J_{G}^{k}(p)}\right)
\rightarrow \left( J_{G}^{l}(p),\mathcal{A}_{J_{G}^{l}(p)}\right) , \quad
k\geq l.
\end{equation*}

\begin{remark}
Sometimes we will write $p_{k,l}$ in order to avoid confusions, as in the
case of the projection $p_{k,k-1}$ (of $J_{G}^{k}(p)$ onto $J_{G}^{k-1}(p))$
and even we will employ $p_{l}^{k}$ indistinctly.
\end{remark}

Each section $\sigma \colon (M,\mathcal{A})\rightarrow (N,\mathcal{B})$ of
the graded submersion $p$ induces a closed embedding of graded manifolds 
\begin{equation*}
j^{k}\sigma \colon (M,\mathcal{A}) \to \left( J_{G}^{k}(p),\mathcal{A}%
_{J_{G}^{k}(p)}\right) ,
\end{equation*}
which is called the graded $k-$jet extension of $\sigma $.

We set $(m|n)=\dim (M,\mathcal{A})$, $(m+r|n+s)=\dim (N,\mathcal{B})$, and
let 
\begin{equation}
\left. 
\begin{array}{ll}
\left( x^{\alpha }\right) , & \alpha =-n,\dotsc ,-1,1,\dotsc ,m, \\ 
\left( y^{\mu }\right) , & \mu =-s,\dotsc ,-1,1,\dotsc ,r.%
\end{array}
\right\}  \label{coordinates1}
\end{equation}
be a fibred coordinate system for the submersion $p\colon (N,\mathcal{B}%
)\rightarrow (M,\mathcal{A})$, defined over an open domain $V\subseteq N$.
This means that the graded functions $(x^{\alpha })$, $i=-n,\dotsc,-1,1,%
\dotsc ,m$, belong to $p^{\ast }\mathcal{A}(U)$, where $U=\tilde{p}(V)$.

The system \eqref{coordinates1} induces a coordinate system for $%
J_{G}^{k}(p) $ on $(\tilde{p}_{k0})^{-1}(V)$, denoted by $y_{IA}^{\mu }$,
where $\mu =-s,\dotsc ,-1,1,\dotsc ,r$, $I=(i_{1},\dotsc ,i_{m})\in \mathbb{N%
}^{m}$, and $A=(-\alpha _{1},\dotsc ,-\alpha _{l})\in (\mathbb{Z}^{-})^{l}$,
for $l=0,\dotsc ,n$, is a strictly decreasing multi-index, such that $%
|I|+|A|\leq k$, with the assumption $y_{0\emptyset }^{\mu }=y^{\mu }$. This
system of coordinates is determined by the following equations: 
\begin{equation*}
\left( j^{k}\sigma \right) ^{\ast }y_{IA}^{\mu } =\frac{\partial ^{i_{1}}}{%
(\partial x^{1})^{i_{1}}} \circ \ldots \circ \frac{\partial ^{i_{m}}}{%
(\partial x^{m})^{i_{m}}} \circ \frac{\partial }{\partial x^{-\alpha _{l}}}
\circ \ldots \circ \frac{\partial }{\partial x^{-\alpha _{1}}} \left( \sigma
^{\ast }y^{\mu }\right) ,
\end{equation*}
for every smooth section $\sigma \colon (U,\mathcal{A}|_{U})\to (V,\mathcal{B%
}|_{V})$ of the given graded submersion.

Sometimes we will write expressions such as $y_{I,A}^{\mu }$ instead of $%
y_{IA}^{\mu }$. This will be done in order to avoid confusions involving
positive and negative multiindices.

The parity of $y_{IA}^{\mu }$ is the sum modulo $2$ of the parity of $y^{\mu
}$ and $|A|$. In particular, the parity of the coordinate system induced by %
\eqref{coordinates1} on $J_{G}^{1}(p)$ is explicitly given by 
\begin{equation*}
\left. 
\begin{array}{lll}
\left\vert y_{i}^{\mu }\right\vert =0, & i=-n,\dotsc ,-1, & \mu =-s,
\dotsc,-1 \\ 
\left\vert y_{i}^{\mu }\right\vert =1, & i=1,\dotsc ,m, & \mu =-s, \dotsc ,-1
\\ 
\left\vert y_{i}^{\mu }\right\vert =1, & i=-n,\dotsc ,-1, & \mu =1, \dotsc ,r
\\ 
\left\vert y_{i}^{\mu }\right\vert =0, & i=1,\dotsc ,m, & \mu =1, \dotsc ,r%
\end{array}
\right\}
\end{equation*}
and we accordingly have, 
\begin{equation*}
\dim \left( J_{G}^{1}(p), \mathcal{A}_{J_{G}^{1}(p)}\right)
=(m+r+mr+ns|n+s+ms+nr).
\end{equation*}

We also work with the inverse limit 
\begin{equation*}
\left( J^\infty _G(p) =\lim _\leftarrow J^k_G(p), \mathcal{A}_{J^\infty
_G(p)} =\lim _\rightarrow \mathcal{A}_{J^k _G(p)} \right)
\end{equation*}
of the system $(J^k_G(p), \mathcal{A}_{J^k_G(p)}; p_{kl},k\geq l)$, with
natural projections 
\begin{eqnarray*}
p_\infty \colon \left( J_G^\infty (p), \mathcal{A}_{J_G^\infty (p)} \right)
& \rightarrow & \left( M,\mathcal{A} \right) , \\
p_{\infty k}\colon \left( J_G^\infty (p), \mathcal{A}_{ J_G^\infty (p) }
\right) & \rightarrow & \left( J_G^k(p), \mathcal{A}_{J_G^k(p)} \right) .
\end{eqnarray*}

Given the submersion $p:(N,\mathcal{B})\rightarrow (M,\mathcal{A})$, we
denote by $\mathcal{V}(p)$ the vertical subspace of $S\mathcal{T}(N,\mathcal{%
B})$. In particular, this applies to the various $p_{k}$ and $p_{kl}$
submersions derived from $p$, so we will write $\mathcal{V}(p_{k}),\mathcal{V%
}(p_{kl})$, etc.

In the following we will work with differential operators acting on the
spaces $J_{G}^{k}(p)$, and in order to deal with the multi-index notation
(especially for negative multi-indices) it will be useful to establish the
following conventions.

\begin{enumerate}
\item We will denote positive multi-indices by the capital letters $%
I,J,K,\ldots $ and the negative ones by $A,B,C,...$ An arbitrary multi-index
(containing both positive and negative indices) will be denoted $%
P,Q,R,\ldots $ By $\mathbb{I}_{n}$ we will understand the set $\mathbb{I}%
_{n}=\{ 1,2,\dotsc,n\} $.

\item The multi-index $\emptyset $ amounts to take $0$ within any expression
in which it appears, that is: 
\begin{equation*}
\frac{\partial^{\emptyset}}{\partial x^{\emptyset}}F_{IA}=0,
G_{J\emptyset}=0.
\end{equation*}
The multi-index $(0)$ amounts to take the identity: 
\begin{equation*}
\frac{\partial^{0}}{\partial x^{0}}F_{IA}=F_{IA}, G_{J0}=G_{J}.
\end{equation*}

\item A negative multi-index $A$ with lenght $l$ in $J_{G}^{k}(p)$ has the
structure 
\begin{equation*}
A=(-\alpha_{1},\dotsc,-\alpha _{l})
\end{equation*}
with $l\leq k$, where $\alpha_{i}\in\mathbb{I}_{n}$, $\dim(M,\mathcal{A}%
)=(m|n),1\leq i\leq l$. Each $-\alpha_{i}$ gives the odd coordinate of $(M,%
\mathcal{A})$ with respect to which we are computing the derivative; that
is, the place occupied by $-\alpha_{i}$ in the multi-index only expresses
the order in which the corresponding derivative appears from left to right.
Thus, if $\dim(M,\mathcal{A})=(3|6)$, we could consider $J_{G}^{4}(p)$ and $%
A=(-3,-5,-2)$; then $\frac{\partial^{\left\vert A\right\vert }}{\partial
x^{A}}$ would represent 
\begin{equation*}
\frac{\partial^{\left\vert A\right\vert }}{\partial x^{A}} =\frac{\partial }{%
\partial x^{-3}}\circ \frac{\partial}{\partial x^{-5}}\circ \frac{\partial }{%
\partial x^{-2}}\text{.}
\end{equation*}

\item If we are dealing with $J_{G}^{k}(p)$, a negative multi-index $A$
always has lenght $l\leq k$. By convention, if the lenght of $A$ is $l>k$,
then $A=\emptyset $. Note that if $l>n,$ automatically $A=\emptyset $.
Generally, if a negative multi-index $A$ contains two repeated indices, $%
A=\emptyset $.

\item In principle, a negative multi-index does not need to be ordered, but
nothing prevents from having such ordered indices as the lenght $5$
multi-index 
\begin{equation*}
B=(-9,-7,-4,-2,-1)
\end{equation*}
in $J_{G}^{8}(p)$, with $\dim (M,\mathcal{A})=(2|9)$.

\item For negative multi-indices, we define the operation of (non-ordered) 
\emph{juxtaposition}. If $A$ has lenght $l$ and $B$ has lenght $q$, 
\begin{equation*}
\begin{array}{c}
A=(-\alpha _{1},\dotsc,-\alpha _{l})\text{ with } \alpha _{i}\in \mathbb{I}%
_{n} \\ 
B=(-\beta _{1},\dotsc,-\beta _{q})\text{ with }\beta _{j} \in \mathbb{I}_{n},%
\end{array}%
\end{equation*}
then their juxtaposition is given by: 
\begin{equation*}
A\star B=\left\{
\begin{array}{c}
(-\alpha _{1},\dotsc,-\alpha _{l},
-\beta _{1},\dotsc,-\beta _{q})
\begin{array}{c}
\text{if }l+q\leq k\text{ and }-\xi _{i}\neq -\xi _{j} \\
\text{being }\xi _{i}
\in \{\alpha _{1},\dotsc,\alpha _{l},\beta _{1},
\dotsc,\beta_{q}\}
\end{array}
\\
\\
\multicolumn{1}{l}{\emptyset \text{ other case.}}
\end{array}
\right.
\end{equation*}
Note that $A\star B\neq B\star A$. In particular, if $A=(-j)$ and $B=(-\beta
_{1},\dotsc,-\beta _{q})$, then 
\begin{equation*}
A\star B=(-j,-\beta _{1},\dotsc,-\beta _{q})
\end{equation*}
and that means 
\begin{equation*}
\frac{\partial ^{q+1}}{\partial x^{A+B}} =\frac{\partial }{\partial x^{-j}}
\circ \frac{\partial }{\partial x^{-\beta _{1}}} \circ \cdots \circ \frac{%
\partial }{\partial x^{-\beta _{q}}},
\end{equation*}
\emph{provided }$1+q\leq k$ \emph{and there are no repeated indices}.

\item If we take a positive multiindex $I$ and a negative one $A$ (or a pair
of positive multiindices) their juxtaposition is analogously defined, but in
this case it is a \emph{commutative} operation. To stress this fact we then
write $I+A,I+J$, etc.
\end{enumerate}

\subsection{Graded contact forms}

Let $p\colon (M,\mathcal{A})\rightarrow (N,\mathcal{B})$ be a graded
submersion with $(m|n)=\dim (M,\mathcal{A})$, $(m+r|n+s)=\dim (N,\mathcal{B})
$. The graded manifold $(J_{G}^{k}(p),\mathcal{A}_{J_{G}^{k}(p)})$ is
endowed with a differential system, which characterizes the holonomy of the
sections of $p_{k}\colon (J_{G}^{k}(p),\mathcal{A}_{J_{G}^{k}(p)})\to (M,%
\mathcal{A})$. Precisely, a graded $1-$form $\omega $ on $(J_{G}^{k}(p),%
\mathcal{A}_{J_{G}^{k}(p)})$ is said to be a contact form if $(j^{k}\sigma
)^{\ast }\omega =0$, for every local section $\sigma $ of $p$. With the same
assumptions and notations as in subsection \ref{notdef}, the set of contact
forms is a sheaf of $\mathcal{A}_{J_{G}^{k}(p)}-$modules locally generated
by the forms 
\begin{equation}
\theta _{IA}^{\mu }=d^{G}y_{IA}^{\mu } -\sum_{h=1}^{m}d^{G}x^{h}\cdot
y_{\{h\}\star I,A}^{\mu } -\sum_{j=1}^{n}\varepsilon (j,A)d^{G}x^{-j}\cdot
y_{I,\{-j\}\star A}^{\mu },  \label{contactforms}
\end{equation}
where $\alpha =-s,\dotsc ,-1,1,\dotsc ,r,$ $|I|+|A|\leq k-1.$

These forms fit together in order to define a global $(p_{k,k-1})^{\ast }%
\mathcal{V}(p_{k})-$valued $1-$form on $(J_{G}^{k}(p),\mathcal{A}%
_{J_{G}^{k}(p)})$, called the structure form on the graded $k-$jet bundle,
given by 
\begin{equation}
\theta ^{k}=\theta _{IA}^{\mu }\otimes \frac{\partial }{\partial y_{IA}^{\mu
}},  \label{structureform}
\end{equation}
which characterizes graded $k-$jet extensions of sections of $p$, as
follows: a section $\bar{\sigma}\colon (M,\mathcal{A})\rightarrow
(J_{G}^{k}(p),\mathcal{A}_{J_{G}^{k}(p)})$ of $p_{k}$ coincides with the $k-$%
jet extension of a certain section of $p$ if and only if, $\bar{\sigma}%
^{\ast }\theta ^{k}=0$.

\subsection{Graded lifts of vector fields\label{horver}}

Consider a graded submersion $p\colon (N,\mathcal{B})\to (M,\mathcal{A})$.
We will define liftings of graded vector fields to superjet bundles $%
J_{G}^{k}(p)$, $1\leq k\leq \infty $.

\subsubsection{Horizontal lifts\label{horizontallift}}

Let $X$ be a vector field on $(M,\mathcal{A})$. The horizontal or total
graded lift $X^{H}$ of $X$ is the vector field on $(J_{G}^{\infty }(p),%
\mathcal{A}_{J_{G}^{\infty }(p)})$ uniquely determined by the following
equations: 
\begin{equation*}
j^{k}(\sigma )^{\ast }(X^{H}(f))=X(j^{k}(\sigma )^{\ast }(f)), \quad \forall
k\in \mathbb{N},
\end{equation*}
for all open subsets $V\subseteq N$, $W\subseteq p_{k0}^{-1}(V)$, every $%
f\in \mathcal{A}_{J_{G}^{k}(p)}(W)$, and every smooth section $\sigma \colon
(U,\mathcal{A}(U))\rightarrow (V,\mathcal{B}(V))$ of $p$, with $U=\tilde{p}%
(V)$. A vector field $X$ on $J_{G}^{\infty }(p)$ is said to be horizontal if
vector fields $X_{1},\dotsc ,X_{r}$ on $(M,\mathcal{A})$ and functions $%
f^{1},\dotsc ,f^{r} \in \mathcal{A}_{J_{G}^{\infty }(p)}$ exist, such that $%
X=f^{i}(X_{i})^{H}$.

If $(x^{\alpha },y^{\mu })$ is a fibred coordinate system for the submersion 
$p$, then the expression for the horizontal lift of the basic vector field $%
\partial /\partial x^{\alpha }$ in the induced coordinate system, is 
\begin{align}
\frac{d}{dx^{\alpha }}& =\left( \frac{\partial }{\partial x^{\alpha }}
\right) ^{H}  \notag \\
& =\frac{\partial }{\partial x^{\alpha }} +y_{\{\alpha \}\star Q}^{\mu } 
\frac{\partial }{\partial y_{Q}^{\mu }}.  \label{XH}
\end{align}

The map $X\mapsto X^{H}$ is an $\mathcal{A}-$linear injection of Lie
algebras (cf.\/ \cite{Mon 92a, Mon-Mun 92a}). Note that $X^{H}$ is $%
p_{\infty }-$projectable onto $X$. Moreover, we can consider $\mathcal{A}%
_{J_{G}^{k+1}(p)}$ as a sheaf of $\mathcal{A}_{J_{G}^{k}(p)}-$algebras via
the natural injection 
\begin{equation*}
p_{k+1,k}^{\ast } \colon \mathcal{A}_{J_{G}^{k}(p)} \to \mathcal{A}%
_{J_{G}^{k+1}(p)},
\end{equation*}
and, for every $k\in \mathbb{N}$, $X^{H}$ induces a derivation of $\mathcal{A%
}_{J_{G}^{k}(p)}-$modules, 
\begin{equation*}
X^{H}\colon \mathcal{A}_{J_{G}^{k}(p)}\to \mathcal{A}_{J_{G}^{k+1}(p)}.
\end{equation*}

Let $\Omega _{G}^{k}(J_{G}^{\infty }(p))$ be the space of graded
differential $k-$forms on $J_{G}^{\infty }(p)$. We denote by $%
H_{r}^{s}(J_{G}^{\infty }(p))$ the module of $(r+s)-$forms on $J_{G}^{\infty
}(p)$ that are $r-$times horizontal and $s-$times vertical; that is, such
that they vanish when acting on more than $s$ $p_{\infty }-$vertical vector
fields or more than $r$ $p_{\infty }-$horizontal vector fields.

Let $d^{G}$ be the exterior differential, and let 
\begin{align*}
& D\colon H_{r}^{s}(J_{G}^{\infty }(p)) \rightarrow
H_{r+1}^{s}(J_{G}^{\infty }(p)) \\
& \partial \colon H_{r}^{s}(J_{G}^{\infty }(p)) \rightarrow
H_{r}^{s+1}(J_{G}^{\infty }(p))
\end{align*}
be the horizontal and vertical differentials, respectively. We have 
\begin{align*}
d^{G}& =D+\partial , \\
D^{2}& =0, \\
\partial ^{2}& =0, \\
D\circ \partial +\partial \circ D& =0.
\end{align*}
We can make a local refinement of the bigrading above, which depends on the
chart chosen but we will make use of it only when computing in local
coordinates. Let $(W,\mathcal{A}_{J_{G}^{\infty }(p)}(W))$ be an open
coordinate domain in $J_{G}^{\infty }(p)$. Since 
\begin{equation*}
\left( \frac{\partial }{\partial x^{i}}, \frac{\partial }{\partial x^{-j}}
\right) _{1\leq i\leq m,1\leq j\leq n}
\end{equation*}
is a basis of vector fields for $(M,\mathcal{A})$, we can define $%
H_{r_{1},r_{2}}^{s}(W)$ to be the submodule of differential forms in $%
H_{r_{1}+r_{2}}^{s}(W)$ such that they vanish when acting on more than $%
r_{1} $ vector fields among the $\partial /\partial x^{i}$, or when acting
on more than $r_{2}$ vector fields among the $\partial /\partial x^{-j}$.
Therefore, 
\begin{equation*}
H_{r}^{s}(W)=\bigoplus\limits_{r_{1}+r_{2}=r}H_{r_{1},r_{2}}^{s}(W),
\end{equation*}
with projections $\pi _{r_{1},r_{2}}\colon H_{r}^{s}(W)\rightarrow
H_{r_{1},r_{2}}^{s}(W)$. Considering the action of $D$ on a fixed $%
H_{r_{1},r_{2}}^{s}(W)$, we define 
\begin{align*}
D_{0}& =\pi _{r_{1}+1,r_{2}}\circ D, \\
D_{1}& =D-D_{0}.
\end{align*}

\subsubsection{Infinitesimal contact transformations}

Let $p\colon (M,\mathcal{A})\rightarrow (N,\mathcal{B})$ be a graded
submersion. A homogeneous vector field $Y$ on $(J_{G}^{k}(p),\mathcal{A}%
_{J_{G}^{k}(p)})$ is said to be a $k-$order graded infinitesimal contact
transformation if an endomorphism $h$ of $\mathcal{A}_{J_{G}^{k}(p)} \otimes
_{\mathcal{B}}\mathrm{Der}_{\mathcal{A}}(\mathcal{B})$---considered as a
left $\mathcal{A}_{J_{G}^{k}(p)}-$module---exists such that, 
\begin{equation*}
\mathcal{L}_{Y}^{G}\theta ^{k}=h\circ \theta ^{k},
\end{equation*}
where $\theta ^{k}$ is the structure form (recall (\ref{structureform})).

\begin{theorem}[\protect\cite{Her-Mun 87}]
\label{gcit} Let $p\colon (N,\mathcal{B})\rightarrow (M,\mathcal{A})$ be a
graded submersion. For every graded vector field $X$ on $(N,\mathcal{B})$,
there exists a unique $k-$order graded infinitesimal contact transformation $%
X_{(k)}$ on $(J_{G}^{k}(p),\mathcal{A}_{J_{G}^{k}(p)})$ projecting onto $X$.
\end{theorem}

Moreover, for every $k>l$, the vector field $X_{(k)}$ projects onto $X_{(l)}$
via the natural map $p_{kl}\colon (J_G^k(p), \mathcal{A}_{J_G^k(p)}) \to
(J_G^l(p), \mathcal{A}_{J_G^l}(p))$.

\section{Berezinian sheaves}

\subsection{The Berezinian sheaf of a supermanifold}

The Berezinian sheaf is a geometrical object designed to make possible an
integration theory in supermanifolds, tailored to the needs coming from
Physics. A global description of it can be given as follows (see \cite%
{Her-Mun 85b, Mon 92a}, cf.\ \cite{Man 88}).

Let $(M,\mathcal{A})$ be a graded manifold, of dimension $(m|n)$, and let $%
\mathrm{P}^{k}(\mathcal{A})$ be the sheaf of graded differential operators
of $\mathcal{A}$ of order $k$. This is the submodule of $\mathop{\rm {}End}%
\nolimits(\mathcal{A})$ whose elements $P$ satisfy the following conditions: 
\begin{equation*}
\left[ \,\ldots \left[ \left[ P,a_{0} \right] ,a_{1} \right] , \dotsc ,
a_{k} \right] =0, \quad \forall a_{0},\dotsc ,a_{k}\in \mathcal{A}.
\end{equation*}
Here the element $a\in \mathcal{A}$ is identified with the endomorphism $%
b\mapsto ab$. The sheaf $\mathrm{P}^{k}(\mathcal{A})$ has two essentially
different structures of $\mathcal{A}-$module: For every $P\in \mathrm{P}^{k}(%
\mathcal{A})$ and every $a,b\in \mathcal{A}$,

\begin{enumerate}
\item the left structure is given by $(a\cdot P)(b)=a\cdot P(b)$,

\item and the right structure is given by $(P\cdot a)(b)=P(a\cdot b)$.
\end{enumerate}

This is important as the Berezinian sheaf is considered with its structure
of \emph{right} $\mathcal{A}-$module.

One has that if $(x^{i},x^{-j})$, $1\leq j\leq n$, $1\leq i\leq m$, are
supercoordinates for a splitting neighborhood $U\subset M$, $\mathrm{P}^{k}(%
\mathcal{A}(U))$ is a free module (for both structures, left and right) with
basis 
\begin{equation*}
\left. 
\begin{array}{r}
\left( \dfrac{\partial }{\partial x^{1}} \right) ^{\alpha _{1}} \circ \ldots
\circ \left( \dfrac{\partial }{\partial x^{m}} \right) ^{\alpha _{m}} \circ
\left( \dfrac{\partial }{\partial x^{-1}} \right) ^{\beta _{1}}\circ \ldots
\circ \left( \dfrac{\partial }{\partial x^{-n}} \right) ^{\beta _{n}},
\medskip \\ 
\alpha _{1} +\ldots +\alpha _{m} +\beta _{1} +\ldots +\beta _{n}\leq k.%
\end{array}
\right\}
\end{equation*}
Let us consider the sheaf $\mathrm{P}^{k}(\mathcal{A},\Omega _{G}^{m})
=\Omega _{G}^{m}(M)\otimes _{\mathcal{A}} \mathrm{P}^{k}(\mathcal{A}) $, of $%
m-$form valued differential operators on $\mathcal{A}$ of order $k$, and for
every open subset $U\subset M$, let $\mathrm{K}^{n}(U)$ be the set of
operators $P\in \mathrm{P}^{n}(\mathcal{A}(U),\Omega _{G}^{m}(U))$ such
that, for every $a\in \mathcal{A}(U)$ with compact support, there exists an
ordinary $(m-1)-$form of compact support, $\omega $, satisfying 
\begin{equation*}
\widetilde{P(a)}=d\omega .
\end{equation*}
The idea is to take the quotient of $\mathrm{P}^{n}(\mathcal{A},\Omega
_{G}^{m})$ by $\mathrm{K}^{n}$; in this way, when we later define the
integral operator, two sections differing in a total differential will be
regarded as equivalent (Stokes Theorem). Having this in mind, we observe
that $\mathrm{K}^{n}$ is a submodule of $\mathrm{P}^{n}(\mathcal{A},\Omega
_{G}^{m})$ for its \emph{right structure}, so we can take quotients and
obtain the following description of the Berezinian sheaf, $\mathop{\rm {}Ber}%
\nolimits(\mathcal{A})$: 
\begin{equation*}
\mathop{\rm {}Ber}\nolimits(\mathcal{A}) \simeq \mathrm{P}^{n}(\mathcal{A},
\Omega _{G}^{m})\diagup \mathrm{K}^{n}.
\end{equation*}
We write this as an equivalence because there are other definitions of the
Berezinian sheaf. For us, however, this is $\emph{the}$ definition.

According to this description, a local basis of $\mathop{\rm {}Ber}\nolimits(%
\mathcal{A})$ can be given explicitly: If $(x^i,x^{-j})$, $1\leq j\leq n$, $%
1\leq i\leq m $, are supercoordinates for a splitting neighborhood $U\subset
M$, the local sections of the Berezinian sheaf are written in the form 
\begin{equation}  \label{a.1}
\Gamma \left( U,\mathop{\rm {}Ber}\nolimits (\mathcal{A}) \right) =\left[
d^Gx^1\wedge \cdots \wedge d^Gx^m\otimes \frac{\partial }{\partial x^{-1}}
\circ \cdots \circ \frac{\partial }{\partial x^{-n}} \right] \cdot \mathcal{A%
}(U),
\end{equation}
where $[\cdot ]$ stands for the equivalence class modulo $\mathrm{K}^n$.

\subsection{Higher order Berezinian sheaf}

Let $p\colon (N,\mathcal{B})\rightarrow (M,\mathcal{A})$ be a graded
submersion. Given $P\in \mathrm{P}^{l}(\mathcal{A},\Omega _{G}^{m})$, let 
\begin{equation*}
P^{H}\colon \mathcal{A}_{J_{G}^{k}(p)}\rightarrow H_{m}^{0}(J_{G}^{k}(p))
\end{equation*}
be the first$-$order operator defined by the condition, 
\begin{equation*}
j^{k}(\sigma )^{\ast }P^{H}f=Pj^{k}(\sigma )^{\ast }f,
\end{equation*}
for every $f\in \mathcal{A}_{J_{G}^{k}(p)}$ and every local section $\sigma $
of $p$. We call $P^{H}$ the total or horizontal lift of $P$. Let us denote
by $\mathrm{P}H^{l}(\mathcal{A}_{k},H_{m}^{0})$ (resp.\ $\mathrm{K}H_{l}(%
\mathcal{A}_{k})$) the sheaf of those operators in 
\begin{equation*}
\mathrm{P}^{l} \left( \mathcal{A}_{J_{G}^{k}(p)},H_{m}^{0} \left(
J_{G}^{k}(p )\right) \right)
\end{equation*}
that are horizontal lifts of operators of $\mathrm{P}^{l}(\mathcal{A},\Omega
_{G}^{m})$ (resp.\ $\mathrm{K}_{l}(\mathcal{A})$). Then, the $k-$order
Berezinian sheaf is defined as 
\begin{equation*}
\mathop{\rm {}Ber}\nolimits^{k}(\mathcal{A}_{k}) =\frac{\mathrm{P}H^{n}(%
\mathcal{A}_{k},H_{m}^{0})} {\mathrm{K}H_{n}(\mathcal{A}_{k})} \otimes 
\mathcal{A}_{J_{G}^{k}(p)}.
\end{equation*}
According to this description, a local basis for $\mathop{\rm {}Ber}%
\nolimits^{k}(\mathcal{A}_{k})$ can be given explicitly: If $(x^{i},x^{-j})$%
, $1\leq i\leq m$, $1\leq j\leq n$, are the graded $\mathcal{A}-$coordinates
for the coordinate open domain $(U,\mathcal{A}(U))$ and $(V,\mathcal{B}(V))$
is a $\mathcal{B}-$coordinate open domain with a suitable $V\subseteq \tilde{%
p}^{-1}(U)$, then, if $W$ is an open subset in $J_{G}^{k}(p)$ such that $%
W\subseteq \tilde{p}_{k}^{-1}(U)$ we have 
\begin{equation*}
\Gamma \left( W,\mathop{\rm {}Ber}\nolimits^{k}(\mathcal{A}_{k}) \right) =%
\left[ d^{G}x^{1}\wedge \cdots \wedge d^{G}x^{m} \otimes \frac{d}{dx^{-1}}
\circ \ldots \circ \frac{d}{dx^{-n}} \right] \cdot \mathcal{A}%
_{J_{G}^{k}(p)}(W).
\end{equation*}

\subsection{The Berezin integral}

Given a supermanifold $(M,\mathcal{A})$, the Berezin integral can be defined
over the sections of the Berezinian sheaf with compact support, by means of
the formula 
\begin{equation}
\begin{array}{c}
\int_{\mathop{\rm {}Ber}\nolimits} \colon \Gamma _{U}^{c} (\mathop{\rm {}Ber}%
\nolimits(\mathcal{A})) \rightarrow \mathbb{R} \\ 
\lbrack P]\longmapsto \int_{U}\widetilde{P(1)}.%
\end{array}
\label{a.2}
\end{equation}
In this expression, $M$ is assumed to be oriented, and the right integral is
taken with respect to that orientation. In this sense, having a fixed volume
form on $M$ is not a loss of generality.

\subsubsection{An example}

\label{ex1} Let $(M,\mathcal{A}) =(\mathbb{R}^{m},C^{\infty }(\mathbb{R}%
^{m}) \otimes \Omega (\mathbb{R}^{n}))$ be the standard graded manifold. A
section of $\mathcal{A}$ is just a differential form $\rho=f_{I}(x^{1},%
\dotsc ,x^{m})x^{-I}$, $0\leq |I|\leq n$, where $(x^{i})$, $1\leq i\leq m$,
are the coordinates of $\mathbb{R}^{m}$ and we write $x^{-j}=dx^{j}$ for the
odd coordinates; thus 
\begin{equation*}
\rho =f_{0}+f_{j}x^{-j}+\ldots +f_{1\ldots n}x^{-1}\cdots x^{-n},
\end{equation*}
and we recover the formula for Berezin's expression common in Physics
textbooks (except for a global sign); i.e., \textquotedblleft to integrate
the component of highest odd degree \textquotedblright\ (see \cite{Ber 66}): 
\begin{multline*}
\int_{\mathop{\rm {}Ber}\nolimits} \left[ d^{G}x^{1}\wedge \cdots \wedge
d^{G}x^{m}\otimes \frac{\partial }{\partial x^{-1}} \circ \ldots \circ \frac{%
\partial }{\partial x^{-n}} \right] \cdot \rho \\
=(-1)^{\binom{n}{2}}\int_{\mathbb{R}^{m}} f_{1\ldots n}dx^{1}\cdots dx^{m}.
\end{multline*}

\subsubsection{Lie derivative on the Berezinian sheaf}

If $X$ is a graded vector field, it is possible to define the notion of
graded Lie derivative of sections of the Berezinian sheaf with respect to $X$%
. This is the mapping 
\begin{equation*}
\mathcal{L}_{X}^{G}\colon \Gamma \left( \mathop{\rm {}Ber}\nolimits(\mathcal{%
A}) \right) \rightarrow \mathrm{P}^{n+1} (\mathcal{A},\Omega
_{G}^{m})\diagup \mathrm{K}^{n+1} =\Gamma \left( \mathop{\rm {}Ber}\nolimits(%
\mathcal{A}) \right) ,
\end{equation*}
given by 
\begin{equation}
\mathcal{L}_{X}^{G} \left[ \eta ^{G}\otimes P \right] =(-1)^{|X||\eta
^{G}\otimes P|+1} \left[ \eta ^{G}\otimes P\circ X \right] ,  \label{a.3}
\end{equation}
for homogeneous $X$ and $\eta ^{G}\otimes P$.

This Lie derivative has the properties that one could expect:

\begin{enumerate}
\item For homogeneous $X\in \mathop{\rm {}Der}\nolimits (\mathcal{A})$, $\xi
\in \Gamma (\mathop{\rm {}Ber}\nolimits (\mathcal{A}))$ and $a\in\mathcal{A}$%
, 
\begin{equation*}
\mathcal{L}_X^G (\xi\cdot a) =\mathcal{L}_X^G(\xi ) \cdot a +(-1)^{|X||\xi
|}\xi \cdot X(a).
\end{equation*}

\item For homogeneous $X\in \mathop{\rm {}Der}\nolimits(\mathcal{A})$, $\xi
\in \Gamma (\mathop{\rm {}Ber}\nolimits(\mathcal{A}))$ and $a\in \mathcal{A}$%
, 
\begin{equation*}
\mathcal{L}_{a\cdot X}^{G}(\xi )=(-1)^{|a|(|X|+|\xi |)} \mathcal{L}%
_{X}^{G}(\xi \cdot a).
\end{equation*}

\item Given a system of supercoordinates $(x^i,x^{-j})$, $1\leq j\leq n$, $%
1\leq i\leq m$, if 
\begin{equation*}
\xi _{x^i,x^{-j}} =\left[ d^Gx^1\wedge \cdots \wedge d^Gx^m \otimes \frac{%
\partial }{\partial x^{-1}} \circ \cdots \circ \frac{\partial }{\partial
x^{-n}} \right]
\end{equation*}
is the local generator of the Berezinian sheaf, then 
\begin{align*}
\mathcal{L}_{ \frac{\partial } {\partial x^i} }^G \left( \xi _{x^i,x^{-j}}
\right) & =0, \\
\mathcal{L}_{ \frac{\partial } {\partial x^{-j}} }^G \left( \xi
_{x^i,x^{-j}} \right) & =0.
\end{align*}
\end{enumerate}

\subsubsection{Berezinian divergence\label{berdiv}}

We can now introduce the notion of Berezinian divergence. Let $(M,\mathcal{A}%
)$ be a graded manifold whose Berezinian sheaf is generated by a section $%
\xi $. The graded function $\mathop{\rm {}div}\nolimits_{B}^{\xi }(X)$
given---for homogeneous $X$---by the formula 
\begin{equation*}
\mathcal{L}_{X}^{G}(\xi ) =(-1)^{|X||\xi |}\xi \cdot \mathop{\rm {}div}
\nolimits_{B}^{\xi }(X)
\end{equation*}
(and extended by $\mathcal{A}-$linearity) is called the Berezinian
divergence of $X$ with respect to $\xi $. When there is no risk of
confusion, we simply write $\mathop{\rm {}div}\nolimits_{B}(X)$.

For example, if we consider the standard graded manifold of Example \ref{ex1}%
; i.e., $(M,\mathcal{A}) =(\mathbb{R}^{m},C^{\infty }(\mathbb{R}^{m})
\otimes \Omega (\mathbb{R}^{n}))$, then $\mathop{\rm {}Ber}\nolimits(%
\mathcal{A})$ is trivial and generated by 
\begin{equation*}
\xi =\left[ d^{G}x^{1}\wedge \cdots \wedge d^{G}x^{m}\otimes \frac{\partial 
}{\partial x^{-1}}\circ \ldots \circ \frac{\partial }{\partial x^{-n}} %
\right] ,
\end{equation*}
and the Berezinian divergence of a graded vector field $X=f_{i}\partial
/\partial x^{i}+g_{j}\partial /\partial x^{-j}$ with respect to $\xi $ is
given by 
\begin{equation}
\mathop{\rm {}div}\nolimits_{B}(X)=\sum\limits_{i=1}^{m} \frac{\partial f_{i}%
}{\partial x^{i}} +\sum\limits_{j=1}^{n}(-1)^{|g_{j}|} \frac{\partial g_{j}}{%
\partial x^{-j}}.  \label{a4}
\end{equation}

Having in mind the previous section, these notions can be carried over to
higher orders with the appropriate modifications.

\subsection{Graded and Berezinian Lagrangian densities\label{comparison}}

Let us introduce the notion of variational problems in terms of the
Berezinian sheaf.

A \emph{Berezinian Lagrangian density} of order $k$ for a graded submersion $%
p\colon (N,\mathcal{B})$$\to (M,\mathcal{A})$, is a section 
\begin{equation*}
\lbrack P^{H}]\cdot L\in \Gamma (\mathop{\rm {}Ber}\nolimits ^{k} (\mathcal{A%
}_{k})).
\end{equation*}
In particular, a first$-$order Berezinian Lagrangian density can locally be
written as $\xi \cdot L$, where 
\begin{equation*}
\xi =\left[ d^{G}x^{1} \wedge \cdots \wedge d^{G}x^{m}\otimes \frac{d}{%
dx^{-1}} \circ \ldots \circ \frac{d}{dx^{-n}} \right]
\end{equation*}
and $L\in \mathcal{A}_{J_{G}^{k}(p)}$ is an element of the ring of graded
functions on the graded bundle of $1-$jets $(J_{G}^{1}(p),\mathcal{A}%
_{J_{G}^{1}(p)})$. In this paper, we only consider first$-$order Berezinian
Lagrangian densities and assume that $M$ is oriented by an ordinary volume
form $\eta $.

The variation of the Berezinian functional associated to $\xi \cdot L$,
along a section $s$ of $p\colon (N,\mathcal{B})\rightarrow (M,\mathcal{A})$,
is the mapping 
\begin{equation*}
\begin{array}{llll}
\delta _{s}\mathrm{I}_{\mathop{\rm {}Ber}\nolimits}(L) \colon & \mathcal{V}%
^{c}(N) & \longrightarrow & \mathbb{R} \\ 
& Y & \longmapsto & \int_{\mathop{\rm {}Ber}\nolimits}(j^{1}s)^{\ast } (%
\mathcal{L}_{Y_{(1)}}^{G}(\xi \cdot L)),%
\end{array}%
\end{equation*}
where $\mathcal{V}^{c}(N)$ denotes the space of graded vector fields on $(N,%
\mathcal{B})$, which are vertical over $(M,\mathcal{A})$ and whose support
has compact image on $M$; $Y_{(1)}$ is the graded infinitesimal contact
transformation associated to $Y$, and $\mathcal{L}_{Y_{(1)}}^{G}(\xi \cdot
L) $ is defined by means of \eqref{a.3}, which makes sense as $Y$ is $p-$%
projectable. A section $s$ is called a Berezinian extremal if $\delta _{s}%
\mathrm{I}_{\mathop{\rm {}Ber}\nolimits}=0$.

Finally, we turn our attention to the relation between Berezinian and graded
variational problems. As we will see shortly, even restricting ourselves to
first$-$or\-der Berezinian Lagrangian densities we must consider higher$-$%
order graded Lagrangian ones.

A \emph{graded Lagrangian density} of order $k$ for a graded submersion $%
p\colon (N,\mathcal{B})\rightarrow (M,\mathcal{A})$ is a section 
\begin{equation*}
\eta ^{G}\cdot L\in \Omega _{G}^{m}(M) \otimes _{\mathcal{A}}\mathcal{A}%
_{J_{G}^{k}(p)},
\end{equation*}
where $(m|n)=\dim (M,\mathcal{A})$, $\eta ^{G}$ is a graded $m-$form on $(M,%
\mathcal{A})$, and $L$ is an element of the graded ring $\mathcal{A}%
_{J_{G}^{k}(p)}$, of graded functions on the graded $k-$jet bundle $%
J_{G}^{k}(p)$.

The variation of the functional associated to a graded $k-$order Lagrangian
density $\eta ^{G}\cdot L$ along a section $s$ of $p\colon (N,\mathcal{B}%
)\to (M,\mathcal{A})$, is the mapping \vskip5mm\hskip66mm 
\begin{picture}(40,20) 
\qbezier(0,10)(20,20)(40,10) \qbezier(40,10)(60,0)(80,10) 
\end{picture}
\vskip-11mm 
\begin{equation*}
\begin{array}{llll}
\delta _s \mathrm{I}_{ \func{grad} }^k(L) \colon & \mathcal{V}_G^c(N) & 
\longrightarrow & \mathbb{R} \\ 
& Y & \longmapsto & \int _M {\ \left( j^ks \right) ^\ast \left( \mathcal{L}%
_{Y_{(k)}}^G (\eta ^G\cdot L) \right) },%
\end{array}%
\end{equation*}
where $\mathcal{V}^{c}(N)$ is as before and $Y_{(k)}$ is the $k-$graded
infinitesimal contact transformation prolongation of $Y$.

Berezinian and graded variational problems are related through the following
result (usually known as the Comparison Theorem):

\begin{theorem}[\protect\cite{Her-Mun 87, Mon 92a}]
\label{comparisonth} Let $p\colon (N,\mathcal{B})\rightarrow (M,\mathcal{A})$
be a graded submersion, with $(m|n)=\dim (M,\mathcal{A})$. Every first$-$%
order Berezinian Lagrangian density $\xi \cdot L$ for $p$ is equivalent to a
graded Lagrangian density of order $n+1$ in the following sense: There
exists an element $L^{\prime }$ in the graded ring $\mathcal{A}%
_{J_{G}^{n+1}(p)}$ of the graded $(n+1)-$jet bundle $J_{G}^{n+1}(p)$ such
that the Berezinian variation of the functional associated to $\xi \cdot L$
equals the graded variation of the functional associated to $\eta ^{G}\cdot
L^{\prime } =d^{G}x^{1}\wedge \cdots \wedge d^{G}x^{m}\cdot L^{\prime }$;
that is, 
\begin{equation*}
\left( \delta _{s}\mathrm{I}_{\mathop{\rm {}Ber}\nolimits}(L) \right)
(Y)=\left( \delta _{s}\mathrm{I}_{\mathop{\rm {}grad}\nolimits}^{n+1}(L^{%
\prime }) \right) (Y),
\end{equation*}
for every section $s$ of $p$, and every graded $p-$vertical $Y\in \mathcal{V}%
^{c}(N,\mathcal{B})$.
\end{theorem}

\section{The $\mathcal{J}_{k}$ operators\label{jotas}}

As stated in the Introduction, our intention is to study the Cartan
formalism for variational problems and in this formalism a central object is
the so called Cartan form for field theory, denoted $\Theta _{0}^{L}$ and
locally given by 
\begin{eqnarray}
\Theta _{0}^{L} &=&\sum\limits_{i=1}^{m}\sum\limits_{\beta =-s}^{r}
(-1)^{m+i}d^{G}x^{1}\wedge \ldots \wedge \widehat{d^{G}x^{i}} \wedge \ldots
\wedge d^{G}x^{m}  \label{cartanlocal} \\
&&\wedge \left( d^{G}y^{\mu } -\sum\limits_{\alpha =-n}^{m}d^{G}x^{\alpha
}\cdot y_{\alpha }^{\mu } \right) \frac{\partial L}{\partial y_{i}^{\mu }}
+\eta ^{G}\cdot L.  \notag
\end{eqnarray}

We will provide an intrinsic construction of $\Theta _{0}^{L}$ and we will
develop from it a consistent theory of the first$-$order calculus of
variations on supermanifolds. The idea is the same as those used in the
formulation of mechanics (see \cite{God 69, Sau 89}), but with some new
details that arise because this time we deal with fields (for an interesting
discussion of the classical formalism in this case see also \cite{GMS 97}
and \cite{GMS 98}); let us describe it very briefly.

The graded generalization of the vertical endomorphism of the tangent bundle
used in classical mechanics would be (unlike the case of mechanics, note
that $\mathcal{\tilde{J}}$ is \emph{not} an endomorphism here): 
\begin{equation*}
\mathcal{\tilde{J}}\doteq \sum\limits_{i=1}^{m} \sum \limits_{\beta =-s}^{r}
(-1)^{m+i}d^{G}x^{1}\wedge \ldots \wedge \widehat{d^{G}x^{i}} \wedge \ldots
\wedge d^{G}x^{m}\wedge d^{G}y^{\mu } \otimes \frac{\partial }{\partial
y_{i}^{\mu }}.
\end{equation*}
Also, for each $\alpha \in \{-n,\dotsc ,-1,1,\dotsc ,m\}$, $i\in \{1,\dotsc
,m\}$, the graded analogue of the Liouville vector field would be 
\begin{equation*}
\Delta _{\alpha i} =\sum\limits_{\beta =-s}^{r}(-1)^{m+i} y_{\alpha }^{\mu } 
\frac{\partial }{\partial y_{i}^{\mu }},
\end{equation*}
and finally (by using Einstein's convention, from now on we omit the
summation symbols), 
\begin{equation*}
\mathcal{J}\doteq \mathcal{\tilde{J}} -d^{G}x^{1}\wedge \cdots \wedge 
\widehat{d^{G}x^{i}} \wedge \cdots \wedge d^{G}x^{m}\wedge d^{G}x^{\alpha }
\otimes \Delta _{\alpha i}.
\end{equation*}

Let us evaluate $\mathcal{L}_{\mathcal{J}}^{G}(L)$. It will be useful to
bear in mind the developed expression for $\mathcal{J}$: 
\begin{equation}
\mathcal{J}=(-1)^{m+i}d^{G}x^{1}\wedge \cdots \wedge \widehat{d^{G}x^{i}}%
\wedge \cdots \wedge d^{G}x^{m}\wedge \left( d^{G}y^{\mu }-d^{G}x^{\alpha
}\cdot y_{\alpha }^{\mu }\right) \otimes \frac{\partial }{\partial
y_{i}^{\mu }}.  \label{jotacal}
\end{equation}%
Note that 
\begin{align*}
\iota _{\frac{\partial }{\partial y_{i}^{\mu }}}(d^{G}L)& =\iota _{\frac{%
\partial }{\partial y_{i}^{\mu }}}\left( d^{G}x^{\alpha }\cdot \frac{dL}{%
dx^{\alpha }}+d^{G}y^{\nu }\cdot \frac{\partial L}{\partial y^{\nu }}%
+d^{G}y_{\alpha }^{\nu }\cdot \frac{\partial L}{\partial y_{\alpha }^{\nu }}%
\right)  \\
& =\frac{\partial L}{\partial y_{i}^{\nu }}.
\end{align*}%
We then have 
\begin{align*}
\mathcal{L}_{\mathcal{J}}^{G}(L)& =\iota _{\mathcal{J}}(d^{G}L) \\
& =(-1)^{m+i}d^{G}x^{1}\wedge \cdots \wedge \widehat{d^{G}x^{i}}\wedge
\cdots \wedge d^{G}x^{m} \\
& \qquad \qquad \qquad \qquad \qquad \wedge \left( d^{G}y^{\mu
}-d^{G}x^{\alpha }\cdot y_{\alpha }^{\mu }\right) \cdot \iota _{\frac{%
\partial }{\partial y_{i}^{\mu }}}(d^{G}L) \\
& =(-1)^{m+i}d^{G}x^{1}\wedge \cdots \wedge \widehat{d^{G}x^{i}}\wedge
\cdots \wedge d^{G}x^{m} \\
& \qquad \qquad \qquad \qquad \qquad \wedge (d^{G}y^{\mu }-d^{G}x^{\alpha
}\cdot y_{\alpha }^{\mu })\cdot \frac{\partial L}{\partial y_{i}^{\mu }} \\
& =\Theta _{0}^{L}-d^{G}x^{1}\wedge \cdots \wedge d^{G}x^{m}\cdot L,
\end{align*}%
so that 
\begin{equation*}
\Theta _{0}^{L}=\mathcal{L}_{\mathcal{J}}^{G}(L)+\eta ^{G}\cdot L.
\end{equation*}%
Thus, to have $\Theta _{0}^{L}$ intrinsically defined, there remains to
prove that this is the case for $\mathcal{J}$. Notice that $\mathcal{J}$ is
the graded analogue of the $(1,m)-$tensor field $S_{\eta }$ that appears in 
\cite{Sau 89}\ (for arbitrary $m$, see pp.\ 156$-$158). We will now study
the intrinsic construction of these objects in the graded context, but the
generalization is not straightforward, as the classical point constructions
are not applicable now.

\subsection{Algebraic preliminaries\label{preliminaries}}

Let $p\colon (N,\mathcal{B})\rightarrow (M,\mathcal{A})$ be a graded
submersion. Consider the cotangent supervector bundle $\mathcal{ST}^{\ast}(M,%
\mathcal{A})\rightarrow (M,\mathcal{A})$, and its pull-back $p^{\ast }%
\mathcal{ST}^{\ast }(M,\mathcal{A})$ to $(N,\mathcal{B})$. Furthermore, let $%
\mathcal{V}(p)\subset \mathcal{ST}(N,\mathcal{B})$ be the vertical
sub-bundle of $p$. This is the supervector bundle on $(N,\mathcal{B})$
defined by the short exact sequence 
\begin{equation}
0\rightarrow \mathcal{V}(p) \rightarrow \mathcal{ST}( N,\mathcal{B})\overset{%
p_{\ast }}{\longrightarrow }p^{\ast } \mathcal{ST}(M,\mathcal{A})
\rightarrow 0.  \label{eq5_4}
\end{equation}
We can thus construct the tensor product supervector bundle 
\begin{equation*}
\pi \colon p^{\ast }\mathcal{ST}^{\ast } (M,\mathcal{A})\otimes \mathcal{V}%
(p) \rightarrow J_{G}^{0}(p)\simeq (N,\mathcal{B}).
\end{equation*}
From a homological point of view, we have a natural identification 
\begin{equation*}
p^{\ast }\mathcal{ST}^{\ast }(M,\mathcal{A}) \otimes \mathcal{V}(p)\simeq %
\mathop{\rm {}Hom}\nolimits \left( p^{\ast }\mathcal{ST}(M,\mathcal{A}), 
\mathcal{V}(p) \right) ,
\end{equation*}
and whithin this algebraic setting, we can obtain a representation for $%
J_{G}^{1}(p)$ by considering the short exact sequence \eqref{eq5_4} and
thinking of $J_{G}^{1}(p)$ as being the space of its splittings.

\begin{proposition}
\label{identification} Let $p\colon (N,\mathcal{B})\rightarrow (M,\mathcal{A}%
)$ be a graded submersion with dimensions $(m|n)=\dim (M,\mathcal{A})$, $%
(m+r|n+s)=\dim (N,\mathcal{B})$. A unique isomorphism 
\begin{equation}
p_{10}^{\ast } \left( p^{\ast }\mathcal{ST}^{\ast }(M,\mathcal{A})\otimes 
\mathcal{V}(p) \right) \simeq \mathcal{V}(p_{10}),  \label{isomorphism}
\end{equation}
exists, which is given by 
\begin{equation}
d^{G}x^{i}\otimes \frac{\partial }{\partial y^{\mu }} \mapsto \frac{\partial 
}{\partial y_{i}^{\mu }},  \label{coordinates2}
\end{equation}
on every fibred coordinate system \eqref{coordinates1}.
\end{proposition}

\begin{proof}
As the formula
\eqref{coordinates2}
completely determines the
isomorphism \eqref{isomorphism},
we need only to prove that
the isomorphism is independent
of the fibred coordinate system
chosen. This reduces to compute
how the tensor fields in the
formula \eqref{coordinates2}
transform under a change of
fibred coordinates, from
\begin{equation}
\left.
\begin{array}
[c]{ll}
\left(
x^\alpha
\right) ,
&
\alpha =-n,\dotsc,-1,1,\dotsc,m,\\
\left(
y^\mu
\right) ,
&
\mu
=-s,\dotsc,-1,1,\dotsc,r,
\end{array}
\right\}
\label{coordinates3}
\end{equation}
to
\begin{equation}
\left.
\begin{array}
[c]{ll}
\left(
\bar{x}^\beta
\right) ,
&
\alpha =-n,\dotsc,-1,1,\dotsc,m,\\
\left(
\bar{y}^\nu
\right) ,
&
\nu
=-s,\dotsc,-1,1,\dotsc,r,
\end{array}
\right\}
\label{coordinates4}
\end{equation}
and the corresponding change
in $J_G^1(p).$
From the very definition of
$y_\gamma ^\rho $ as a
coordinate in $J_G^1(p)$
(e.g.\ see
\cite[Section 1]{Her-Mun 85a}),
we can compute
\begin{eqnarray*}
d^G\bar{y}^\nu
\otimes
\frac{\partial }
{\partial \bar{x}^\beta }
&=&
\left(
d^Gx^\alpha
\frac{\partial \bar{y}^\nu }
{\partial x^\alpha }
+d^Gy^\mu
\frac{\partial \bar{y}^\nu }
{\partial y^\mu }
\right)
\otimes
\frac{\partial x^\sigma }
{\partial \bar{x}^\beta }
\frac{\partial }
{\partial x^\sigma } \\
&=&
(-1)^{
\alpha (\alpha +\nu +\sigma +\beta )
}
\frac{\partial \bar{y}^\nu }
{\partial x^\alpha }
\frac{\partial x^\sigma }
{\partial \bar{x}^\beta }
d^Gx^\alpha
\otimes
\frac{\partial }
{\partial x^\sigma } \\
&&
+(-1)^{
\mu (\mu +\nu +\sigma +\beta )
}
\frac{\partial \bar{y}^\nu }
{\partial y^\mu }
\frac{\partial x^\sigma }
{\partial \bar{x}^\beta }
d^Gy^\mu
\otimes
\frac{\partial }
{\partial x^\sigma }.
\end{eqnarray*}
By passing to coordinates in
$J_G^1(p)$ this tells us
the following:
\begin{equation}
\bar{y}_{\beta }^{\nu }
=(-1)^{\alpha (\nu +\beta )}
\frac{\partial \bar{y}^\nu }
{\partial x^\alpha }
\frac{\partial x^\alpha }
{\partial \bar{x}^\beta }
+(-1)^{
\mu
(\mu +\nu +\sigma +\beta )
}
\frac{\partial \bar{y}^\nu }
{\partial y^{\mu }}
\frac{\partial x^\sigma }
{\partial \bar{x}^\beta }
y_{\sigma }^{\mu }.
\label{ybar}
\end{equation}
With this expression in mind,
we are going to compute
the graded $1$-form
$d^G\bar{y}_\beta ^\nu $.
Initially, we should have
\begin{equation*}
d^G\bar{y}_\beta ^\nu
=d^Gx^\gamma
\frac{\partial
\bar{y}_\beta ^\nu }
{\partial x^\gamma }
+d^Gy^\mu
\frac{\partial
\bar{y}_\beta ^\nu }
{\partial y^\mu }
+d^Gy_\alpha ^\mu
\frac{\partial
\bar{y}_\beta ^\nu }
{\partial y_\alpha ^\mu },
\end{equation*}
so that we should consider
each term separately, but,
in fact, as we will compute
$\partial /
\partial \bar{y}_\beta ^\nu $
by applying duality, we need only
to compute the coefficient of
$d^Gy_\alpha ^\mu $,
which is given by \eqref{ybar}:
\begin{eqnarray*}
\frac{\partial
\bar{y}_\beta ^\nu }
{\partial
y_\alpha ^\mu }
&=&
\frac{\partial }
{\partial y_\alpha ^\mu }
\left\{
(-1)^{
\rho (\rho +\nu +\sigma +\beta )
}
\frac{\partial \bar{y}^\nu }
{\partial y^\rho }
\frac{\partial x^\sigma }
{\partial \bar{x}^\beta }
y_\sigma ^\rho
\right\}
\\
&=&
(-1)^{
\rho
(\rho +\nu +\sigma +\beta )
+(\mu +\alpha )
(\rho +\nu +\sigma +\beta )
}
\frac{\partial \bar{y}^\nu }
{\partial y^\rho }
\frac{\partial x^\sigma }
{\partial \bar{x}^\beta }
\delta _\mu ^\rho
\delta _\sigma ^\alpha \\
&=&
(-1)^{
\alpha
(\mu +\nu +\alpha +\beta )
}
\frac{\partial \bar{y}^\nu }
{\partial y^\mu }
\frac{\partial x^\alpha }
{\partial \bar{x}^\beta }.
\end{eqnarray*}
We can then write
\begin{equation}
d^G\bar{y}_\beta ^\nu
=d^Gx^\gamma
A_{\beta \gamma }^\nu
+d^Gy^\mu
B_{\beta \mu }^\nu
+(-1)^{
\alpha
(\mu +\nu +\alpha +\beta )
}
d^Gy_\alpha ^\mu
\frac{\partial \bar{y}^\nu }
{\partial y^\mu }
\frac{\partial x^\alpha }
{\partial \bar{x}^\beta },
\label{dybar}
\end{equation}
where
$A_{\beta \gamma }^\nu $,
$B_{\beta \mu }^\nu $
are coefficients whose explicit
expression is not needed.

We also remark
\begin{equation}
\left\{
\begin{array}{l}
d^G\bar{y}^\nu
=d^Gx^\alpha
\dfrac{\partial \bar{y}^\nu }
{\partial x^\alpha }
+d^Gy^\mu
\dfrac{\partial \bar{y}^\nu }
{\partial y^\mu },
\medskip \\
d^G\bar{x}^\beta
=d^Gx^\alpha
\dfrac{\partial \bar{x}^\beta }
{\partial x^\alpha }.
\end{array}
\right.
\label{dyandx}
\end{equation}
Next, we consider
the expression for
$\partial /
\partial \bar{y}_\gamma ^\mu $
as a tangent vector on $J_G^1(p)$.
Initially, we should have
\begin{equation*}
\frac{\partial }
{\partial \bar{y}_\gamma ^\mu }
=K^\sigma
\frac{\partial }
{\partial x^\sigma }
+L^\eta
\frac{\partial }
{\partial y^\eta }
+P_{\sigma \gamma }^{\rho \tau }
\frac{\partial }
{\partial y_\gamma ^\tau },
\end{equation*}
and we can compute the action
of the basic differentials
\eqref{dybar}, \eqref{dyandx}
on it.
This gives
\begin{align*}
0
&
=\left\langle
\dfrac{\partial }
{\partial
\bar{y}_\gamma ^\mu };
d^G\bar{x}^\beta
\right\rangle \\
&
=K^{\sigma}
\dfrac{\partial
\bar{x}^\beta }
{\partial x^\sigma },
\medskip \\
0
&
=\left\langle
\dfrac{\partial }
{\partial
\bar{y}_\gamma ^\mu };
d^G\bar{y}^\nu
\right\rangle \\
&
=L^\eta
\dfrac{\partial \bar{y}^\nu }
{\partial y^\eta },
\medskip \\
\delta _\nu ^\rho
\delta _\sigma ^\beta
&
=\left\langle
\dfrac{\partial }
{\partial
\bar{y}_\gamma ^\mu };
d^G\bar{y}_\beta ^\nu
\right\rangle \\
&
=(-1)^{
\alpha
(\mu +\nu +\alpha +\beta )
}
P_{\sigma \alpha }^{\rho \mu }
\dfrac{\partial \bar{y}^\nu }
{\partial y^\mu }
\dfrac{\partial x^\alpha }
{\partial\bar{x}^\beta },
\end{align*}
and from these equations,
we obtain
\begin{equation*}
\left\{
\begin{array}{l}
K^\gamma =0,\\
L^\mu =0,\\
P_{\sigma \alpha }^{\rho \mu }
=(-1)^{
\alpha
(\mu +\nu +\alpha +\beta )
}
\dfrac{\partial \bar{x}^\sigma }
{\partial x^\alpha }
\dfrac{\partial y^\mu }
{\partial \bar{y}^\rho }.
\end{array}
\right.
\end{equation*}
Hence, the law of transformation for
$\partial /
\partial \bar{y}_\sigma ^\rho $
is
\begin{equation*}
\frac{\partial }
{\partial \bar{y}_\sigma ^\rho }
=(-1)^{
\alpha
(\mu +\rho +\alpha +\sigma )
}
\dfrac{\partial \bar{x}^\sigma }
{\partial x^\alpha }
\dfrac{\partial y^\mu }
{\partial \bar{y}^\rho }
\frac{\partial }
{\partial y_\alpha ^\mu }.
\end{equation*}
This coincides with the law
of transformation for
$d^G\bar{x}^\sigma
\otimes
\partial /
\partial \bar{y}^\rho $.
Indeed,
\begin{eqnarray*}
d^G\bar{x}^\sigma
\otimes
\frac{\partial }
{\partial \bar{y}^\rho }
&=&
d^Gx^\alpha
\frac{\partial \bar{x}^\sigma }
{\partial x^\alpha }
\otimes
\frac{\partial y^\mu }
{\partial \bar{y}^\rho }
\frac{\partial }
{\partial y^{\mu }} \\
&=&
(-1)^{
\alpha
(\alpha +\sigma +\rho +\mu )
}
\frac{\partial \bar{x}^\sigma }
{\partial x^\alpha }
\frac{\partial y^\mu }
{\partial \bar{y}^\rho }
d^Gx^\alpha
\otimes
\frac{\partial }
{\partial y^\mu }.
\end{eqnarray*}
Thus, the isomorphism
in the statement of the proposition
is well defined.
\end{proof}

\subsection{Intrinsic construction of $\mathcal{J}$}

Let $p\colon (N,\mathcal{B})\rightarrow (M,\mathcal{A})$ be a graded
submersion. On the module of the graded vector fields on the graded $1-$jet
bundle $(J_{G}^{1}(p),\mathcal{A}_{J_{G}^{1}(p)})$, a $\mathcal{V}(p_{10})-$%
valued mapping acting upon $m$ arguments, is defined as follows: 
\begin{equation}
\mathcal{J} \left( D_{1},\dotsc ,D_{m} \right) =(-1)^{j+m}\iota _{D_{(\hat{%
\jmath})}} \left( p_{1}^{\ast } \left( d^{G}x^{1}\wedge \cdots \wedge
d^{G}x^{m} \right) \right) \otimes \theta (D_{j}),  \label{eq5_8}
\end{equation}
where $\iota _{D_{(\hat{\jmath})}} =\iota _{D_{m}}\circ \ldots \circ 
\widehat{\iota _{D_{j}}} \circ \ldots \circ \iota _{D_{1}}$ and 
\begin{equation}
\theta =\left( d^{G}y^{\mu } -d^{G}x^{\alpha }\cdot y_{\alpha }^{\mu }
\right) \otimes \frac{\partial }{\partial y^{\mu }},  \label{eq5_9}
\end{equation}
is intrinsically defined in \cite[Theorem 1.7]{Her-Mun 84a}, and $%
d^{G}x^{1}\wedge \cdots \wedge d^{G}x^{m}$ comes from a volume form $\eta $
on $M$, $\eta =dx^{1}\wedge \cdots \wedge dx^{m}$, so that $\mathcal{J}$ is
an intrinsic object.

\begin{proposition}
The operator $\mathcal{J}$ defined by \eqref{eq5_8} is a graded $m-$form.
\end{proposition}

\begin{proof}
First, multilinearity is a consequence of that of $\theta $ and the
properties of the insertion operator,
\begin{align*}
\iota _{\alpha \cdot D}\Lambda
& =\alpha \wedge \iota _{D}\Lambda , \\
\iota _{D_{1}+D_{2}}\Lambda
& =\iota _{D_{1}}\Lambda
+\iota _{D_{2}}\Lambda ,
\end{align*}
for all $D\in \mathcal{X}_{G}(J_{G}^{1}(p))$,
$\alpha \in \mathcal{A}_{J_{G}^{1}(p)}$,
$\Lambda \in \Omega _{G}^{1}(J_{G}^{1}(p))$. Second, we
have skew symmetry, which is rather obvious in view that
$\iota _{D_{1}}\iota _{D_{2}}
=-\iota _{D_{2}}\iota _{D_{1}}$.
\end{proof}

Now, we must check if the local expression for $J$ obtained from %
\eqref{eq5_8} gives the expression we want \eqref{jotacal}. Because of the
term $p_{1}^{\ast }(d^{G}x^{1} \wedge \cdots \wedge d^{G}x^{m})$, we just
need to evaluate 
\begin{equation*}
\mathcal{J} \left( \frac{\partial }{\partial x^{1}} ,\dotsc , \frac{\partial 
}{\partial x^{m}} \right) , \quad \mathcal{J} \left( \frac{\partial }{%
\partial x^{1}} ,\dotsc , \widehat{\frac{\partial }{\partial x^{i}}} ,\dotsc
, \frac{\partial }{\partial x^{m}}, \frac{\partial }{\partial x^{-j}}
\right) ,
\end{equation*}
and 
\begin{equation*}
\mathcal{J} \left( \frac{\partial }{\partial x^{1}},\dotsc , \widehat{\frac{%
\partial }{\partial x^{i}}} ,\dotsc , \frac{\partial }{\partial x^{m}}, 
\frac{\partial }{\partial y^{\nu }} \right) ,
\end{equation*}
where $i\in \{1,\dotsc ,m\}$, $\nu \in \{-s,\dotsc ,-1,1,\dotsc ,r\}$. Now,
we have 
\begin{align*}
& (-1)^{m-1}\mathcal{J} \left( \dfrac{\partial }{\partial x^{1}},\dotsc , 
\dfrac{\partial }{\partial x^{m}} \right) \\
& =(-1)^{j-1} \iota _{\frac{\partial }{\partial x^{m}}} \circ \ldots \circ 
\widehat{\iota _{\frac{\partial }{\partial x^{j}}}} \circ \ldots \circ \iota
_{\frac{\partial }{\partial x^{1}}} p_{1}^{\ast } \left( d^{G}x^{1}\wedge
\cdots \wedge d^{G}x^{m} \right) \otimes \theta \left( \frac{\partial }{%
\partial x^{j}} \right) \\
& =(-1)^{j-1} \iota _{\frac{\partial }{\partial x^{m}}} \circ \ldots \circ 
\widehat{\iota _{\frac{\partial }{\partial x^{j}}}} \circ \ldots \circ \iota
_{\frac{\partial }{\partial x^{1}}}p_{1}^{\ast } \left( d^{G}x^{1}\wedge
\cdots \wedge d^{G}x^{m} \right) \otimes \left( -y_{j}^{\mu } \frac{\partial 
}{\partial y^{\mu }} \right) \\
& =(-1)^{m}d^{G}x^{j}\otimes y_{j}^{\mu } \frac{\partial }{\partial y^{\mu }}
\\
& =(-1)^{m}d^{G}x^{j}\cdot y_{j}^{\mu } \otimes \frac{\partial }{\partial
y^{\mu }} \\
& \simeq (-1)^{m}y_{j}^{\mu } \frac{\partial }{\partial y_{j}^{\mu }},
\end{align*}
where the last identification comes from \ref{coordinates2}. Also, 
\begin{align*}
& (-1)^{m-1}\mathcal{J} \left( \frac{\partial }{\partial x^{1}}, \dotsc , 
\widehat{\frac{\partial }{\partial x^{i}}}, \dotsc ,\frac{\partial }{%
\partial x^{m}}, \frac{\partial }{\partial x^{-j}} \right) \\
& =(-1)^{m-1} \iota _{\frac{\partial }{\partial x^{m}}} \circ \ldots \circ 
\widehat{\iota _{\frac{\partial }{\partial x^{i}}}} \circ \ldots \circ \iota
_{\frac{\partial }{\partial x^{1}}}p_{1}^{\ast } \left( d^{G}x^{1}\wedge
\cdots \wedge d^{G}x^{m} \right) \otimes \theta \left( \frac{\partial }{%
\partial x^{-j}} \right) \\
& =(-1)^{m-1}(-1)^{m-i}d^{G}x^{i}\otimes \theta \left( \frac{\partial }{%
\partial x^{-j}} \right) \\
& =-(-1)^{i}d^{G}x^{i}\otimes \left( -y_{-j}^{\mu }\frac{\partial }{\partial
y^{\mu }} \right) \\
& =(-1)^{i}d^{G}x^{i}\cdot y_{-j}^{\mu } \otimes frac{\partial }{\partial
y^{\mu }} \\
& =(-1)^{i}y_{-j}^{\mu }d^{G}x^{i} \otimes \frac{\partial }{\partial y^{\mu }%
} \\
& \simeq (-1)^{i}y_{-j}^{\mu } \frac{\partial }{\partial y_{i}^{\mu }}.
\end{align*}
Moreover, by noting that each term $\iota _{\frac{\partial }{\partial
y^{\gamma }}} p_{1}^{\ast }(d^{G}x^{1}\wedge \cdots \wedge d^{G}x^{m})$
vanishes, we obtain 
\begin{align*}
& (-1)^{m-1}\mathcal{J} \left( \frac{\partial }{\partial x^{1}}, \dotsc , 
\widehat{\frac{\partial }{\partial x^{i}}}, \dotsc ,\frac{\partial }{%
\partial x^{m}}, \frac{\partial }{\partial y^{\nu }} \right) \\
& =(-1)^{m-1}\iota _{\frac{\partial }{\partial x^{m}}} \circ \ldots \circ 
\widehat{\iota _{\frac{\partial }{\partial x^{i}}}} \circ \ldots \circ \iota
_{\frac{\partial }{\partial x^{1}}}p_{1}^{\ast } \left( d^{G}x^{1}\wedge
\ldots \wedge d^{G}x^{m} \right) \otimes \theta \left( \frac{\partial }{%
\partial y^{\nu }} \right) \\
& =(-1)^{i-1}d^{G}x^{i}\otimes \frac{\partial }{\partial y^{\nu }} \\
& \simeq (-1)^{i-1} \frac{\partial }{\partial y_{i}^{\nu }}.
\end{align*}
We thus conclude that the local expression for $\mathcal{J}$ is 
\begin{align*}
(-1)^{m-1}\mathcal{J}& =(-1)^{m}d^{G}x^{1}\wedge \ldots \wedge
d^{G}x^{m}\otimes y_{i}^{\mu } \frac{\partial }{\partial y_{i}^{\mu }} \\
& \\
& +(-1)^{i}d^{G}x^{1}\wedge \ldots \wedge \widehat{d^{G}x^{i}}\wedge \ldots
\wedge d^{G}x^{m}\wedge d^{G}x^{-j} \otimes y_{-j}^{\mu } \frac{\partial }{%
\partial y_{i}^{\mu }} \\
& \\
& +(-1)^{i-1}d^{G}x^{1}\wedge \ldots \wedge \widehat{d^{G}x^{i}}\wedge
\ldots \wedge d^{G}x^{m}\wedge d^{G}y^{\nu } \otimes \frac{\partial }{%
\partial y_{i}^{\nu }} \\
& \\
& =(-1)^{i}d^{G}x^{1}\wedge \ldots \wedge \widehat{d^{G}x^{i}}\wedge \ldots
\wedge d^{G}x^{m}\wedge d^{G}x^{i} \cdot y_{i}^{\mu }\otimes \frac{\partial 
}{\partial y_{i}^{\mu }} \\
& \\
& +(-1)^{i}d^{G}x^{1}\wedge \ldots \wedge \widehat{d^{G}x^{i}}\wedge \ldots
\wedge d^{G}x^{m}\wedge d^{G}x^{-j} \cdot y_{-j}^{\mu }\otimes \frac{%
\partial }{\partial y_{i}^{\mu }} \\
& \\
& +(-1)^{i-1}d^{G}x^{1}\wedge \ldots \wedge \widehat{d^{G}x^{i}}\wedge
\ldots \wedge d^{G}x^{m}\wedge d^{G}y^{\mu } \otimes \frac{\partial }{%
\partial y_{i}^{\mu }} \\
& \\
& =(-1)^{i}d^{G}x^{1}\wedge \ldots \wedge \widehat{d^{G}x^{i}}\wedge \ldots
\wedge d^{G}x^{m}\wedge \left( d^{G}x^{\alpha }\cdot y_{\alpha }^{\mu }
-d^{G}y^{\mu } \right) \otimes \frac{\partial }{\partial y_{i}^{\mu }}.
\end{align*}
Hence, 
\begin{equation*}
\mathcal{J}=(-1)^{m-1} \iota _{\frac{\partial }{\partial x^{i}}} \eta
^{G}\wedge \theta ^{\mu }\otimes \frac{\partial }{\partial y_{i}^{\mu }},
\end{equation*}
where $\theta ^{\mu }=d^{G}y^{\mu }-d^{G}x^{\alpha }\cdot y_{\alpha }^{\mu }$
is the horizontal differential of $y^{\mu }$, and this is precisely %
\eqref{jotacal} written in a more compact form.

In this way, we have constructed a canonical $\mathcal{V}(p_{10})-$valued $m-
$form $\mathcal{J}$ for any graded submersion $p\colon (N,\mathcal{B})\to (M,%
\mathcal{A})$. This is appropriate for the case of graded mechanics, but if
we want to study graded fields, we must go on to higher$-$order jet bundles;
let us see how to extend the previous construction to $J_{G}^{k}((N,\mathcal{%
B}),(M,\mathcal{A})) \equiv (J_{G}^{k}(p),\mathcal{A}_{J_{G}^{k}(p)})$ for
any $k$.

\subsection{Intrinsic construction of $\mathcal{J}_{k}$}

Consider the following submersion playing the r\^{o}le of $p$ in previous
sections: $p_{k-1}\colon J_{G}^{k-1}(p)\to (M,\mathcal{A})$. Then, the
preceding construction tells us that we have defined a $\mathcal{J}$ on the
graded bundle $J_{G}^{1}((J_{G}^{k-1}(p),\mathcal{A}_{J_{G}^{k-1}(p)}), (M,%
\mathcal{A}))\doteq J_{G}^{1}(p_{k-1})$, which is a graded $m-$form with
values on $\mathcal{V}((p_{1}^{k-1})_{10})$ that will be denoted $\mathcal{J}%
_{k}$ (here, $p_{1}^{k-1}$ is defined by $p_{1}^{k-1}\colon
J_{G}^{1}(p_{k-1})\rightarrow (M,\mathcal{A})$, and $(p_{1}^{k-1})_{10}%
\colon J_{G}^{1}(p_{1}^{k-1})\to J_{G}^{1}(p_{k-1})$ the target projection).

If $(x^{\alpha },y^{\mu },z_{Q}^{\mu })$, $1\leq |Q|\leq k-1$, is a system
of coordinates for $J_{G}^{k-1}(p)$, then $(x^{\alpha },y^{\mu },z_{Q}^{\mu
},w_{R}^{\mu })$ (with $1\leq |R|\leq k$, $1\leq |Q|\leq k-1$, recall that $%
Q,R$ denote arbitrary multi$-$indices) is a system for $J_{G}^{1}(p_{k-1})$,
and we have the local expression (with $1\leq i\leq m$ a \emph{positive}
index) 
\begin{equation}
\mathcal{J}_{k}=(-1)^{m-1} \iota _{\frac{\partial }{\partial x^{i}}} \eta
^{G}\wedge \left( \theta ^{y^{\mu }}\otimes \frac{\partial }{\partial
y_{i}^{\mu }} +\theta ^{z_{Q}^{\mu }}\otimes \frac{\partial }{\partial
z_{i+Q}^{\mu }} +\theta ^{w_{R}^{\mu }}\otimes \frac{\partial }{\partial
w_{i+R}^{\mu }} \right)  \label{eq5_10}
\end{equation}
(it must be noted that we are using the canonical identification %
\eqref{coordinates2} in writing $\frac{\partial }{\partial w_{i+R}^{\mu }}$,
also, note that the sum $\iota _{\frac{\partial }{\partial x^{i}}}\eta
^{G}\wedge \theta ^{w_{R}^{\mu }}\otimes \frac{\partial }{\partial
w_{i+R}^{\mu }}$ only runs up to $|R|=k-1$), where $\theta ^{z_{Q}^{\mu
}}=d^{G}z_{Q}^{\mu }-d^{G}x^{\alpha }\cdot z_{\alpha \star Q}^{\mu }$ and so
on. Now, we observe that there exists a canonical graded immersion $%
J_{G}^{k}(p)\overset{\Psi }{\hookrightarrow }J_{G}^{1}(p_{k-1})$, which
expressed through its action on coordinates, reads (here, $(x^{\alpha
},y^{\mu },v_{Q}^{\mu })$, $1\leq |Q|\leq k$, is a system of coordinates for 
$J_{G}^{k}(p)$): 
\begin{equation}
\left. 
\begin{array}{rl}
\Psi ^{\ast }(x^{\alpha })= & \!\!\!x^{\alpha } \\ 
\Psi ^{\ast }(y^{\mu })= & \!\!\!y^{\mu } \\ 
\Psi ^{\ast }(z_{Q}^{\mu })= & \!\!\!v_{Q}^{\mu } \\ 
\Psi ^{\ast }(w_{R}^{\mu })= & \!\!\!v_{R}^{\mu }%
\end{array}
\right\}  \label{eq5_11}
\end{equation}
Of course, when acting upon jet extensions of sections $\sigma $, this
action reads 
\begin{equation*}
\Psi ^{\ast }(j^{1}(j^{k-1}(\sigma )))=j^{k}(\sigma ).
\end{equation*}
Now, it is clear that (as $\Psi ^{\ast }$ commutes with $d^{G}$), 
\begin{align*}
\Psi ^{\ast } \left( \iota _{\frac{\partial }{\partial x^{i}}}\eta ^{G}
\right) & =(-1)^{i-1}\Psi ^{\ast } \left( d^{G}x^{1}\wedge \cdots \wedge 
\widehat{d^{G}x^{i}}\wedge \cdots \wedge d^{G}x^{m} \right) \\
& =\iota _{\frac{\partial }{\partial x^{i}}}\eta ^{G}, \\
\Psi ^{\ast } \left( \theta ^{y^{\mu }} \right) & =\Psi ^{\ast } \left(
d^{G}y^{\mu }-d^{G}x^{\alpha }\cdot y_{\alpha }^{\mu } \right) \\
& =d^{G}y^{\mu }-d^{G}x^{\alpha }\cdot y_{\alpha }^{\mu } \\
& =\theta ^{\mu }, \\
\Psi ^{\ast } \left( \theta ^{z_{Q}^{\mu }} \right) & =\Psi ^{\ast } \left(
d^{G}z_{Q}^{\mu }-d^{G}x^{\alpha }\cdot z_{\alpha \star Q}^{\mu } \right) \\
& =d^{G}y_{Q}^{\mu }-d^{G}x^{\alpha }\cdot y_{\alpha \star Q}^{\mu } \\
& =\theta _{Q}^{\mu },
\end{align*}
and so on. We can apply $\Psi ^{\ast }$ to \eqref{eq5_10} to obtain a graded 
$m-$form on $J_{G}^{k}(p)$; according to the preceding observations, the
only terms that represent some problem are those duplicated in $\partial
/\partial z_{i+Q}^{\mu }$ and $\partial /\partial w_{i+R}^{\mu }$ (see %
\eqref{eq5_11}). But these terms are precisely the ones coming from a single
supervector on $J_{G}^{k}(p)$ through \eqref{eq5_11}; to be more precise,
let us study $\Psi _{\ast }(\partial /\partial v_{i+Q}^{\mu })$. We would
like to see that 
\begin{equation*}
\Psi _{\ast } \left( \frac{\partial }{\partial v_{i+Q}^{\mu }} \right) =%
\frac{\partial }{\partial z_{i+Q}^{\mu }} +\frac{\partial }{\partial
w_{i+Q}^{\mu }},
\end{equation*}
as an extreme case we have $|Q|=k$, but then this reduces to 
\begin{equation*}
\Psi _{\ast } \left( \frac{\partial }{\partial v_{i+Q}^{\mu }} \right) =%
\frac{\partial }{\partial w_{i+Q}^{\mu }},
\end{equation*}
and as we have the canonical identification \eqref{coordinates2}, what we
really want is to prove 
\begin{equation*}
\Psi _{\ast } \left( \frac{\partial }{\partial v_{Q}^{\mu }} \right) =\frac{%
\partial }{\partial z_{Q}^{\mu }} +\frac{\partial }{\partial w_{Q}^{\mu }},
\end{equation*}
for an arbitrary multiindex $Q$.

Thus, consider the action of 
\begin{equation*}
\Psi _{\ast } \left( \frac{\partial }{\partial v_{Q}^{\mu }} \right) .
\end{equation*}
We have 
\begin{align*}
\Psi _{\ast } \left( \frac{\partial }{\partial v_{Q}^{\mu }} \right)
(z_{R}^{\nu })& =\frac{\partial }{\partial v_{Q}^{\mu }} \left( \Psi ^{\ast
}(z_{R}^{\nu }) \right) \\
& =\frac{\partial }{\partial v_{Q}^{\mu }}v_{R}^{\nu } \\
& =\delta _{\mu }^{\nu }\delta _{Q}^{R}, \\
& \\
\Psi _{\ast } \left( \frac{\partial }{\partial v_{Q}^{\mu }} \right)
(w_{R}^{\nu }) & =\frac{\partial }{\partial v_{Q}^{\mu }} \left( \Psi ^{\ast
}(w_{R}^{\nu }) \right) \\
& =\frac{\partial }{\partial v_{Q}^{\mu }}v_{R}^{\nu } \\
& =\delta _{\mu }^{\nu }\delta _{Q}^{R},
\end{align*}
and this is precisely the action of $\partial /\partial z_{Q}^{\mu
}+\partial /\partial w_{Q}^{\mu }$, as wanted.

As a consequence, we have the following result (see \cite{Ald-Azc 78a,
Ald-Azc 78b} for a classical version):

\begin{theorem}
On $J_{G}^{k}(p)$ (for any $k$) there is defined a canonical graded $m-$form
with values on $\mathcal{V}((p_{k})_{10})\subset \mathcal{V}%
((p_{1}^{k-1})_{10})$, which we denote by $\mathcal{J}_{k}$, and whose local
expression is 
\begin{equation*}
\mathcal{J}_{k} =(-1)^{m-1}\iota _{\frac{\partial }{\partial x^{i}}} \eta
^{G}\wedge \theta _{Q}^{\mu }\otimes \frac{\partial }{\partial y_{i+Q}^{\mu }%
} (1\leq i\leq m=\dim M),
\end{equation*}
being $0\leq |Q|\leq k-1$, with the usual convention $\theta _{Q}^{\mu
}=\theta ^{\mu }$ when $|Q|=0$.
\end{theorem}

\begin{remark}
In the statement of the theorem, we are writing collectively $\theta
_{Q}^{\mu }$ instead of $\theta ^{z_{Q}^{\mu }}$ and $\theta ^{w_{R}^{\mu }}$
(it is a shorthand for \eqref{eq5_10}).
\end{remark}

Generalizing the classical expression (see, for instance \cite{Sau 89},
Theorem $5.5.2$), for any $L\in \mathcal{A}_{J_{G}^{k}(p)}$, we define the
graded $m-$form (the so called \emph{Poincar\'e$-$Cartan form of order} $k$) 
\begin{equation*}
\tilde{\Theta}^{L}=\mathcal{L}_{\mathcal{J}_{k}}^{G}(L)+\eta ^{G}\cdot L.
\end{equation*}

Let us make a remark. Let $\mathcal{A}\overset{\sim }{\rightarrow }C^{\infty
}(M)$ be the structural morphism and $C^{\infty }(M)\overset{\sigma }{%
\rightarrow }\mathcal{A}$ a global section of it. Then, to every volume form 
$\eta $ on $M$ we can associate a graded volume form $\eta ^{G}=\sigma (\eta
)$ on $(M,\mathcal{A})$. On the other hand, note that a graded Lagrangian
density is an $m-$form of the type 
\begin{equation*}
\eta ^{G}\cdot L, L\in \mathcal{A}_{J_{G}^{k}(p)}.
\end{equation*}
Thus, if we change the volume form $\eta $ on $M$ to a form $\mu =\eta \cdot
f$ where $f$ is a differentiable function on $M$, $f\in C^{\infty }(M)$, we
will have a induced change in the Lagrangian: 
\begin{equation*}
\eta ^{G}\cdot f\cdot L.
\end{equation*}
Moreover, recall that from the local expression of the $\mathcal{J}_{k}$
morphism \eqref{eq5_10} it is clear that replacing $\eta $ for $\mu $
amounts to passing from $\mathcal{J}_{k}$ to $f\cdot \mathcal{J}_{k}$.
Putting these observations together we get (introducing temporarily an
obvious notation to distinguish which graded volume form is in use): 
\begin{align*}
\tilde{\Theta}_{\mu }^{L}& =\mathcal{L}_{\mathcal{J}_{k}^{\mu }}^{G}(L) +\mu
^{G}\cdot L \\
& =\mathcal{L}_{f\cdot \mathcal{J}_{k}^{\eta }}^{G}(L) +\eta ^{G}\cdot
f\cdot L \\
& =\mathcal{L}_{\mathcal{J}_{k}^{\eta }}^{G}(f\cdot L) +\eta ^{G}(f\cdot L),
\end{align*}
where in the last step use has been made of the fact that $f\in C^{\infty
}(M)$ is not affected by the derivative on the fiber coordinates, carried on
by $\mathcal{L}_{\mathcal{J}_{k}^{\mu }}^{G}$. If we denote $f\cdot L\in 
\mathcal{A}_{J_{G}^{k}(p)}$ by $L_{f}$, what we have obtained is 
\begin{equation*}
\tilde{\Theta}_{\mu }^{L}=\tilde{\Theta}_{\eta }^{L_{f}},
\end{equation*}
so the graded Poincar\'{e}$-$Cartan form $\tilde{\Theta}^{L}$ is well
behaved under the decomposition \textquotedblleft graded Lie derivative plus
graded Lagrangian density\textquotedblright .

\section{Equivalence between Graded and Berezinian variational problems 
\label{equiv}}

Let us make some remarks about the correspondence between Berezinian and
graded variational problems. It is well known how to obtain the equations of
the solutions to a graded variational problem (see \cite{Her-Mun 84a}); on
the contrary, for Berezinian problems an intrinsic formulation in Cartan's
spirit has been not available up until now. What does exist, is a way (based
on the Comparison Theorem) to associate to each graded problem a Berezinian
one and to establish a correspondence between their solutions. The basic
idea is as follows: given a lagrangian $L\in \mathcal{A}_{J_{G}^{1}(p)}$,
let $\xi _{L}$ be the first$-$order Berezinian density that it determines,
which is given by 
\begin{equation*}
\left[ d^{G}x^{1}\wedge \cdots \wedge d^{G}x^{m}\otimes \frac{d}{dx^{-1}}
\circ \ldots \circ \frac{d}{dx^{-n}}\right] \cdot L,
\end{equation*}
and let 
\begin{equation*}
\lambda _{\xi _{L}}=d^{G}x^{1}\wedge \cdots \wedge d^{G}x^{m}\cdot \frac{%
d^{n}L}{dx^{-1}\ldots dx^{-n}}
\end{equation*}
be the corresponding graded Lagrangian density. In \cite{Her-Mun 84a}, to
each first order graded Lagrangian density $\lambda $ a canonical graded $m$%
-form is associated, the graded Poincar\'{e}$-$Cartan form for the
Lagrangian density $\lambda _{\xi _{L}}$. Here, we denote by $\Theta
_{0}^{L} $ the graded Poincar\'{e}$-$Cartan form corresponding to $-\lambda
_{L}$; in local coordinates, it is given by the expression (\ref{cartanlocal}%
)\ and, as we have proved, it can be constructed as an intrinsic object.
Now, as 
\begin{equation*}
\lambda _{\xi _{L}} =\mathcal{L}_{\frac{d}{dx^{-1}}}^{G}\circ \ldots \circ 
\mathcal{L}_{\frac{d}{dx^{-n}}}^{G} \left( d^{G}x^{1}\wedge \cdots \wedge
d^{G}x^{m}\cdot L \right) ,
\end{equation*}
it is natural to consider the graded $m$-form 
\begin{equation}
\Theta ^{L} =\mathcal{L}_{\frac{d}{dx^{-1}}}^{G} \circ \ldots \circ \mathcal{%
L}_{\frac{d}{dx^{-n}}}^{G}\Theta _{0}^{L}  \label{berdens}
\end{equation}
as the graded Poincar\'{e}$-$Cartan form for the Berezinian density $\xi _{L}
$. But we could as well follow other way to define $\Theta ^{L}$: Instead of
taking the first$-$order Lagrangian density $d^{G}x^{1}\wedge \ldots \wedge
d^{G}x^{m}\cdot L$, construct $\Theta _{0}^{L}=\mathcal{L}_{\mathcal{J}%
}^{G}(L)+\eta ^{G}\cdot L$ and then apply 
\begin{equation*}
\mathcal{L}_{\frac{d}{dx^{-1}}}^{G}\circ \ldots \circ \mathcal{L}_{\frac{d}{%
dx^{-n}}}^{G},
\end{equation*}
we could have considered the Lagrangian density, of order $(n+1)$, $\frac{%
d^{n}L}{dx^{-1}\ldots dx^{-n}}$ and apply $\mathcal{L}_{\mathcal{J}%
_{n+1}}^{G}$ to obtain 
\begin{equation*}
\tilde{\Theta}^{L} =\mathcal{L}_{\mathcal{J}_{n+1}}^{G} \left( \frac{d^{n}L}{%
dx^{-1}\cdots dx^{-n}} \right) +\eta ^{G}\cdot \frac{d^{n}L}{dx^{-1}\cdots
dx^{-n}}.
\end{equation*}

The first procedure is designed to take benefit of the graded variational
calculus, but has the handicap of presenting an expression like %
\eqref{berdens}), with the factors 
\begin{equation*}
\mathcal{L}_{\frac{d}{dx^{-1}}}^{G}\circ \ldots \circ \mathcal{L}_{\frac{d}{%
dx^{-n}}}^{G}
\end{equation*}
destroying, at a first glance, covariance. On the other hand, once a volume
form $\eta =dx^{1}\wedge \cdots \wedge dx^{m}$ has been fixed on the base,
the second one proceeds intrinsically to obtain 
\begin{equation*}
\frac{d^{n}L}{dx^{-1}\cdots dx^{-n}}
\end{equation*}
from the Berezinian density $\xi _{L}$ and then $\tilde{\Theta}^{L}$, thus
developing a Cartan formalism in an analogous manner to the classical
tangent bundle formulation. Nevertheless, it would be desirable the
convergence of the two ways, in the sense that $\Theta ^{L}=$ $\tilde{\Theta}%
^{L}$ for any $L\in \mathcal{A}_{J_{G}^{1}(p)}$; indeed, this is the case as
we will see in Theorem \ref{maintheorem}.

We will need some notations and technical lemmas that also will be useful
later.

\subsection{Preliminaries\label{lemmata}}

Let $B\in (\mathbb{Z}^{-})^{k}$ be a strictly decreasing multi-index. For
every $b\in B$, we define $p(b)$, $q(b)$ as follows: 
\begin{align*}
p(b)& =((\text{position of $b$ in }B)-1)\mathop{\rm {}mod}\nolimits2, \\
q(b)& =(\text{position of $b$ in }B)\mathop{\rm {}mod}\nolimits2.
\end{align*}
For example, if $B=(-1,-5,-7,-8)$, then $p(-7)=0$, $q(-7)=1$. The symbol $%
B-\{b\}$ denotes the $(|B|-1)$-multi-index obtained by removing $b$ from $B$%
; e.g., in the previous example we have $B-\{ -5\} =(-1,-7,-8)$. We also set 
$|Z|_{2}=|Z|\mathop{\rm {}mod}\nolimits2$, for every multi-index $Z$. For
any pair of multi-indices 
\begin{equation*}
Q=(i_{1},\dotsc ,i_{|Q|})\in \mathbb{Z}^{|Q|}, \quad
B=(b_{1},\dotsc,b_{|B|})\in (\mathbb{Z}^{-})^{|B|},
\end{equation*}
such that $|B|\geq |Q|$, we define $\varphi (Q,B)$ as follows: 
\begin{equation*}
\varphi (Q,B)=\sum\limits_{k=1}^{|I|}i_{k}\varphi _{k}(b_{k}),
\end{equation*}
where 
\begin{equation*}
\begin{array}{cc}
\varphi _{k}(b)= & \left\{ 
\begin{array}{c}
p(b),\,\text{if }k\equiv 1\mathop{\rm {}mod}\nolimits2 \\ 
q(b),\,\text{if }k\equiv 0\mathop{\rm {}mod}\nolimits2%
\end{array}
\right.%
\end{array}%
\end{equation*}
and $\varphi (Q,B)=0$, if $|Q|=0$.

Finally, as usual, the symbol $\star $, applied to a pair of multi-indices,
means juxtaposition.

In what follows, we denote by $\frac{d}{dx^\alpha }$ the graded horizontal
lift of $\partial /\partial x^\alpha $ to $J_G^\infty (p)$, whose local
expression is given in the formula \eqref{XH}.

\begin{lemma}
For any strictly decreasing multi-index $B\in (\mathbb{Z}^-)^k$, we have 
\begin{equation*}
\left[ \frac{\partial }{\partial y^\mu }, \frac{d^{|B|}}{dx^B} \right] =0,
\end{equation*}
when acting on $\mathcal{A}_{J_G^r(p)}$.
\end{lemma}

\begin{lemma}
Let $k$ be a positive integer. Given $i_{0}\in \mathbb{Z}$ and $j\in
\{1,\dotsc ,n\}$, for every $Q\in \mathbb{Z}^{k}$, we have 
\begin{equation*}
\left[ \frac{\partial }{\partial y_{\{i_{0}\}\star Q}^{\mu }}, \frac{d}{%
dx^{-j}} \right] =\delta _{i_{0}}^{-j}\frac{\partial }{\partial y_{Q}^{\mu }}%
,
\end{equation*}
both sides acting on $\mathcal{A}_{J_{G}^{r}}(p)$.
\end{lemma}

Note $\partial /\partial y_{\{i_{0}\}\star Q}^{\mu }$ vanishes on $\mathcal{A%
}_{J_{G}^{r}}(p)$ whenever $|Q|>r$. In particular, for every $L\in \mathcal{A%
}_{J_{G}^{1}(p)}$ we have 
\begin{equation*}
\frac{\partial }{\partial y_{\alpha }^{\mu }} \left( \frac{dL}{dx^{-j}}
\right) =\delta _{\alpha }^{-j}\frac{\partial L}{\partial y^{\mu }}, \quad 
\frac{\partial }{\partial y_{\alpha \beta }^{\mu }} \left( \frac{dL}{dx^{-j}}
\right) =\delta _{\alpha }^{-j} \frac{\partial L}{\partial y_{\beta }^{\mu }}%
,
\end{equation*}
but 
\begin{equation*}
\frac{\partial }{\partial y_{Q}^{\mu }} \left( \frac{dL}{dx^{-j}} \right)
=0,\; \text{for }|Q|>2.
\end{equation*}

\begin{lemma}
For any strictly decreasing multi-index $B\in (\mathbb{Z}^{-})^{k}$, we have 
\begin{equation*}
\left[ \frac{\partial }{\partial y_{\alpha }^{\mu }}, \frac{d^{|B|}}{dx^{B}} %
\right] =\sum\limits_{b\in B}(-1)^{\mu (|B|_{2}+1) +\alpha \cdot p(b)}\delta
_{b}^{\alpha } \frac{d^{|B|-1}}{dx^{B-\{b\}}} \frac{\partial }{\partial
y^{\mu }},
\end{equation*}
when acting on $\mathcal{A}_{J_{G}^{r}(p)}$, where it is assumed 
\begin{equation*}
\frac{d^{0}F}{dx^{\emptyset }}=F, \quad \forall F\in \mathcal{A}%
_{J_{G}^{r}(p)}.
\end{equation*}
\end{lemma}

In particular, we have 
\begin{equation*}
\frac{\partial }{\partial y_i^\mu } \frac {d^{|B|}F}{dx^B} =(-1)^{\mu |B|_2} 
\frac{d^{|B|}}{dx^B} \frac{\partial }{\partial y_i^\mu }, \quad \forall i>0.
\end{equation*}

\begin{proposition}
For every $L\in \mathcal{A}_{J_{G}^{1}(p)}$, every strictly decreasing $B\in
(\mathbb{Z}^{-})^{k}$, and every $Q\in \mathbb{Z}^{r}$ such that $k\geq 2$, $%
1\leq r\leq k$, we have 
\begin{multline*}
\frac{\partial }{\partial y_{\{i\}\star Q}^{\mu }} \frac{d^{|B|}L}{dx^{B}} \\
={\sum\limits_{\underset{b_{c_{1}},\dotsc ,b_{c_{|Q|}} \in B}{%
-b_{c_{1}}>\ldots >-b_{c_{|Q|}}}}} (-1)^{\mu (|B|_{2}+|Q|_{2}) +\varphi
(Q,B)}\delta _{b_{c_{1}}}^{i_{1}} \cdots \delta _{b_{c_{|Q|}}}^{i_{|Q|}} 
\frac{d^{|B|-|Q|}}{dx^{B-\{b_{c_{1}}, \dotsc ,b_{c_{|Q|}}\}}} \frac{\partial
L}{\partial y_{i}^{\mu }}.
\end{multline*}
\end{proposition}

The proof of these results is a lengthy induction, but only involving
standard computations.

\subsection{The main theorem}

In this subsection, as announced in the Introduction, we study the
equivalence between first$-$order Berezinian variational problems and higher$%
-$order graded variational problems. As the computations are rather
cumbersome, we will illustrate the general situation by considering the case 
$n=2$ (that is, a base manifold of graded dimension $(m|2)$).

\begin{theorem}
\label{maintheorem} Let $\xi _{L}$ be a first$-$order Berezinian density, 
\begin{equation}
\xi _{L} =\left[ d^{G}x^{1}\wedge \cdots \wedge d^{G}x^{m}\otimes \frac{d}{%
dx^{-1}}\circ \ldots \circ \frac{d}{dx^{-n}} \right] \cdot L,\quad L\in 
\mathcal{A}_{J_{G}^{1}(p)},  \label{berezlocal}
\end{equation}
and let 
\begin{equation*}
\lambda _{\xi _{L}} =d^{G}x^{1}\wedge \ldots \wedge d^{G}x^{m} \frac{d^{n}L}{%
dx^{-1}\ldots dx^{-n}}.
\end{equation*}
Let $\Theta _{0}^{L}$ be the graded Poincar\'e$-$Cartan form corresponding
to $-\lambda _{\xi _{L}}$, and let us set 
\begin{equation*}
\Theta ^{L} =\mathcal{L}_{\frac{d}{dx^{-1}}}^{G}\circ \ldots \circ \mathcal{L%
}_{\frac{d}{dx^{-n}}}^{G}\Theta _{0}^{L}
\end{equation*}
and 
\begin{equation*}
\tilde{\Theta}^{L}=\mathcal{L}_{\mathcal{J}_{n+1}}^{G} \left( \frac{d^{n}L}{%
dx^{-1}\ldots dx^{-n}} \right) +\eta ^{G}\cdot \frac{d^{n}L}{dx^{-1}\ldots
dx^{-n}}.
\end{equation*}
Then, we have 
\begin{equation*}
\Theta ^{L}=\tilde{\Theta}^{L}.
\end{equation*}
\end{theorem}

\begin{proof}
Let $T$ be the totally odd multi-index $T=(-1,\dotsc ,-n)$,
so that $|T|=n$.
We also write
$\varepsilon =(-1)^{\mu (|B|_{2}+|Q|_{2})+\varphi (Q,B)}$.
By applying the preceding proposition, we obtain
\begin{align*}
\tilde{\Theta}^{L}-\eta ^{G}\cdot \frac{d^{n}L}{dx^{T}}
&
=\mathcal{L}_{\mathcal{J}_{n+1}}^{G}
\left( \
frac{d^{n}L}{dx^{T}}
\right)  \\
&
=\sum\limits_{0\leq |Q|\leq n}(-1)^{m-1}
\iota _{\frac{\partial }{\partial x^{i}}}
\eta ^{G}\wedge \theta _{Q}^{\mu }
\cdot
\frac{\partial }{\partial y_{i+Q}^{\mu }}
\frac{d^{n}L}{dx^{T}} \\
& \\
&
=\sum\limits_{0\leq |Q|\leq n}(-1)^{m-1}
\iota _{\frac{\partial }{\partial x^{i}}}
\eta ^{G}\wedge \theta _{Q}^{\mu } \\
& \quad \qquad \cdot
\left(
{\sum\limits_{\underset{b_{c_{1}},
\dotsc,b_{c_{|Q|}}\in B}{-b_{c_{1}}
>\ldots >-b_{c_{|Q|}}}}}
\varepsilon \delta _{b_{c_{1}}}^{i_{1}}
\cdots
\delta _{b_{c_{|Q|}}}^{i_{|Q|}}
\frac{d^{n-|Q|}}{dx^{T-\{b_{c_{1}},
\dotsc ,b_{c_{|Q|}}\}}}
\frac{\partial L}{\partial y_{i}^{\mu }}
\right)  \\
& \\
&
=\sum\limits_{0\leq |Q|\leq n}(-1)^{m-1}
\varepsilon \cdot
\iota _{\frac{\partial }{\partial x^{i}}}
\eta ^{G}\wedge \theta _{i_{1}\ldots
i_{|Q|}}^{\mu }\frac{d^{n-|Q|}}{dx^{T-\{i_{1},
\dotsc ,i_{|Q|}\}}}
\frac{\partial L}{\partial y_{i}^{\mu }} \\
&
=\mathcal{L}_{\frac{d}{dx^{-1}}}^{G}
\circ \ldots \circ \mathcal{L}_{\frac{d}{dx^{-n}}}^{G}
\left(
(-1)^{m-1}\iota _{\frac{\partial }{\partial x^{i}}}
\eta ^{G}\wedge \theta ^{\mu }\otimes
\frac{\partial L}{\partial y_{i}^{\mu }}
\right)  \\
& =\Theta ^{L}-\eta ^{G}\cdot \frac{d^{n}L}{dx^{T}}.
\end{align*}
\end{proof}

\subsection{$(m|2)-$superfield theory}

As the use and notations for multi-indices are rather cumbersome, let us
analyze a specific case in detail, that of supermanifold with $m$ even and $2
$ odd coordinates. We start with a Berezinian density 
\begin{equation*}
\xi _{L} =\left[ d^{G}x^{1}\wedge \ldots \wedge d^{G}x^{m}\otimes \frac{d}{%
dx^{-1}}\circ \frac{d}{dx^{-2}} \right] \cdot L,
\end{equation*}
where $L\in \mathcal{A}_{J_{G}^{1}(p)}$; i.e., $L=L(x^{\alpha },y^{\mu
},y_{\alpha }^{\mu })$. The associated graded Lagrangian density is 
\begin{equation*}
d^{G}x^{1}\wedge \ldots \wedge d^{G}x^{m}\cdot \frac{d^{2}L}{dx^{-1}dx^{-2}}.
\end{equation*}
Next, from $L$ we can obtain, by applying 
\begin{equation*}
\mathcal{J}_{1}=(-1)^{m-1}\iota _{\frac{\partial }{\partial x^{i}}} \eta
^{G}\wedge \theta ^{\mu }\otimes \frac{\partial }{\partial y_{i}^{\mu }}%
(1\leq i\leq m=\dim M),
\end{equation*}
the graded form 
\begin{equation}
\begin{array}{l}
\Theta ^{L}-\eta ^{G}\cdot \dfrac{d^{2}L}{dx^{-1}dx^{-2}} =\mathcal{L}_{%
\frac{d}{dx^{-1}}}^{G}\circ \mathcal{L}_{\frac{d}{dx^{-2}}}^{G} \mathcal{L}_{%
\mathcal{J}_{1}}^{G}L \\ 
\\ 
=\mathcal{L}_{\frac{d}{dx^{-1}}}^{G} \circ \mathcal{L}_{\frac{d}{dx^{-2}}%
}^{G} \left( (-1)^{m-1} \iota _{\frac{\partial }{\partial x^{i}}} \eta
^{G}\wedge \theta ^{\mu } \cdot \dfrac{\partial L}{\partial y_{i}^{\mu }}
\right) \\ 
\\ 
=(-1)^{m-1}\mathcal{L}_{\frac{d}{dx^{-1}}}^{G} \left( \iota _{\frac{\partial 
}{\partial x^{i}}} \eta ^{G}\wedge \left( \theta _{-2}^{\mu }\cdot \dfrac{%
\partial L}{\partial y_{i}^{\mu }} +(-1)^{\mu }\theta ^{\mu }\cdot \dfrac{d}{%
dx^{-2}} \dfrac{\partial L}{\partial y_{i}^{\mu }} \right) \right) \\ 
\\ 
=(-1)^{m-1}\iota _{\frac{\partial }{\partial x^{i}}} \eta ^{G}\wedge \left(
\theta _{-1,-2}^{\mu }\cdot \dfrac{\partial L}{\partial y_{i}^{\mu }}
+(-1)^{\mu +1}\theta _{-2}^{\mu }\cdot \dfrac{d}{dx^{-1}} \dfrac{\partial L}{%
\partial y_{i}^{\mu }} \right. \\ 
\left. \hspace{3cm} +(-1)^{\mu }\theta _{-1}^{\mu }\cdot \dfrac{d}{dx^{-2}} 
\dfrac{\partial L}{\partial y_{i}^{\mu }} +\theta ^{\mu }\cdot \dfrac{d}{%
dx^{-1}}\dfrac{d}{dx^{-2}} \dfrac{\partial L}{\partial y_{i}^{\mu }} \right)
.%
\end{array}
\label{eq5_14}
\end{equation}
Moreover, we can apply the $\mathcal{V}((p_{3})_{10})-$valued $m-$form $%
\mathcal{J}_{3}$ on $J_{G}^{3}(p)$ to the superfunction 
\begin{equation*}
\frac{d^{2}L}{dx^{-1}dx^{-2}}\in \mathcal{A}_{J_{G}^{3}(p)},
\end{equation*}
the result being 
\begin{align}
\!\!\!\!\!\!\tilde{\Theta}^{L} -\eta ^{G}\cdot \frac{d^{2}L}{dx^{-1}dx^{-2}}
& =\mathcal{L}_{\mathcal{J}_{3}}^{G} \frac{d^{2}L}{dx^{-1}dx^{-2}}  \notag \\
& =(-1)^{m-1}\iota _{\frac{\partial }{\partial x^{i}}} \eta ^{G}\wedge
\left( \theta ^{\mu }\cdot \frac{\partial }{\partial y_{i}^{\mu }} \frac{%
d^{2}L}{dx^{-1}dx^{-2}} \right.  \notag \\
& \left. \quad +\theta _{\alpha }^{\mu }\cdot \frac{\partial }{\partial
y_{\alpha i}^{\mu }} \frac{d^{2}L}{dx^{-1}dx^{-2}} +\theta _{\alpha \beta
}^{\mu } \cdot \frac{\partial }{\partial y_{\alpha \beta i}^{\mu }} \frac{%
d^{2}L}{dx^{-1}dx^{-2}} \right) .  \label{eq5_15}
\end{align}

The factor 
\begin{equation*}
\frac{d^2L} {dx^{-1}dx^{-2}}
\end{equation*}
can be evaluated in two different ways: 
\begin{align}
\frac{d^2L}{dx^{-1}dx^{-2}} & =\frac{d}{dx^{-1}} \left( \frac{dL}{dx^{-2}}
\right)  \notag \\
& =\frac{d}{dx^{-1}} \left( \frac{\partial L}{\partial x^{-2}} +y_{-2}^\nu 
\frac{\partial L}{\partial y^\nu } +y_{-2,\alpha }^\nu \frac{\partial L}{%
\partial y_\alpha ^\nu } \right)  \notag \\
& =\frac{d}{dx^{-1}} \frac{\partial L}{\partial x^{-2}} +y_{-1,-2}^\nu \frac{%
\partial L}{\partial y^\nu } -(-1)^\nu y_{-2}^\nu \frac{d}{dx^{-1}} \frac{%
\partial L} {\partial y^\nu }  \notag \\
& \quad +y_{-1,-2,\alpha }^\nu \frac{\partial L}{\partial y_\alpha ^\nu }
+(-1)^{\nu +\alpha +1}y_{-2,\alpha }^\nu \frac{d}{dx^{-1}} \frac{\partial L}{%
\partial y_\alpha ^\nu },  \label{eq5_16}
\end{align}
or else,

\begin{align}
\frac{d^{2}L}{dx^{-1}dx^{-2}}& =-\frac{d}{dx^{-2}} \left( \frac{dL}{dx^{-1}}
\right)  \notag \\
& =-\frac{d}{dx^{-2}} \left( \frac{\partial L}{\partial x^{-1}} +y_{-1}^{\nu
} \frac{\partial L}{\partial y^{\nu }} +y_{-1,\alpha }^{\nu } \frac{\partial
L}{\partial y_{\alpha }^{\nu }} \right)  \notag \\
& =-\frac{d}{dx^{-2}} \frac{\partial L}{\partial x^{-1}}-y_{-2,-1}^{\nu } 
\frac{\partial L}{\partial y^{\nu }} +(-1)^{\nu }y_{-1}^{\nu }\frac{d}{%
dx^{-2}} \frac{\partial L}{\partial y^{\nu }}  \notag \\
& -y_{-2,-1,\alpha }^{\nu } \frac{\partial L}{\partial y_{\alpha }^{\nu }}
-(-1)^{\nu +\alpha +1}y_{-1,\alpha }^{\nu } \frac{d}{dx^{-2}} \frac{\partial
L}{\partial y_{\alpha }^{\nu }}.  \label{eq5_17}
\end{align}
In any case, neither the factors of $\partial L/\partial y^{\mu }$, $%
\partial L/\partial x^{-i}$ nor $d/dx^{-i}(\partial L/\partial y^{\mu })$
contain $3$-derivatives of the kind $y_{-1,-2,\alpha }^{\nu }$. Thus, by
using \eqref{eq5_16}, we have 
\begin{align*}
\frac{\partial }{\partial y_{-j-ki}^{\mu }} \frac{d^{2}L}{dx^{-1}dx^{-2}} & =%
\frac{\partial }{\partial y_{-j-ki}^{\mu }} \left( y_{-1,-2,\alpha }^{\nu } 
\frac{\partial L}{\partial y_{\alpha }^{\nu }} \right) \\
& =\delta _{-j}^{-1}\delta _{-k}^{-2} \frac{\partial L}{\partial y_{i}^{\mu }%
}.
\end{align*}
Also, there are no terms like $y_{ij}^{\mu }$ in $L\in \mathcal{A}%
_{J_{G}^{1}(p)}$, neither in $d^{2}L/dx^{-1}dx^{-2}$ (as $d/dx^{-1}$, $%
d/dx^{-2}$ just introduce derivatives with respect to odd indices), so that

\begin{equation*}
\frac{\partial }{\partial y_{i\alpha }^\mu } \frac{d^2L}{dx^{-1}dx^{-2}} =%
\frac{\partial }{\partial y_{-ji}^\mu } \frac{d^2L}{dx^{-1}dx^{-2}}.
\end{equation*}

Now, comparing \eqref{eq5_14} and \eqref{eq5_15}, we see that proving $%
\Theta ^{L}=\tilde{\Theta}^{L}$ reduces to see whether 
\begin{align*}
\iota _{\frac{\partial }{\partial x^{i}}} \eta ^{G} \wedge \theta _{-j}^{\mu
}\cdot \frac{\partial }{\partial y_{-ji}^{\mu }} \frac{d^{2}L}{dx^{-1}dx^{-2}%
} & =\iota _{\frac{\partial }{\partial x^{i}}} \eta ^{G}\wedge \left(
(-1)^{\mu }\theta _{-1}^{\mu } \cdot \frac{d}{dx^{-2}} \frac{\partial L}{%
\partial y_{i}^{\mu }} \right. \\
& \left. \quad +(-1)^{\mu +1}\theta _{-2}^{\mu }\cdot \frac{d}{dx^{-1}} 
\frac{\partial L}{\partial y_{i}^{\mu }} \right) ,
\end{align*}
or, developing the left-hand side, 
\begin{multline*}
\iota _{\frac{\partial }{\partial x^{i}}}\eta ^{G}\wedge \left( \theta
_{-1}^{\mu }\cdot \frac{\partial }{\partial y_{-1,i}^{\mu }} +\theta
_{-2}^{\mu }\cdot \frac{\partial }{\partial y_{-2,i}^{\mu }} \right) \frac{%
d^{2}L}{dx^{-1}dx^{-2}} \\
=(-1)^{\mu }\iota _{\frac{\partial }{\partial x^{i}}} \eta ^{G}\wedge \left(
\theta _{-1}^{\mu }\cdot \frac{d}{dx^{-2}} \frac{\partial L}{\partial
y_{i}^{\mu }} -\theta _{-2}^{\mu }\cdot \frac{d}{dx^{-1}} \frac{\partial L}{%
\partial y_{i}^{\mu }} \right) .
\end{multline*}
What we are going to see is 
\begin{equation}
\left. 
\begin{array}{c}
\dfrac{\partial }{\partial y_{-1,i}^{\mu }} \dfrac{d^{2}L}{dx^{-1}dx^{-2}}
=(-1)^{\mu }\dfrac{d}{dx^{-2}} \dfrac{\partial L}{\partial y_{i}^{\mu }}
\medskip \\ 
\dfrac{\partial }{\partial y_{-2,i}^{\mu }} \dfrac{d^{2}L}{dx^{-1}dx^{-2}}
=-(-1)^{\mu }\dfrac{d}{dx^{-1}} \dfrac{\partial L}{\partial y_{i}^{\mu }}%
\end{array}
\right\}  \label{eq5_18}
\end{equation}
To prove the first formula in \eqref{eq5_18}, we use \eqref{eq5_16}. It is
clear that the only terms containing factors like $y_{-1,i}^{\mu }$, $%
y_{-2,i}^{\mu }$ are those indicated:

\begin{align*}
\frac{\partial L}{\partial y_{-1,i}^\mu } & =\frac{\partial }{\partial
y_{-1,i}^\mu } \frac{d^2L}{dx^{-1}dx^{-2}} \\
& =\frac{\partial }{\partial y_{-1,i}^\mu } \left( y_{-1,\alpha}^\nu \frac{%
\partial ^2L}{\partial y_\alpha ^\nu \partial x^{-2}} -(-1)^\nu y_{-2}^\nu
y_{-1,\alpha}^\xi \frac{\partial ^2L}{\partial y_\alpha ^\xi \partial y^\nu }
\right. \\
& \hspace{4.5cm} \left. -(-1)^{\nu +\alpha } y_{-2,\alpha }^\nu y_{-1\beta
}^\xi \frac {\partial ^2L}{\partial y_\beta ^\xi \partial y_\alpha ^\nu }
\right) \\
& =\frac{\partial ^2L} {\partial y_i^\mu \partial x^{-2}} +(-1)^{\mu (\nu
+1)} y_{-2}^\nu \frac{\partial ^2L}{\partial y_i^\mu \partial y^\nu }
+(-1)^{\mu +\mu (\nu +\alpha )} y_{-2,\alpha }^\nu \frac{\partial ^2L}{%
\partial y_i^\mu \partial y_\alpha ^\nu } \\
& =(-1)^\mu \left( \frac{\partial } {\partial x^{-2}} +y_{-2}^\nu \frac{%
\partial }{\partial y^\nu } +y_{-2,\alpha }^\nu \frac{\partial }{\partial
y_\alpha ^\nu } \right) \frac{\partial L} {\partial y_i^\mu } \\
& =(-1)^\mu \frac{d}{dx^{-2}} \frac{\partial L}{\partial y_i^\mu }.
\end{align*}
To prove the second formula in \eqref{eq5_18}, one has just to repeat the
preceding computations but using \eqref{eq5_17}.

\begin{remark}
The proof of the lemmas and the proposition in subsection \ref{lemmata} is
just a generalization (by induction) of the computations leading to
equations \eqref{eq5_16}, \eqref{eq5_17}, and \eqref{eq5_18}.
\end{remark}

Thus, once a volume form has been chosen on the base manifold $M$, we have
constructed a Poincar\'{e}$-$Cartan form, 
\begin{equation*}
\Theta ^{L} =\mathcal{L}_{\mathcal{J}_{n+1}}^{G} \left( \frac{d^{n}L}{%
dx^{-1}\ldots dx^{-n}} \right) +\eta ^{G}\cdot \frac{d^{n}L}{dx^{-1}\ldots
dx^{-n}}
\end{equation*}
out of intrinsically defined objects. Moreover, we have proved the
equivalence with the alternative expression 
\begin{equation*}
\Theta ^{L} =\mathcal{L}_{\frac{d}{dx^{-1}}}^{G} \circ \ldots \circ \mathcal{%
L}_{\frac{d}{dx^{-n}}}^{G} \Theta _{0}^{L},
\end{equation*}
which, as it does not involve higher$-$order operators, could be more
appropriate for explicit computations.

\section{Deduction of the Euler - Lagrange equations from the Poincar\'e - Cartan form}

\subsection{The exterior derivative of the Poincar\'e - Cartan form}

According to the previous section, we have a well-defined procedure to
obtain the Euler$-$Lagrange superequations for a superfield theory described
by a first$-$order Berezinian density 
\begin{equation*}
\xi _{L} =\left[ d^{G}x^{1}\wedge \ldots \wedge d^{G}x^{m}\otimes \frac{d}{%
dx^{-1}}\circ \ldots \circ \frac{d}{dx^{-n}} \right] \cdot L,\quad L\in 
\mathcal{A}_{J_{G}^{1}(p)},
\end{equation*}
in a similar way to that of the classical case: First, we must consider the
Poincar\'{e}$-$Cartan form $\Theta ^{L}$, then its differential $d^{G}\Theta
^{L}$ and finally study the insertion of vertical superfields. The idea is
to obtain a decomposition of $d^{G}\Theta ^{L}$ as the product of the Euler$-
$Lagrange operator by the graded contact $1-$forms and/or their derivatives
plus other terms, as expressed in the following proposition.

We make use of the decomposition $d^{G}=D+\partial $, where 
\begin{align*}
D& =D_{0}+D_{1} \\
& =d^{G}x^{\alpha }\otimes \mathcal{L}_{\frac{d}{dx^{\alpha }}}^{G}
\end{align*}
is the graded horizontal differential (given as a sum of the horizontal
differential with respect to even and odd coordinates on the base manifold)
and $\partial =d^{G}-D$ is the graded vertical differential, which
differentiates with respect to the fiber coordinates (recall subsection \ref%
{horver}).

\begin{proposition}
For every $L\in \mathcal{A}_{J_{G}^{1}(p)}$, we have 
\begin{equation*}
d^{G}\Theta ^{L} =\mathcal{L}_{\frac{d}{dx^{-1}}}^{G} \circ \ldots \circ 
\mathcal{L}_{\frac{d}{dx^{-n}}}^{G} \left( \alpha _{L}+\varpi _{L}+D_{1}
\left( \Theta _{0}^{L}-\eta ^{G}\cdot L \right) +\partial \left( \Theta
_{0}^{L}-\eta ^{G}\cdot L \right) \right) ,
\end{equation*}
where $\varpi _{L}$ and $\alpha _{L}$ are the $(m+1)-$forms on $J_{G}^{2}(p)$%
, defined by 
\begin{align*}
\varpi _{L} & =(-1)^{m}\eta ^{G}\wedge \left( \theta ^{\mu } \left( \frac{%
\partial L}{\partial y^{\mu }} -\frac{d}{dx^{i}} \frac{\partial L}{\partial
y_{i}^{\mu }} \right) +\theta _{-i}^{\mu } \frac{\partial L}{\partial
y_{-i}^{\mu }} \right) \\
\alpha _{L} & =(-1)^{m}\eta ^{G}\wedge d^{G}x^{\alpha } \cdot \left( 2\frac{%
dL}{dx^{\alpha }} -\frac{\partial L}{\partial x^{\alpha }} \right) .
\end{align*}
\end{proposition}

\begin{proof}
From the preceding section,
recalling that the operators
$\mathcal{L}_{d/dx^{-1}}^G$
and $d^G$ commute, we obtain
\begin{align*}
d^G\Theta ^L
&
=\mathcal{L}_{
\frac{d}{dx^{-1}}
}^G
\circ \ldots \circ
\mathcal{L}_{
\frac{d}{dx^{-n}}
}^Gd^G
\Theta _0^L\\
&
=\mathcal{L}_{
\frac{d}{dx^{-1}}
}^G
\circ \ldots \circ
\mathcal{L}_{
\frac{d}{dx^{-n}}
}^G
\left(
D_0
\left(
\Theta _0^L-\eta ^G\cdot L
\right)
\right. \\
&
\quad
+D_1
\left(
\Theta _0^L-\eta ^G\cdot L
\right)
+\partial
\left(
\Theta _0^L-\eta ^G\cdot L
\right) \\
&
\left.
\quad
+(-1)^m\eta ^G\wedge d^GL
\right) .
\end{align*}
Let us concentrate in the terms
$D_0(\Theta _0^L-\eta ^G\cdot L)
+(-1)^m\eta ^G\wedge d^GL$.
On the one hand, we have
\begin{align*}
D_0(\Theta _0^L-\eta ^G\cdot L)
&
=d^Gx^i \wedge \mathcal{L}_{\frac{d}{dx^i}}^G
(\Theta _0^L-\eta ^G\cdot L)\\
&
=(-1)^{m-1}\eta ^G\wedge
\left(
\theta _i^\mu
\frac{\partial L}
{\partial y_i^\mu }
+\theta ^\mu
\frac{d}{dx^i}
\frac{\partial L}
{\partial y_i^\mu }
\right) ,
\end{align*}
and, on the other,
\begin{equation*}
\eta ^G\wedge d^GL
=\eta ^G\wedge
\left(
d^Gx^\alpha \cdot
\frac{dL}{dx^\alpha }
+d^Gy^\mu \cdot
\frac{\partial L}
{\partial y^\mu }
+d^Gy_\alpha ^\mu \cdot
\frac{\partial L}
{\partial y_\alpha ^\mu }
\right) .
\end{equation*}
Thus, substituting,
\begin{align*}
&
\!\!\!\!\!\!\!\!\!\!\!\!
\!\!\!\!\!\!\!\!\!\!\!\!
\!\!\!
D_0
\left(
\Theta _0^L-\eta ^G \cdot L
\right)
+(-1)^m\eta ^G\wedge d^GL\\
\qquad
\qquad
&
=(-1)^{m}\eta ^G\wedge
\left(
-\theta _i^\mu
\frac{\partial L}
{\partial y_i^\mu }
-\theta ^\mu
\frac{d}{dx^i}
\frac{\partial L}
{\partial y_i^\mu }
\right. \\
&
\quad
\qquad
\qquad
\left.
+d^Gx^\alpha
\cdot
\frac{dL}{dx^\alpha }
+d^Gy^\mu
\cdot
\frac{\partial L}
{\partial y^\mu }
+d^Gy_\alpha^\mu
\cdot
\frac{\partial L}
{\partial y_\alpha ^\mu }
\right)  \\
\qquad
\qquad
&
=(-1)^m\eta ^G\wedge
\left(
-\theta _i^\mu
\frac{\partial L}
{\partial y_i^\mu }
-\theta ^\mu
\frac{d}{dx^i}
\frac{\partial L}
{\partial y_i^\mu }
+d^Gx^\alpha
\cdot
\frac{dL}{dx^\alpha }
\right. \\
&
\quad
\qquad
\qquad
\left.
+\left(
\theta ^\mu +d^Gx^\alpha
\cdot
y_\alpha ^\mu
\right)
\frac{\partial L}
{\partial y^\mu }
+\left(
\theta _\alpha ^\mu
+d^Gx^\beta
\cdot
y_{\beta \alpha }^\mu
\right)
\frac{\partial L}
{\partial y_\alpha ^\mu }
\right)  \\
\qquad
\qquad
&
=(-1)^m\eta ^G\wedge
\left(
\theta _{-i}^\mu
\frac{\partial L}
{\partial y_{-i}^\mu }
+\theta ^\mu
\left(
\frac{\partial L}
{\partial y^\mu }
-\frac{d}{dx^i}
\frac{\partial L}
{\partial y_i^\mu }
\right)
\right.  \\
&
\quad
\qquad
\qquad
\left.
+d^Gx^\alpha
\cdot
\frac{dL}{dx^\alpha }
+d^Gx^\alpha
\cdot
y_\alpha ^\mu
\frac{\partial L}
{\partial y^\mu }
+d^Gx^\beta
\cdot
y_{\beta \alpha}^\mu
\frac{\partial L}
{\partial y_\alpha ^\mu}
\right) \\
\qquad
\qquad
&
=(-1)^{m}\eta ^G\wedge
\left(
\theta _{-i}^\mu
\frac{\partial L}
{\partial y_{-i}^\mu }
+\theta ^\mu
\left(
\frac{\partial L}
{\partial y^\mu }
-\frac{d}{dx^i}
\frac{\partial L}
{\partial y_i^\mu }
\right)
\right. \\
&
\quad
\qquad
\qquad
\left.
+d^Gx^\alpha
\cdot
\left(
2\frac{dL}{dx^\alpha }
-\frac{\partial L}
{\partial x^\alpha }
\right)
\right) \\
\qquad
\qquad
&
=\varpi _L+\alpha _L.
\end{align*}
\end{proof}

We should also remark that for every vector field $X$ on $J_G^2(p)$ vertical
over $(M,\mathcal{A})$, we have $\iota _X\alpha_L=0$.

Now, we would like to extract the Euler$-$Lagrange superequations of field
theory from the decomposition of the previous proposition. To this end, we
first need the following technical lemma, whose proof reduces to a simple
computation:

\begin{lemma}
\label{Sigma} Let $\Sigma _{-n}$ denote the group of permutations of $%
\{-1,\dotsc ,-n\}$. For any $A,B\in \Omega _{G}(J_{G}^{1}(p))$, we have 
\begin{multline*}
\mathcal{L}_{\frac{d}{dx^{-1}}}^{G}\circ \ldots \circ \mathcal{L}_{\frac{d}{%
dx^{-n}}}^{G}(A\wedge B) \\
=\sum\limits_{\substack{ \sigma =\sigma _{1}\cup \sigma _{2}\in \Sigma _{-n} 
\\ 0\leq |\sigma |\leq n}}(-1)^{|\sigma _{2}||A|+\tau } \left( \mathcal{L}_{%
\frac{d}{dx^{\sigma _{1}(-1)}}}^{G} \circ \ldots \circ \mathcal{L}_{\frac{d}{%
dx^{\sigma _{1} (-|\sigma _{1}|)}}}^{G}A \right) \\
\cdot \left( \mathcal{L}_{\frac{d}{dx^{\sigma _{2}(-|\sigma _{1}|+1)}}}^{G}
\circ \ldots \circ \mathcal{L}_{\frac{d}{dx^{\sigma _{2}(-n)}}}^{G}B \right)
,
\end{multline*}
where $\tau $ is the number of transpositions needed to reorder $(\sigma
_{1}(-1),\dotsc ,\sigma _{2}(-n))$.
\end{lemma}

\begin{proposition}
\label{Pro17} With the preceding notations, we have 
\begin{multline*}
\mathcal{L}_{\frac{d}{dx^{-1}}}^{G} \circ \ldots \circ \mathcal{L}_{\frac{d}{%
dx^{-n}}}^{G}(\varpi _{L}) \\
=\sum\limits_{\substack{ \sigma=\sigma _{1}\cup \sigma _{2}\in \Sigma _{-n} 
\\ 0\leq |\sigma _{2}|\leq n  \\  \\ |\sigma _{2}|\mu +\tau }} (-1)^{|\sigma
_{2}|\mu +\tau +m} \eta ^{G}\wedge \theta _{\sigma _{1}(-1)\ldots \sigma
_{1} (-|\sigma _{1}|)}^{\mu } \frac{d^{|\sigma _{2}|}\mathcal{E}(L)} {%
dx^{\sigma _{2}(-|\sigma _{1}|-1)}\ldots dx^{\sigma _{2}(-n)}},
\end{multline*}
where $\mathcal{E}$ is the Euler$-$Lagrange operator, 
\begin{equation*}
\mathcal{E}(L) =\frac{\partial L}{\partial y^{\mu }} -\frac{d}{dx^{i}} \frac{%
\partial L}{\partial y_{i}^{\mu }} -(-1)^{\mu }\frac{d}{dx^{-i}} \frac{%
\partial L}{\partial y_{-i}^{\mu }}.
\end{equation*}
\end{proposition}

\begin{proof}
Let us write
\begin{equation*}
\varpi _L
=(-1)^m\eta ^G\wedge
\left(
\theta ^\mu \omega _\mu
+\theta _{-i}^\mu
\frac{\partial L}
{\partial y_{-i}^\mu }
\right) ,
\end{equation*}
where
\begin{equation*}
\omega _\mu
=\frac{\partial L}
{\partial y^\mu }
-\frac{d}{dx^i}
\frac{\partial L}
{\partial y_i^\mu }.
\end{equation*}
Then, we have
\begin{multline*}
\mathcal{L}_{
\frac{d}{dx^{-1}}}^G
\circ \ldots \circ
\mathcal{L}_{
\frac{d}{dx^{-n}}
}^G
(\varpi _L) \\
=(-1)^m\eta ^G\wedge
\mathcal{L}_{
\frac{d}{dx^{-1}}
}^G
\circ \ldots \circ
\mathcal{L}_{
\frac{d}{dx^{-n}}
}^G
\left(
\theta ^\mu \omega _\mu
+\theta _{-i}^\mu
\frac{\partial L}
{\partial y_{-i}^\mu }
\right) ,
\end{multline*}
and by applying Lemma
\ref{Sigma}, we obtain
\begin{align*}
&
\mathcal{L}_{
\frac{d}{dx^{-1}}
}^G
\circ \ldots \circ
\mathcal{L}_{
\frac{d}{dx^{-n}}
}^G
\left(
\theta ^\mu \omega _\mu
+\theta _{-i}^\mu
\frac{\partial L}
{\partial y_{-i}^\mu }
\right) \\
&
=\sum \limits _{
\substack{
\sigma
=\sigma _1\cup \sigma _2
\in \Sigma _{-n}
\\
0\leq |\sigma _2|\leq n
}
}
\left(
(-1)^{
|\sigma _2|\mu +\tau
}
\left(
\mathcal{L}_{
\frac{d}{
dx^{\sigma _1(-1)}
}
}^G
\circ \ldots \circ
\mathcal{L}_{
\frac{d}{
dx^{
\sigma _1(-|\sigma _1|)
}
}
}^G
\right)
\theta ^\mu \cdot
\right. \\
&
\hspace{5cm}
\cdot
\left(
\mathcal{L}_{
\frac{d}
{dx^{
\sigma _2(-|\sigma _1|+1)
}
}
}^G
\circ \ldots \circ
\mathcal{L}_{
\frac{d}{
dx^{\sigma _2(-n)}
}
}^G
\right)
\omega _\mu \\
&
\quad
+(-1)^{
|\sigma _2|(\mu +1)+\tau
}
\left(
\mathcal{L}_{
\frac{d}{
dx^{\sigma _1(-1)}
}
}^G
\circ \ldots \circ
\mathcal{L}_{
\frac{d}{
dx^{
\sigma _1(-|\sigma _1|)
}
}
}^G
\theta _{-i}^\mu
\right)
\cdot\\
&
\hspace{5cm}
\left.
\cdot
\left(
\mathcal{L}_{
\frac{d}{
dx^{
\sigma _2(-|\sigma _1|+1)
}
}
}^G
\circ \ldots \circ
\mathcal{L}_{
\frac{d}{
dx^{\sigma _2(-n)}
}
}^G
\frac{\partial L}
{\partial y_{-i}^\mu }
\right)
\right) \\
&
\\
&
=\sum \limits _{
\substack{
\sigma
=(\sigma _1\cup \sigma _2)
\in \Sigma _{-n}\\
0\leq |\sigma _2|\leq n
}
}
\left(
(-1)^{|\sigma _2|\mu +\tau }
\theta _{
\sigma _1(-1)\ldots
\sigma _1(-|\sigma _1|)
}^\mu
\frac{d^{|\sigma _2|}}
{
dx^{\sigma _2
(-|\sigma _1|-1)}
\ldots
dx^{\sigma _2(-n)}
}
\omega _\mu
\right.  \\
&
\quad+
\left.
(-1)^{
|\sigma _2|(\mu +1)+\tau
}
\theta _{\sigma _1
(-1)\ldots \sigma _1
(-|\sigma _1|),-i
}^\mu
\frac{d^{|\sigma _2|}}
{dx^{\sigma _2
(-|\sigma _1|-1)}
\ldots
dx^{\sigma _2(-n)}}
\frac{\partial L}
{\partial y_{-i}^\mu }
\right) \\
&
\\
&
=\sum \limits_{
\substack{
\sigma
=(\sigma _1\cup \sigma _2)
\in \Sigma _{-n}\\
0\leq |\sigma _2|\leq n
}
}
(-1)^{|\sigma _2|\mu +\tau }
\theta _{
\sigma _1(-1)
\ldots
\sigma _1(-|\sigma _1|)
}^\mu
\left(
\frac{d^{|\sigma _2|}}
{dx^{\sigma _2(-|\sigma _1|-1)}
\ldots
dx^{\sigma _2(-n)}}
\omega _\mu
\right. \\
&
\hspace{4.5cm}
\left.
-(-1)^\mu
\frac{d^{|\sigma _2|}}
{dx^{\sigma _2(-|\sigma _1|-1)}
\ldots
dx^{\sigma _2(-n)}}
\frac{d}{dx^{-i}}
\frac{\partial L}
{\partial y_{-i}^\mu }
\right) \\
& \\
&
=\sum \limits _{\substack{
\sigma
=\sigma _1\cup \sigma _2
\in \Sigma_{-n}\\
0\leq |\sigma _2|\leq n
}
}
(-1)^{|\sigma _2|\mu +\tau }
\theta _{
\sigma _1(-1)
\ldots
\sigma _1(-|\sigma _1|)
}^\mu
\frac{d^{|\sigma _2|}}
{dx^{\sigma _2
(-|\sigma _1|-1)}
\ldots dx^{\sigma _2(-n)}}
\left(
\mathcal{E}(L)
\right) .
\end{align*}
\end{proof}

\subsection{An example}

Again, let us clarify the notation by working out the example of $(m|2)-$%
superfield theory. Here we have

\begin{align*}
& \mathcal{L}_{\frac{d}{dx^{-1}}}^{G}\mathcal{L}_{\frac{d}{dx^{-2}}}^{G}
\left( \theta ^{\mu } \left( \frac{\partial L}{\partial y^{\mu }} -\frac{d}{%
dx^{i}}\frac{\partial L}{\partial y_{i}^{\mu }} \right) +\theta _{-1}^{\mu } 
\frac{\partial L}{\partial y_{-1}^{\mu }} +\theta _{-2}^{\mu }\frac{\partial
L}{\partial y_{-2}^{\mu }} \right) \\
& =\mathcal{L}_{\frac{d}{dx^{-1}}}^{G} \left( \theta _{-2}^{\mu } \left( 
\frac{\partial L}{\partial y^{\mu }} -\frac{d}{dx^{i}}\frac{\partial L}{%
\partial y_{i}^{\mu }} \right) +(-1)^{\mu }\theta ^{\mu }\frac{d}{dx^{-2}}
\left( \frac{\partial L}{\partial y^{\mu }} -\frac{d}{dx^{i}}\frac{\partial L%
}{\partial y_{i}^{\mu }} \right) \right. \\
& \quad +\left. \theta _{-2,-1}^{\mu } \frac{\partial L}{\partial
y_{-1}^{\mu}} -(-1)^{\mu }\theta _{-1}^{\mu } \frac{d}{dx^{-2}}\frac{%
\partial L}{\partial y_{-1}^{\mu }} -(-1)^{\mu }\theta _{-2}^{\mu }\frac{d}{%
dx^{-2}} \frac{\partial L}{\partial y_{-2}^{\mu }} \right) \\
& =\theta _{-1,-2}^{\mu } \left( \frac{\partial L}{\partial y^{\mu }} -\frac{%
d}{dx^{i}}\frac{\partial L}{\partial y_{i}^{\mu }} \right) -(-1)^{\mu
}\theta _{-2}^{\mu }\frac{d}{dx^{-1}} \left( \frac{\partial L}{\partial
y^{\mu }} -\frac{d}{dx^{i}}\frac{\partial L}{\partial y_{i}^{\mu }} \right)
\\
& \quad +(-1)^{\mu }\theta _{-1}^{\mu } \frac{d}{dx^{-2}} \left( \frac{%
\partial L}{\partial y^{\mu }} -\frac{d}{dx^{i}}\frac{\partial L}{\partial
y_{i}^{\mu }} \right) +\theta ^{\mu }\frac{d}{dx^{-1}}\frac{d}{dx^{-2}}
\left( \frac{\partial L}{\partial y^{\mu }} -\frac{d}{dx^{i}}\frac{\partial L%
}{\partial y_{i}^{\mu }} \right) \\
& \quad +(-1)^{\mu }\theta _{-2,-1}^{\mu } \frac{d}{dx^{-1}}\frac{\partial L%
}{\partial y_{-1}^{\mu }} +\theta _{-1}^{\mu }\frac{d}{dx^{-1}}\frac{d}{%
dx^{-2}} \frac{\partial L}{\partial y_{-1}^{\mu }} \\
& \quad -(-1)^{\mu }\theta _{-1,-2}^{\mu } \frac{d}{dx^{-2}}\ \frac{\partial
L}{\partial y_{-2}^{\mu }} +\theta _{-2}^{\mu }\frac{d}{dx^{-1}}\frac{d}{%
dx^{-2}} \frac{\partial L}{\partial y_{-2}^{\mu }}.
\end{align*}
Next, grouping common factors of the contact $1-$forms, 
\begin{align*}
& \mathcal{L}_{\frac{d}{dx^{-1}}}^{G} \mathcal{L}_{\frac{d}{dx^{-2}}}^{G}
\left( \theta ^{\mu } \left( \frac{\partial L}{\partial y^{\mu }} -\frac{d}{%
dx^{i}}\frac{\partial L}{\partial y_{i}^{\mu }} \right) -\theta _{-1}^{\mu }%
\frac{\partial L}{\partial y_{-1}^{\mu }} -\theta _{-2}^{\mu }\frac{\partial
L}{\partial y_{-2}^{\mu }} \right) \\
& =\theta _{-1,-2}^{\mu } \left( \frac{\partial L}{\partial y^{\mu }} -\frac{%
d}{dx^{i}}\frac{\partial L}{\partial y_{i}^{\mu }} -(-1)^{\mu }\frac{d}{%
dx^{-1}} \frac{\partial L}{\partial y_{-1}^{\mu }} -(-1)^{\mu }\frac{d}{%
dx^{-2}} \frac{\partial L}{\partial y_{-2}^{\mu }} \right) \\
& +\theta _{-1}^{\mu } \left( (-1)^{\mu }\frac{d}{dx^{-2}} \left( \frac{%
\partial L}{\partial y^{\mu }} -\frac{d}{dx^{i}} \frac{\partial L}{\partial
y_{i}^{\mu }} \right) +\frac{d^{2}}{dx^{-1}dx^{-2}} \frac{\partial L}{%
\partial y_{-1}^{\mu }} \right) \\
& -\theta _{-2}^{\mu } \left( (-1)^{\mu }\frac{d}{dx^{-1}} \left( \frac{%
\partial L}{\partial y^{\mu }} -\frac{d}{dx^{i}}\frac{\partial L}{\partial
y_{i}^{\mu }} \right) +\frac{d^{2}}{dx^{-1}dx^{-2}} \frac{\partial L}{%
\partial y_{-2}^{\mu }} \right) \\
& +\theta ^{\mu } \left( \frac{d^{2}}{dx^{-1}dx^{-2}} \left( \frac{\partial L%
}{\partial y^{\mu }} -\frac{d}{dx^{i}}\frac{\partial L}{\partial y_{i}^{\mu }%
} \right) \right) ,
\end{align*}
and an algebraic rearrangement finally gives,

\begin{align*}
& \mathcal{L}_{\frac{d}{dx^{-1}}}^{G} \mathcal{L}_{\frac{d}{dx^{-2}}}^{G}
\left( \theta ^{\mu } \left( \frac{\partial L}{\partial y^{\mu }} -\frac{d}{%
dx^{i}}\frac{\partial L}{\partial y_{i}^{\mu }} \right) -\theta _{-1}^{\mu }%
\frac{\partial L}{\partial y_{-1}^{\mu }} -\theta _{-2}^{\mu }\frac{\partial
L }{\partial y_{-2}^{\mu }} \right) \\
& =\theta _{-1,-2}^{\mu } \left( \frac{\partial L}{\partial y^{\mu }} -\frac{%
d}{dx^{i}}\frac{\partial L}{\partial y_{i}^{\mu }} -(-1)^{\mu }\frac{d}{%
dx^{-j}} \frac{\partial L}{\partial y_{-j}^{\mu }} \right) \\
& +(-1)^{\mu }\theta _{-1}^{\mu } \left( \frac{d}{dx^{-2}} \left( \frac{%
\partial L}{\partial y^{\mu }} -\frac{d}{dx^{i}} \frac{\partial L}{\partial
y_{i}^{\mu }} \right) -(-1)^{\mu }\frac{d}{dx^{-2}} \frac{d}{dx^{-j}} \frac{%
\partial L}{\partial y_{-j}^{\mu }} \right) \\
& -(-1)^{\mu }\theta _{-2}^{\mu } \left( frac{d}{dx^{-1}} \left( \frac{%
\partial L}{\partial y^{\mu }} -\frac{d}{dx^{i}}\frac{\partial L}{\partial
y_{i}^{\mu } }\right) -(-1)^{\mu }\frac{d}{dx^{-1}}\frac{d}{dx^{-j}} \frac{%
\partial L}{\partial y_{-j}^{\mu }} \right) \\
& +\theta ^{\mu } \left( \frac{d^{2}}{dx^{-1}dx^{-2}} \left( \frac{\partial L%
}{\partial y^{\mu }} -\frac{d}{dx^{i}}\frac{\partial L}{\partial y_{i}^{\mu }%
} -(-1)^{\mu } \frac{d}{dx^{-j}}\frac{\partial L}{\partial y_{-j}^{\mu }}
\right) \right) \\
& =\left( \theta _{-1,-2}^{\mu } +(-1)^{\mu }\theta _{-1}^{\mu }-(-1)^{\mu }
\theta _{-2}^{\mu }+\theta ^{\mu } \right) \mathcal{E}(L).
\end{align*}

\subsection{The Euler - Lagrange equations}

In view of Proposition \ref{Pro17}, the term $\varpi _{L}$ alone already
gives us the Euler$-$Lagrange equations, so we must study the vanishing of
the terms 
\begin{equation*}
D_{1}\left( \Theta _{0}^{L}-\eta ^{G}\cdot L \right) +\partial \left( \Theta
_{0}^{L}-\eta ^{G}\cdot L \right) .
\end{equation*}

\begin{lemma}
\label{lemma14} With the preceding notations, we have 
\begin{equation*}
\partial =\theta ^\mu \wedge \mathcal{L}_{ \frac{\partial }{\partial y^\mu }%
}^G +\theta _\alpha ^\mu \wedge \mathcal{L}_{ \frac{\partial }{\partial
y_\alpha ^\mu }}^G -d^G\theta ^\mu \wedge \iota _{ \frac{\partial }{\partial
y^\mu } } -d^G\theta _\beta ^\mu \wedge \iota _{ \frac{\partial }{\partial
y_\beta ^\mu }}.
\end{equation*}
\end{lemma}

\begin{proof}
From the very definition we have
\begin{equation*}
d^G=d^Gx^\beta\wedge
\mathcal{L}_{
\frac{\partial }
{\partial x^\beta }
}^G
+d^Gy^\mu \wedge
\mathcal{L}_{
\frac{\partial }
{\partial y^\mu }
}^G
+d^Gy_\alpha ^\mu\wedge
\mathcal{L}_{
\frac{\partial }
{\partial y_\alpha ^\mu }
}^G,
\end{equation*}
and also from the definition,
$\partial =d^G-D$, where
$D$ is the horizontal
differential. Therefore,
\begin{align*}
\partial
&
=d^G-D\\
&
=d^G-d^Gx^\gamma \wedge
\mathcal{L}_{
\frac{d}{dx^\gamma }
}^G\\
&
=d^Gx^\beta \wedge
\mathcal{L}_{
\frac{\partial }
{\partial x^\beta }
}^G
+d^Gy^\mu \wedge
\mathcal{L}_{
\frac{\partial }
{\partial y^\mu }
}^G
+d^Gy_\alpha ^\mu \wedge
\mathcal{L}_{
\frac{\partial }
{\partial y_\alpha ^\mu }
}^G
-d^Gx^\gamma \wedge
\mathcal{L}_{
\frac{d}{dx^\gamma }
}^G.
\end{align*}
Furthermore, as
\begin{equation*}
\frac{d}{dx^\gamma }
=\frac{\partial }
{\partial x^\gamma }
+y_\gamma ^\mu
\frac{\partial }
{\partial y^\mu }
+y_{\gamma \alpha }^\mu
\frac{\partial }
{\partial y_\alpha ^\mu },
\end{equation*}
taking the properties
of the graded Lie derivative
into account, we obtain
\begin{align*}
\partial
&
=d^Gy^\mu\wedge
\mathcal{L}_{
\frac{\partial }
{\partial y^\mu }
}^G
+d^Gy_\alpha ^\mu \wedge
\mathcal{L}_{
\frac{\partial }
{\partial y_\alpha ^\mu }
}^G
-d^Gx^\alpha \wedge
\mathcal{L}_{
y_\alpha ^\mu
\frac{\partial }
{\partial y^\mu }
}^G
-d^Gx^\alpha \wedge
\mathcal{L}_{
y_{\alpha \beta }^\mu
\frac{\partial }
{\partial y_\beta ^\mu }
}^G\\
& \\
&
=d^Gy^\mu \wedge
\mathcal{L}_{
\frac{\partial }
{\partial y^\mu }
}^G
+d^Gy_\alpha ^\mu \wedge
\mathcal{L}_{
\frac{\partial }
{\partial y_\alpha ^\mu }
}^G
-d^Gx^\alpha \wedge
\left(
d^G\iota _{
y_\alpha ^\mu
\frac{\partial }
{\partial y^\mu }
}
+\iota _{
y_\alpha ^\mu
\frac{\partial }
{\partial y^\mu }
}d^G
\right) \\
&
\hspace{5cm}
-d^Gx^\alpha \wedge
\left(
d^G\iota_{
y_{\alpha \beta}^\mu
\frac{\partial }
{\partial y_\beta ^\mu }
}
+\iota _{
y_{\alpha \beta }^\mu
\frac{\partial }
{\partial y_\beta ^\mu }
}d^G
\right) \\
& \\
&
=d^Gy^\mu \wedge
\mathcal{L}_
{\frac{\partial }
{\partial y^\mu }
}^G
+d^Gy_\alpha ^\mu \wedge
\mathcal{L}_{
\frac{\partial }
{\partial y_\alpha ^\mu }
}^G
-d^Gx^\alpha
\cdot y_\alpha ^\mu \wedge
\mathcal{L}_{
\frac{\partial }
{\partial y^\mu }
}^G
-d^Gx^\alpha
\cdot y_{\alpha \beta }^\mu
\wedge
\mathcal{L}_{
\frac{\partial }
{\partial y_\beta ^\mu }
}^G\\
&
\hspace{4cm}
-d^Gx^\alpha \wedge
d^Gy_\alpha ^\mu
\wedge
\iota _{
\frac{\partial }
{\partial y^\mu }
}
-d^Gx^\alpha \wedge
d^Gy_{\alpha \beta}^\mu
\wedge
\iota _{
\frac{\partial }
{\partial y_\beta^\mu }
}.
\end{align*}
Finally, by grouping
the correct terms
and by noting that
\begin{equation*}
d^G\theta _Q^\nu
=d^Gx^\alpha \wedge d^Gy_{\alpha \star Q}^\nu ,
\end{equation*}
we arrive at the statement
of the lemma.
\end{proof}

\begin{lemma}
For every vector field $X$ on $J_G^{n+1}(p)$, vertical over $(M,\mathcal{A})$%
, and for any local section $s$ of $p$, we have 
\begin{equation*}
(j^{n+1}s)^\ast \left( \iota _X \left( \mathcal{L}_{ \frac{d}{dx^{-1}} }^G
\circ \ldots \circ \mathcal{L}_{ \frac{d}{dx^{-n}} }^G \left( D_{1} \left(
\Theta _{0}^L-\eta ^G\cdot L \right) +\partial \left( \Theta _0^L-\eta
^G\cdot L \right) \right) \right) \right) =0.
\end{equation*}
\end{lemma}

\begin{proof}
As
\begin{equation*}
D_1=d^Gx^{-i}\wedge
\mathcal{L}_{
\frac{d}{dx^{-i}}
}^G,
\end{equation*}
it is clear that
\begin{equation*}
\mathcal{L}_{
\frac{d}{dx^{-1}}
}^G
\circ \ldots \circ
\mathcal{L}_{
\frac{d}{dx^{-n}}
}^G
\circ D_1=0
\end{equation*}
(one of the
$\frac{d}{dx^{-i}}$
factors appears twice).
Now, let us see that
\begin{equation*}
(j^{n+1}s)^\ast
\left(
\iota _X
\left(
\mathcal{L}_{
\frac{d}{dx^{-1}}
}^G
\circ \ldots \circ
\mathcal{L}_{
\frac{d}{dx^{-n}}
}^G
(\partial (\Theta _0^L
-\eta ^G\cdot L))
\right)
\right)
=0.
\end{equation*}
The bidegree of
$\partial $ is $(1,0)$,
so we have
\begin{align*}
\partial (\Theta _0^L
-\eta ^G\cdot L)
&
=\partial
\left(
(-1)^{m-1}
\iota _{
\frac{\partial }
{\partial x^j}
}
\eta ^G\wedge \theta ^\mu
\frac{\partial L}
{\partial y_j^\mu }
\right) \\
&
=\iota _{
\frac{\partial }
{\partial x^j}
}
\eta ^G\wedge
\left(
\partial \theta ^\mu \cdot
\frac{\partial L}
{\partial y_j^\mu }
-\theta ^\mu \wedge \partial
\left(
\frac{\partial L}
{\partial y_j^\mu }
\right)
\right) .
\end{align*}
From Lemma \ref{lemma14},
we know the explicit
expression for $\partial $.
Making use of it, along with
the formulas
\begin{align*}
\mathcal{L}_{
\frac{\partial }
{\partial y^\mu }
}^G
\theta ^\nu
&
=0,\\
\mathcal{L}_{
\frac{\partial }
{\partial y_\alpha ^\mu }
}^G
\theta ^\nu
&
=-(-1)^{
\alpha (\mu +\alpha )
}
d^Gx^\alpha
\delta _\mu ^\nu ,\\
\iota _{
\frac{\partial }
{\partial y^\mu }
}
\theta ^\nu
&
=\delta _\mu ^\nu ,\\
\iota _{
\frac{\partial }
{\partial y_\alpha ^\mu }
}
\theta ^\nu
&
=0,
\end{align*}
we obtain
\begin{multline*}
\partial
\left(
\Theta _0^L
-\eta ^G\cdot L
\right) \\
=\iota _{\frac{\partial }
{\partial x^j}
}
\eta ^G\wedge
\left(
d^Gx^\alpha \wedge
\theta _\alpha ^\mu
\frac{\partial L}
{\partial y_j^\mu }
-d^G\theta ^\mu
\frac{\partial L}
{\partial y_j^\mu}
-\theta ^\nu \wedge
\theta ^\mu
\frac{\partial ^2L}
{\partial y^\mu
\partial y_j^\nu }
-\theta ^\nu \wedge
\theta _\alpha ^\mu
\frac{\partial ^2L}
{\partial y_\alpha ^\mu
\partial y_j^\nu }
\right) ,
\end{multline*}
and remarking that
\begin{equation*}
d^Gx^\alpha \wedge
\theta _\alpha ^\mu
-d^G\theta ^\mu
=-d^Gx^\alpha \wedge
d^Gx^\beta \cdot
y_{\beta \alpha }^\mu ,
\end{equation*}
we deduce
\begin{multline*}
\partial
\left(
\Theta _0^L
-\eta ^G\cdot L
\right) \\
=-\iota _{
\frac{\partial }
{\partial x^j}
}
\eta ^G\wedge
\left(
d^Gx^\alpha \wedge
d^Gx^\beta \cdot
y_{\beta \alpha }^\mu
+\theta ^\nu
\wedge \theta ^\mu
\frac{\partial ^2L}
{\partial y^\mu
\partial y_j^\nu }
+\theta ^\nu \wedge
\theta _\alpha ^\mu
\frac{\partial ^2L}
{\partial y_\alpha ^\mu
\partial y_j^\nu }
\right) .
\end{multline*}
Here, the first term
in the right-hand side
vanishes when a vertical
vector field is inserted.
The other two, when
the pull-back
$(j^{n+1}s)^\ast $ is taken,
as a contact form
$\theta ^\mu $
remains even after the insertion
of the vertical field.
\end{proof}

As a consequence of these results, we can see that the Euler$-$Lagrange
equations for a superfield are those expected.

\begin{theorem}
A local section $s$ of $p$ is a critical section for the Berezinian density $%
\xi _{L}=[d^{G}x^{1}\wedge \cdots \wedge d^{G}x^{m}\otimes \frac{d}{dx^{-1}}%
\circ \ldots \circ \frac{d}{dx^{-n}}]L$ with $L\in \mathcal{A}%
_{J_{G}^{1}(p)},$ if and only if the following equations holds: 
\begin{equation}
(j^{n+1}s)^{\ast } \left( \iota _{X}d^{G}\Theta ^{L} \right) =0,
\label{eq5_19}
\end{equation}
for every vector field $X$ on $J_{G}^{n+1}(p)$ vertical over $(M,\mathcal{A})
$.
\end{theorem}

\begin{proof}
As we have seen, the equation
\eqref{eq5_19} is equivalent to
the Euler-Lagrange equations
\begin{equation*}
(j^{n+1}s)^\ast
\left(
\frac{\partial L}
{\partial y^\mu }
-\frac{d}{dx^i}
\frac{\partial L}
{\partial y_i^\mu }
-(-1)^\mu \frac{d}{dx^{-j}}
\frac{\partial L}
{\partial y_{-j}^\mu }
\right)
=0,
\end{equation*}
and these are the well-known
conditions on $s$ to be
a critical section (see
\cite[Theorem 6.3]{Mon 92a}).
\end{proof}

\section{Some applications}

\subsection{Noether Theorem}

Next, we consider the infinitesimal symmetries of Berezinian densities. The
basic idea is to study under which conditions we can interchange $\iota _{X}$
with $d^{G}$ in \eqref{eq5_19} to obtain the equation 
\begin{equation*}
d^{G}(j^{n+1}s)^{\ast } \left( \iota _{X}\Theta ^{L} \right) =0,
\end{equation*}
giving us an invariant, $\iota _{X}\Theta ^{L}$. In Classical Mechanics,
this is the case when the Lagrangian is invariant under the action of some
group whose infinitesimal generator is precisely $X$; this observation
motivates the following definitions.

A $p-$projectable vector field $X$ on $(N,\mathcal{B})$ is said to be an 
\emph{infinitesimal supersymmetry} of the Berezinian density 
\begin{equation*}
\xi _{L} =\left[ d^{G}x^{1}\wedge \cdots \wedge d^{G}x^{m}\otimes \frac{d}{%
dx^{-1}}\circ \ldots \circ \frac{d}{dx^{-n}} \right] \cdot L,\quad L\in 
\mathcal{A}_{J_{G}^{1}(p)},
\end{equation*}
if 
\begin{equation*}
\mathcal{L}_{X_{(n+1)}}^{G}\xi _{L}=0,
\end{equation*}
where $X_{(n+1)}$ is the $(n+1)-$jet extension of $X$ by graded contact
infinitesimal transformations.

Now, the desired interchange amounts to have $\mathcal{L}_{X_{(n+1)}}^{G}%
\Theta ^{L}=0$. A basic result in this direction is the infinitesimal
functoriality of the Poincar\'e$-$Cartan form, a concept which requires a
previous definition.

A graded vector field $X^{\prime }$ on $(M,\mathcal{A})$ is said to have a 
\emph{graded divergence} with respect to a graded volume $m-$form $\eta ^{G}$
on $(M,\mathcal{A})$ if there exists a function $f\in \mathcal{A}$ such
that, 
\begin{equation*}
\mathcal{L}_{X^{\prime }}^{G}\eta ^{G}=\eta ^{G}f.
\end{equation*}
In this case, we put $f=\mathop{\rm {}div}\nolimits_{G}(X^{\prime })$. A
graded vector field $X$ on $(N,\mathcal{B})$ is said to have graded
divergence if it is $p-$projectable and if its projection $X^{\prime }$ has
graded divergence.

\begin{theorem}[Infinitesimal functoriality of $\Theta ^{L}$, \protect\cite%
{Her-Mun 84b}]
Let $\eta ^{G}\cdot L$ be a graded Lagrangian density on $p\colon (N,%
\mathcal{B})\rightarrow (M,\mathcal{A})$ $(L\in \mathcal{A})$ and $\Theta
^{L}$ the corresponding graded Poincar\'{e}$-$Cartan form. For every vector
field $X$ on $(N,\mathcal{B})$ with divergence, we have 
\begin{equation}
\mathcal{L}_{X_{(n+1)}}^{G}\Theta ^{L}=\Theta ^{L^{\prime }},  \label{eq5_20}
\end{equation}
where $L^{\prime } =X_{(n+1)}(L)+\mathop{\rm {}div}\nolimits_{G}(X^{\prime
})\cdot L$.
\end{theorem}

According to this result, what we want $(\mathcal{L}_{X_{(n+1)}}^{G}\Theta
^{L}=0)$ is equivalent to $\Theta ^{L^{\prime }}=0$, that is, to $L^{\prime
}=0$. Let us see under which conditions this is true for an infinitesimal
supersymmetry. Let us write the Berezinian density as $\xi _{L}=[\xi ]L$ and
assume that $X$ is such a supersymmetry; then 
\begin{equation}
0=\mathcal{L}_{X_{(n+1)}}^{G}\xi _{L} =\left( \mathcal{L}_{X_{(n+1)}}^{G}[%
\xi ]\right) L+(-1)^{|X_{(n+1)}||\xi |}[\xi ]X_{(n+1)}(L).  \label{eq5_21}
\end{equation}
As the Berezinian module plays a r\^{o}le akin to that of the volume forms
(at least with respect to integration), we can use the concept of Berezinian
divergence (see Section \ref{berdiv}). We recall that if $X^{\prime }$ is a
graded vector field on $(M,\mathcal{A})$ and $\xi $ is a Berezinian density
on $(M,\mathcal{A})$, we have $\mathcal{L}_{X^{\prime }}^{G}[\xi ]
=(-1)^{|X^{\prime }||\xi |}[\xi ] \cdot \mathop{\rm {}div}%
\nolimits_{B}(X^{\prime })$.

Note that the graded divergence of a given graded vector field on $(M,%
\mathcal{A})$, 
\begin{equation*}
X^{\prime } =\left( X^{\prime } \right) ^{i} \frac{\partial }{\partial x^{i}}
+\left( X^{\prime } \right) ^{-j} \frac{\partial }{\partial x^{-j}},
\end{equation*}
does not necessarily exist. Indeed, the existence of the graded divergence
requires, 
\begin{equation*}
\frac{\partial \left( X^{\prime } \right) ^{i}} {\partial x^{-j}}=0,
\end{equation*}
for any $i,-j$. On the other hand, the Berezinian divergence always exists.

If $X$ on $(N,\mathcal{B})$ is $p-$projectable, we write 
\begin{equation*}
\mathcal{L}_{X}^{G}[\xi ] =(-1)^{|X||\xi |}[\xi ]\cdot \mathop{\rm {}div}%
\nolimits_{B}(X).
\end{equation*}
This makes sense as long as $X$ is projectable (with projection $X^{\prime }$%
); then, if the Berezinian is given by $[\xi ]=[\eta \otimes P]$ for some
graded form $\eta \in \Omega _{G}((M,\mathcal{A}))$ and some differential
operator $P\in \mathcal{D}(\mathcal{A})$, we extend the previous definition
to 
\begin{align*}
\mathcal{L}_{X}^{G}[\xi ]& =(-1)^{|X||\omega \otimes P|+1} \left[ \eta
\otimes P\circ X^{\prime }\right] \\
& =\mathcal{L}_{X^{\prime }}^{G}[\xi ].
\end{align*}
In other words, the graded Lie derivative of $[\xi ]$ with respect to $X$ is
that respect to its projection. The same observation (and definition)
applies to a graded vector field on $(J_{G}^{k}(p),\mathcal{A}%
_{J_{G}^{k}(p)})$ projectable onto $(M,\mathcal{A})$.

Thus, the equation \eqref{eq5_21} can be rewritten as 
\begin{eqnarray*}
(-1)^{ \left\vert X_{(n+1)} \right\vert |\xi| } \mathop{\rm {}div}
\nolimits_B \left( X_{(n+1)} \right) \cdot L & + & (-1)^{ \left\vert
X_{(n+1)} \right\vert |\xi | } X_{(n+1)}(L) \\
& = & \mathop{\rm {}div}\nolimits_B \left( X_{(n+1)} \right) \cdot L
+X_{(n+1)}(L) \\
& = & 0,
\end{eqnarray*}
and this is the expression of $L^\prime =0$ except for the fact that we have
two different divergences. In this way, we are led to the following result.

\begin{theorem}[Noether]
Assume $X$ is an infinitesimal supersymmetry of the Berezinian density 
\begin{equation*}
\xi _{L} =\left[ d^{G}x^{1}\wedge \ldots \wedge d^{G}x^{m}\otimes \frac{d}{%
dx^{-1}}\circ \ldots \circ \frac{d}{dx^{-n}} \right] \cdot L,\quad L\in 
\mathcal{A}_{J_{G}^{1}(p)},
\end{equation*}
such that,

\begin{enumerate}
\item[$(1)$] The projection $X^\prime $ of $X$ onto $(M,\mathcal{A})$ has a
divergence with respect to 
\begin{equation*}
d^Gx^1\wedge \cdots \wedge d^Gx^m,
\end{equation*}

\item[$(2)$] $\mathop{\rm {}div}\nolimits_B(X^\prime ) =\mathop{\rm {}div}%
\nolimits_G(X^\prime )$.
\end{enumerate}

Then, for every critical section $s$ of $\xi _L$ we have 
\begin{equation*}
d^G \left[ \left( j^{n+1}s \right) ^\ast \left( \iota _{X_{(n+1)}}\Theta ^L
\right) \right] =0.
\end{equation*}
\end{theorem}

\begin{proof}
If $X$ is an infinitesimal
supersymmetry of
$\xi _L$, by \eqref{eq5_21}
we have
\begin{equation*}
\operatorname{div}_B(X_{(n+1)})
\cdot L
+X_{(n+1)}(L)=0
\end{equation*}
and by $(1)$, $(2)$,
$L^\prime
=\operatorname{div}_G
(X_{(n+1)})\cdot L
+X_{(n+1)}(L)=0$.
Moreover, from \eqref{eq5_20},
we have
\begin{equation*}
\Theta ^{L^\prime }
=0
=\mathcal{L}_{X_{(n+1)}}^G
\Theta ^L.
\end{equation*}
Thus,
\begin{equation*}
\left(
j^{n+1}s
\right) ^\ast
\left(
d^G\iota _{X_{(n+1)}}\Theta ^L
\right)
+\left(
j^{n+1}s
\right) ^\ast
\left(
\iota _{X_{(n+1)}}d^G\Theta ^L
\right)
=0,
\end{equation*}
and since $s$ is a critical
section,
\begin{equation*}
\left(
j^{n+1}s
\right) ^\ast
\left(
\iota _{X_{(n+1)}}d^G\Theta ^L
\right)
=0.
\end{equation*}
The statement now follows
from the fact that $d^G$
commutes with pullbacks.
\end{proof}

The superfunctions $\iota _{X_{(n+1)}}\Theta ^{L}$ appearing in the
statement, are called \emph{Noether supercurrents}. Analogously, the graded
vector fields $X$ satisfying the conditions of the theorem (and, in general,
those leading to Noether supercurrents; note that these conditions are
sufficient, but not necessary) are called \emph{Noether supersymmetries}.

\begin{corollary}
Assume $X$ is a $p-$vertical graded vector field which also is an
infinitesimal supersymmetry of the Berezinian density (\ref{berezlocal}).
Then, for every critical section $s$ of $\xi _{L}$ we have 
\begin{equation*}
d^{G}\left[ \left( j^{n+1}s\right) ^{\ast }\left( \iota _{X_{(n+1)}}\Theta
^{L}\right) \right] =0.
\end{equation*}
\end{corollary}

\begin{proof}
If $X$ is vertical, its projection is $0$ and so
$\operatorname{div}_B(X^\prime )
=\operatorname{div}_G(X^\prime )
=0$.
\end{proof}

\subsection{The case of supermechanics}

Consider the supermanifold $\mathbb{R}^{1|1}\doteq (\mathbb{R},\Omega (%
\mathbb{R}))$ and the graded submersion 
\begin{equation*}
p\colon (N,\mathcal{B})\rightarrow \mathbb{R}^{1|1}, \quad (N,\mathcal{B})=%
\mathbb{R}^{1|1}\times \mathbb{R}^{1|1},
\end{equation*}
defined by the projection onto the first factor, which determines the graded
bundle of $1-$jets $(J_{G}^{1}(p),\mathcal{A}_{J_{G}^{1}(p)})$. This is the
situation that would correspond to supermechanics (see \cite{Fre 88, Fre 99,
Mon-Mun 02, Mon-Val 03}). If $(t,s)$ and $(t,s,y,z)$ are supercoordinates
for $\mathbb{R}^{1|1}$ and $(N,\mathcal{B})$, respectively (even $y$ and odd 
$z$), we have a system $(s,t,y,z,y_{t},y_{s},z_{t},z_{s})$ for $%
(J_{G}^{1}(p),\mathcal{A}_{J_{G}^{1}(p)})$. These coordinates are defined
through 
\begin{equation*}
\begin{array}{rlr}
(j^{1}\sigma )^{\ast }t = & \!\!\!\sigma ^{\ast }(t)=t \medskip &  \\ 
(j^{1}\sigma )^{\ast }s = & \!\!\!\sigma ^{\ast }(s)=s \medskip &  \\ 
(j^{1}\sigma )^{\ast }y = & \!\!\!\sigma ^{\ast }(y)=\varphi (t) \medskip & 
\\ 
(j^{1}\sigma )^{\ast }z = & \!\!\!\sigma ^{\ast }(z)=\psi (t)s & 
\end{array}%
\end{equation*}

\begin{equation*}
\begin{array}{rll}
(j^1\sigma )^\ast y_t = & \!\!\! \dfrac{\partial }{\partial t}
(j^1\sigma)^\ast y & \!\!\!\! =\varphi ^\prime (t) \medskip \\ 
(j^1\sigma )^\ast y_s = & \!\!\! \dfrac{\partial } {\partial s} (j^1\sigma
)^\ast y & \!\!\!\! =0 \medskip \\ 
(j^1\sigma )^\ast z_t = & \!\!\! \dfrac{\partial}{\partial t} (j^1\sigma
)^\ast z & \!\!\!\! =\psi ^\prime (t)s \medskip \\ 
(j^1\sigma )^\ast z_s = & \!\!\! \dfrac{\partial }{\partial s} (j^1\sigma
)^\ast z & \!\!\!\! =\psi (t)%
\end{array}%
\end{equation*}
for a section $\sigma \colon \mathbb{R}^{1|1} \to (N,\mathcal{B})$ of $p$.
Here, $\varphi $ and $\psi $ are just real functions. Note the particularity
of the coordinate $y_s$, which evaluated on sections vanishes; this is a
special feature of the $(1|1)$-dimension.

The traditional (physics oriented) notation would write $\partial y/\partial
t$ instead of $\varphi ^{\prime }(t)$ and so on. In this way, the preceding
observation about $y_{s}$ is masked, so we prefer ours.

A graded Lagrangian is an element $L\in \mathcal{A}_{J_{G}^{1}(p)}$ (i.e., a
`superfunction\ of $(s,t,y,z,y_{t},y_{s},z_{t},z_{s})$'). We are interested
in determining the class of Lagrangians which admit a $p-$projectable graded
vector field on $(N,\mathcal{B})$, of the particular form 
\begin{equation*}
D=f\frac{\partial }{\partial t}+g\frac{\partial }{\partial s},
\end{equation*}
as a Noether supersymmetry.

A priori, we should have $f=f(t,s)$ and $g=g(t,s)$, but the fact that $D$
must be a supersymmetry imposes some restrictions which we now analyze.
First of all, $\mathop{\rm {}div}\nolimits_{G}(D)$ must exist, and this
forces $f=f(t)$; hence 
\begin{equation}
\frac{\partial f}{\partial s}=0.  \label{examp1}
\end{equation}
Moreover, it is immediate from the definition of graded divergence that 
\begin{equation}
\mathop{\rm {}div}\nolimits_{G}(D)=\frac{\partial f}{\partial t}
\label{examp2}
\end{equation}
Secondly, $\mathop{\rm {}div}\nolimits_{B}(D)$ must coincide with $%
\mathop{\rm {}div}\nolimits_{G}(D)$; from the expression \eqref{a4} in
Section \ref{berdiv}, we find that, necessarily, 
\begin{equation}
\frac{\partial g}{\partial s}=0.  \label{examp3}
\end{equation}

With the restrictions \eqref{examp1}, \eqref{examp3}, the computation of the
extension $D_{(2)}$ is relatively easy, and the result is 
\begin{equation}
\left. 
\begin{array}{r}
D_{(2)}=f\dfrac{\partial }{\partial t} +g\dfrac{\partial }{\partial s}
-\left( \dfrac{df}{dt}y_t+\dfrac{dg}{dt}y_s \right) \dfrac{\partial }{%
\partial y_t} -\left( \dfrac{df}{dt}z_t +\dfrac{dg}{dt}z_s \right) \dfrac{%
\partial }{\partial z_t} \medskip \\ 
-\dfrac{df}{dt}y_{st} \dfrac{\partial }{\partial y_{st}} -\left( \dfrac{d^2f%
}{dt^2}y_t +\dfrac{d^2g}{dt^2}y_s +2\dfrac{df}{dt}y_{tt} +2\dfrac{dg}{dt}%
y_{ts} \right) \dfrac{\partial }{\partial y_{tt}} \medskip \\ 
-\dfrac{df}{dt}z_{st} \dfrac{\partial }{\partial z_{st}} -\left( \dfrac{d^2f%
}{dt^2}z_t +\dfrac{d^2g}{dt^2}z_s +2\dfrac{df}{dt}z_{tt} +2\dfrac{dg}{dt}%
z_{ts} \right) \dfrac{\partial }{\partial z_{tt}}%
\end{array}
\right\}  \label{examp4}
\end{equation}

Finally, the remaining condition for $D$ to be a Noether supersymmetry is $%
\mathcal{L}_{D_{(2)}}^G \xi _L=0$; that is, 
\begin{equation*}
\mathop{\rm {}div}\nolimits_B(D) \cdot L+D_{(2)}L=0,
\end{equation*}
or, in view of \eqref{examp2}, 
\begin{equation}  \label{examp5}
\frac{\partial f}{\partial t}L+D_{(2)}L=0.
\end{equation}
As $L\in \mathcal{A}_{J_G^1(p)}$, we have 
\begin{equation*}
\frac{\partial L} {\partial y_{tt}} =\frac{\partial L} {\partial y_{st}} =0
\end{equation*}
and 
\begin{equation*}
\frac{\partial L}{\partial z_{tt}} =\frac{\partial L}{\partial z_{st}} =0,
\end{equation*}
so the insertion of \eqref{examp4} into \eqref{examp5} gives 
\begin{equation*}
\frac{\partial f}{\partial t}L +f\frac{\partial L}{\partial t} +g\frac{%
\partial L} {\partial s} -\left( \frac{df}{dt}y_t +\frac{dg}{dt}y_s \right) 
\frac{\partial L} {\partial y_t} -\left( \frac{df}{dt}z_t +\frac{dg}{dt}z_s
\right) \frac{\partial L}{\partial z_t} =0.
\end{equation*}
Now, evaluating on a section $\sigma $ we obtain 
\begin{equation}
\left. 
\begin{array}{rc}
\dfrac{\partial f}{\partial t} (j^1\sigma )^\ast L +f(j_G^1\sigma )^\ast
\left( \dfrac{\partial L}{\partial t} \right) +g(j^1\sigma )^\ast \left( 
\dfrac{\partial L}{\partial s} \right) \medskip &  \\ 
-\left( \dfrac{df}{dt}\varphi ^\prime (t) \right) (j^1\sigma )^\ast \left( 
\dfrac{\partial L}{\partial y_t} \right) \medskip &  \\ 
-\left( \dfrac{df}{dt} \psi ^\prime (t)s +\dfrac{dg}{dt}\psi (t) \right)
(j^1\sigma )^\ast \left( \dfrac{\partial L}{\partial z_t} \right) & =0%
\end{array}
\right\}  \label{examp6}
\end{equation}

Any $L\in \mathcal{A}_{J_G^1(p)}$ solution to this equation, is a
superlagrangian admitting $D$ as a Noether supersymmetry. Conversely, if we
take a fixed $L\in \mathcal{A}_{J_G^1(p)}$, any pair of real functions, $%
f=f(t,s)$ and $g=g(t,s)$, satisfying \eqref{examp6} determines a graded
vector field $D=f\partial /\partial t +g\partial /\partial s$, which is a
Noether supersymmetry for $L$.

A trivial case is that of $f=g\equiv 1$. Then, the equation \eqref{examp6}
reduces to 
\begin{equation*}
\frac{\partial L}{\partial t} +\frac{\partial L}{\partial s}=D(L)=0,
\end{equation*}
that is: \emph{if} $L$ \emph{does not depend explicitly on} $(t,s)$, \emph{%
then the \textquotedblleft supertime translation\textquotedblright } $%
D=\partial /\partial t+\partial /\partial s$ \emph{is a Noether supersymmetry%
}, as in the classical setting (see \cite{Mon-Val 03}).

\paragraph{Acknowledgements.\/}

The junior author (JAV) wants to express his gratitude to the Instituto de F%
\'{\i}sica Aplicada (CSIC; Madrid, Spain) and the Centro de Investigaci\'{o}%
n en Matem\'{a}ticas (CIMAT; Guanajuato, M\'{e}xico) for their warm
hospitality during his stay there, and where part of this work was done.
Also, thanks are due to Adolfo S\'{a}nchez Valenzuela, for numerous and
helpful comments on preliminary drafts of this paper.

Supported by the \textsc{Ministerio de Educaci\'on y Ciencia} of Spain under
grants PB97--1386 and MTM2005-00173.


\begin{thebibliography}{GMS 97}
\bibitem{Ald-Azc 78a} V. Aldaya, J. A. de Azc\'{a}rraga: \emph{Variational
principles on rth order jets of fibre bundles in field theory}. J. Math.\
Phys.\ \textbf{19} (1978) 9, 1869--1875.

\bibitem{Ald-Azc 78b} ---: \emph{Vector bundles, rth order Noether
invariants and canonical symmetries in Lagrangian field theory}. J. Math.\
Phys.\ \textbf{19} (1978)\ 7, 1876--1880.

\bibitem{Bar-Bru-Her 91} C.\ Bartocci, U. Bruzzo, D. Hern\'{a}ndez Ruip\'{e}%
rez: \emph{The geometry of supermanifolds}. Mathematics and its Applications 
\textbf{71}. Kluwer Academic Publishers Group, Dordrecht (The Nederlands),
1991.

\bibitem{Bat 79} M. Batchelor: \emph{The structure of supermanifolds}.
Trans.\ Amer.\ Math.\ Soc.\ \textbf{253} (1979), 329--338.

\bibitem{Bat 80} ---: \emph{Two approaches to supermanifolds}. Trans.\
Amer.\ Math.\ Soc.\ \textbf{258} (1980), 257--270.

\bibitem{Ber 66} F. A. Berezin: \emph{The method of second quantization}
(Russian). Izdat. Nauka, Moscow, 1965. English translation: Pure and Applied
Physics\ \textbf{24}. Academic Press, New York-London, 1966.

\bibitem{Ber 87} ---: \emph{Introduction to superanalysis}. Mathematical
Physics and Applied Mathematics \textbf{9}. D. Reidel Publishing Co.
Dordrecht (The Netherlands), 1987.

\bibitem{Boy-San 91} C. P. Boyer, A. S\'{a}nchez-Valenzuela: \emph{Lie
supergroup actions on Supermanifolds}. Trans.\ Amer.\ Math.\ Soc.\ \textbf{%
323} (1991)\ 1, 151--175.

\bibitem{Cor-Neem-Stern 75} L. Corwin, Y. Ne'eman, S. Sternberg: \emph{%
Graded Lie algebras in mathematics and physics (Bose-Fermi symmetry)}. Rev.\
Modern Phys.\ \textbf{47} (1975), 573--603.

\bibitem{Del-Morg 99} P. Deligne, J. W. Morgan: \emph{Notes on Supersymmetry
(following Joseph Bernstein), in Quantum Fields and Strings: A Course for
Mathematicians (Volume 1, Part 1: Classical Fields and Supersymmetry)}.
American\ Mathematical\ Society. Providence (RI), 1999.

\bibitem{Dubois 01} M. Dubois-Violette: \emph{Lectures on graded
differential algebras and noncommutative geometry}. Noncommutative
differential geometry and its applications to physics (Shonan, 1999), pp.\
245--306. In Math.\ Phys.\ Stud. \textbf{23}. Kluwer Academic Publishers
Group, Dordrecht (The Netherlands), 2001.

\bibitem{Fer-Fra 87} M. Ferraris, M. Francaviglia: \emph{Applications of the
Poincar\'{e}$-$Cartan form in higher order field theories}. Differential
Geometry and its Applications (Brno, 1986), pp.\ 31--52. In Math.\ Appl.\
(East European Ser.) \textbf{27}. Reidel, Dordrecht (The Netherlands), 1987.

\bibitem{Fre 99} D. Freed: \emph{Five lectures on supersymmetry}. American
Mathematical Society. Providence (RI), 1999

\bibitem{Fre 88} P. G. O. Freund: \emph{Introduction to supersymmetry}.
Cambridge University Press, Cambridge, 1998.

\bibitem{Gar-Mun 83} P. L. Garc\'{\i}a, J. Mu\~{n}oz-Masqu\'{e}: \emph{On
the geometrical structure of higher order variational calculus}. Proceedings
of the IUTAM-ISIMM Symposium on Modern Developments in Analytical Mechanics,
Volume I (Torino, 1982). In Atti Accad.\ Sci.\ Torino Cl.\ Sci.\ Fis.\ Mat.\
Natur.\ \textbf{117} (1983) suppl.\ 1, 127--147.

\bibitem{Gaw 77} K. Gawdezki: \emph{Supersymmetries: mathematics of
supergeometry}: Ann.\ Inst.\ Henri Poincar\'{e} XXVII (1977) \textbf{4},
335--366.

\bibitem{GMS 97} G. Giachetta, L. Mangiarotti, G. Sardanashvily: 
\emph{New Lagrangian and Hamiltonian methods in field theory}. World
Scientific Publishing Co. Inc. River Edge (NJ) 1997.

\bibitem{GMS 98} ---: 
\emph{Gauge mechanics}. World Scientific Publishing Co. Inc. River Edge (NJ)
1998.

\bibitem{GMS 05} ---: 
\emph{Lagrangian supersymmetries depending on derivatives. Global analysis
and cohomology}. Comm. Math. Phys. \textbf{259}, no. 1 (2005) 103--128.

\bibitem{God 69} C. Godbillon: \emph{G\'{e}om\'{e}trie diff\'{e}rentielle et
m\'{e}canique analytique}. Hermann, Paris, 1969.

\bibitem{Her-Mun 82} D. Hern\'{a}ndez Ruip\'{e}rez, J. Mu\~{n}oz Masqu\'{e}: 
\emph{Higher order jet bundles for graded manifolds (Superjets)}.
Proceedings of the International Meeting on Geometry and Physics (Florence,
1982), pp.\ 271--278. Pitagora Editrice, Bologna, 1983.

\bibitem{Her-Mun 84a} ---: \emph{Global variational calculus on graded
manifolds I}. J. Math.\ Pures Appl.\ \textbf{63} (1984), 283--309.

\bibitem{Her-Mun 84b} ---: \emph{Infinitesimal functoriality of graded
Poincar\'{e}$-$Cartan forms}. Differential Geometric Methods in Theoretical
Physics (Shumen, 1984), 126--132. World Sci.\ Publishing, Singapore, 1986.

\bibitem{Her-Mun 85a} ---: \emph{Global variational calculus on graded
manifolds II}. J. Math.\ Pures Appl.\ \textbf{63} (1985), 87--104.

\bibitem{Her-Mun 85b} ---: \emph{Construction intrins\`eque du faisceau de
Berezin d'une vari\'{e}t\'{e} gradu\'{e}e}. Comptes Rendus Acad.\ Sci.\
Paris, S\'{e}rie I Math. \textbf{301} (1985), 915--918.

\bibitem{Her-Mun 87} ---: \emph{Variational Berezinian problems and their
relationship with graded variational problems}. Differential Geometric
Methods in Mathematical Physics (Salamanca, 1985), pp.\ 137--149. Lecture
Notes in Math.\ \textbf{1251}. Springer, Berlin, 1987.

\bibitem{Kos 77} B. Kostant: \emph{Graded manifolds, graded Lie theory, and
prequantization}. Differential Geometrical Methods in Mathematical Physics
(Bonn, 1975), pp.\ 177--306. Lecture Notes in Math.\ \textbf{570}. Springer,
Berlin, 1977.

\bibitem{Lei 80} D. A. Leites: \emph{Introduction to the theory of
supermanifolds}. Russ.\ Math.\ Surveys \textbf{35} (1980), 1--64.

\bibitem{Man 88} Y. I. Manin: \emph{Gauge field theory and complex geometry.}
Grundlehren der Mathematischen Wissenschaften [Fundamental Principles of
Mathematical Sciences] \textbf{289}. Springer-Verlag, Berlin Heidelberg,
1988.

\bibitem{Mon 92a} J. Monterde: \emph{Higher order Graded and Berezinian
Lagrangian densities and their Euler$-$Lagrange equations}. Ann.\ Inst.\ H.\
Poincar\'{e} Phys.\ Th\'{e}or.\ \textbf{57} (1992)\ 1, 3--26.

\bibitem{Mon-Mun 92a} J. Monterde, J. Mu\~{n}oz Masqu\'{e}: \emph{%
Variational problems on graded manifolds}. Contemp.\ Math.\ \textbf{132}
(1992), 551--571.

\bibitem{Mon-Mun 92b} ---: \emph{Hamiltonian Formalism for Berezinian
Variational Problems in Supermanifolds}. Group Theoretical Methods in
Physics (Salamanca, 1992), pp.\ 253--256. Anales de F\'{\i}sica, Monograf%
\'{\i}as 2. CIEMAT-Real Sociedad Espa\~{n}ola de F\'{\i}sica, Madrid
(Spain), 1993.

\bibitem{Mon-Mun 02} ---: \emph{Hamiltonian formalism in supermechanics}.
Internat.\ J. Theoret.\ Phys.\ \textbf{41} (2002)\ 3, 429--458.

\bibitem{Mon-San 93} J. Monterde, O.A. S\'{a}nchez-Valenzuela: \emph{%
Calculus of variations in a simple superdomain without the Berezinian
integral}. XXVIth National Congress of the Mexican Mathematical Society
(Spanish) (Morelia, 1993), pp.\ 313--318. Aportaciones Mat.\ Comun., 14.
Soc.\ Mat.\ Mexicana, M\'{e}xico, 1994.

\bibitem{Mon-Val 03} J. Monterde and J. A. Vallejo: \emph{The symplectic
structure of Euler$-$Lagrange superequations and Batalin-Vilkovisky
formalism.} J. Phys.\ A: Math.\ Gen.\ \textbf{36} (2003), 4993--5009.

\bibitem{Ner 93} A.~N.~Nersessian: \emph{On The Geometry of Supermanifolds
with even and odd Kahlerian Structures}. Theor.\ Math.\ Phys.\ \textbf{96}
(1993), 866--871 [Teor.\ Mat.\ Fiz.\ \textbf{96} (1993), 140--149].

\bibitem{San 86} O. A. S\'{a}nchez-Valenzuela: \emph{On Supervector Bundles.}
Comunicaciones T\'{e}cni\-cas IIMAS-IMUNAM (serie naranja) \textbf{457},
1986.

\bibitem{San 88} ---: \emph{Matrix computations in linear superalgebra}.
Linear Algebra and its applications \textbf{111} (1988), 151--181.

\bibitem{San 88a} ---: \emph{Linear supergroup actions.\ I:\ On the defining
properties}. Trans.\ Amer.\ Math.\ Soc.\ \textbf{307} (1988)\ 2, 569--595.

\bibitem{San 88b} ---: \emph{Remarks on Grassmannian supermanifolds}.
Trans.\ Amer.\ Math.\ Soc.\ \textbf{307} (1988), 2, 597--614.

\bibitem{Sau 89} D. J. Saunders: \emph{The geometry of jet bundles}. London
Mathematical Society Lecture Note Series \textbf{142}. Cambridge University
Press, Cambridge (UK), 1989.

\bibitem{Sch 84} Th.\ Schmitt: \emph{Super differential geometry}. Report
MATH, 84--5. Akademie der Wissenschaften der DDR (Institut f\"{u}r
Mathematik), Berlin, 1984.

\bibitem{Var 04} V. S. Varadarajan: \emph{Supersymmetry for mathematicians:
an introduction}. American Mathematical Society-Courant Institute of
Mathematical Sciences.\ Lecture Notes CLN \textbf{11}. New York (NY) and
Providence (RI), 2004.

\bibitem{Vor 91} T. Voronov: \emph{Geometric integration theory on
supermanifolds}. Soviet Scientific Reviews. Section C: Mathematical Physics
Reviews \textbf{9}, Part 1.\ Harwood Academic Publishers, Chur., 1991.

\bibitem{Wess-Bag 92} J. Wess, J. Bagger: \emph{Supersymmetry and
supergravity} (Second edition). Princeton Series in Physics.\ Princeton
University Press, Princeton (NJ), 1992.
\end{thebibliography}
\end{document}